\documentclass[oneside,reqno,11pt,dvipsnames]{amsart}
\usepackage[a4paper,left=3cm,right=3cm,top=3cm,bottom=3cm]{geometry}

\usepackage{xcolor}                                        
\usepackage[T1]{fontenc}                                   
\usepackage[utf8]{inputenc}                                
\usepackage{microtype}                                     
\usepackage{preamble/symbols}                              
\usepackage{preamble/environments}                         
\usepackage{enumitem}                                      
\usepackage{subcaption}                                    
\usepackage{booktabs}                                      
\usepackage{hyperref}                                      
\definecolor{colorlinks}{RGB}{0, 24, 168}                  
\definecolor{colorcites}{RGB}{124, 10, 2}                  
\hypersetup{                                               
    colorlinks=true,                                       
    linkcolor=colorlinks,                                  
    citecolor=colorcites,                                  
    urlcolor=colorlinks,                                   
    pdfborder={0 0 0}                                      
}                                                          

\usepackage{preamble/boxes}                                
\renewcommand\boxstatus[1]{} 
\PassOptionsToPackage{disable}{todonotes} 
\usepackage{preamble/todos}                                
\usepackage{preamble/headings}                             

\usepackage{bbm}

\newcommand\commentTiancheng[1]{\comment{Tiancheng: #1}}
\newcommand\commentHugo[1]{\comment{Hugo: #1}}
\newcommand\commentIoan[1]{\comment{Ioan: #1}}
\renewcommand\im[1]{\commentIoan{#1}}

\newcommand\commentEmile[1]{\comment{Emile: #1}}
\newcommand\commentPiet[1]{\comment{Piet: #1}}

\newcommand\eav[1]{\commentEmile{#1}}
\newcommand\hugo[1]{\commentHugo{#1}}
\newcommand\hide[1]{}
\newcommand\ind{\mathbbm{1}}

\newcommand\even\circ
\newcommand\odd\bullet
\renewcommand\P{\mathbb{P}}

\newcommand\Circuit{\mathrm{Circuit}}
\newcommand\Vertical{\operatorname{Vertical}}
\newcommand\Horizontal{\operatorname{Horizontal}}

\newcommand\osc{\operatorname{oscillation}}

\renewcommand\a{\mathbf{a}}
\renewcommand\b{\mathbf{b}}
\renewcommand\c{\mathbf{c}}

\newcommand\cmulti{c_{\mathrm{multi}}}

\newcommand\cinter{c_{\mathrm{inter}}}
\newcommand\charm{c_{\mathrm{harm}}}

\newcommand\eps{\epsilon}

\makeatletter
\@namedef{subjclassname@2020}{\textup{2020} Mathematics Subject Classification}
\makeatother
\title{Conformal invariance of the six-vertex model's height function}
\author[Emile Averous, Hugo Duminil-Copin, Tiancheng He, Piet Lammers and Ioan Manolescu]{Emile Averous\textsuperscript{1}, Hugo Duminil-Copin\textsuperscript{1,2}, Tiancheng He\textsuperscript{1}, \\ Piet Lammers\textsuperscript{3} and Ioan Manolescu\textsuperscript{4}}
\thanks{\textsuperscript{1}Universit\'e de Gen\`eve,
\texttt{emile.averous@unige.ch}, \texttt{hugo.duminil@unige.ch},
\texttt{tiancheng.he@unige.ch}.}
\thanks{\textsuperscript{2}Institut des Hautes \'Etudes Scientifiques,
\texttt{duminil@ihes.fr}.}
\thanks{\textsuperscript{3}CNRS, Sorbonne Universit\'e, LPSM,
\texttt{piet.lammers@cnrs.fr}.}
\thanks{\textsuperscript{4}Universit\'e de Fribourg,
\texttt{ioan.manolescu@unifr.ch}.}
\date{5 October 2026}
\begin{document}



\maketitle
\vfill
\begin{abstract}
    In this paper, we show that the height function of the six-vertex model in
    simply connected domains with piecewise constant boundary conditions
    converges to the Gaussian Free Field (GFF) with Dirichlet boundary
    conditions. This implies the conformal invariance of the scaling limit of the
    model. 
    
    The argument relies on the full plane GFF convergence combined with
    Reflexion Positivity and qualitative estimates. The proof works for a wide
    range of spectral parameters $\Delta$ (among which  $\Delta$ close to $0$, and
    between $-1$ and $-1/2$) making this paper the first instance of a conformal
    invariance result for a continuum of models. The paper has consequences for
    other models, including the FK percolation with cluster-weight $q=4$ and the
    Ashkin-Teller model, which we will pursue in future work.
\end{abstract}
\vfill
\setcounter{tocdepth}{1}
\tableofcontents

\newpage



\section{Introduction}

\subsection{Motivation}

For more than fifty years, it has been conjectured that the large-scale behaviour of lattice models undergoing a continuous phase transition is described by quantum field theories that are invariant under {\em conformal transformations}; see, for instance, the early works of Patashinskii and Pokrovskii \cite{PatPok64}, followed by those of Polyakov \cite{Pol68, Pol70a,Pol70b}. Such quantum field theories are known as Conformal Field Theories (CFTs), and their enhanced symmetry makes them particularly amenable to analysis. In two dimensions, the conformal symmetry is infinite-dimensional, and this abundance of local symmetries renders many planar CFTs integrable, as revealed in the seminal work of Belavin, Polyakov, and Zamolodchikov \cite{BPZ84a,BPZ84b}. This observation led to a remarkably detailed understanding of two-dimensional CFTs and, in turn, to a wealth of predictions for the critical behaviour of lattice models.

The physical predictions associated with these models suggest a path toward a comprehensive mathematical understanding of phase transitions in two-dimensional lattice models. The program consists of two steps: first, to study the corresponding CFTs and extract, at the continuum level, all the relevant information; second, to prove that the original lattice models are indeed conformally invariant in the scaling limit.

The first step has been particularly successful. It is fair to say that the mathematical understanding of CFT is now well underway, with spectacular progress having been made in recent decades. The introduction of the Schramm-Loewner Evolution \cite{Sch00,Law05} and the Conformal Loop Ensembles \cite{She09,SheWer12} (see also \cite{MS16,MS20}), as well as advances in the mathematical foundations of CFT \cite{Dav+16,RV14,GRV19} have reshaped the field of 2D statistical physics over the past thirty years.

By contrast, every advance toward the second step has been the result of a hard-fought battle. Proofs of conformal invariance remain confined to a handful of specific models. These include site Bernoulli percolation on the triangular lattice \cite{Smi01,KhrSmi21,CamNew07,CamNew06}, with results ranging from Cardy's formula to convergence toward SLE(6) and CLE(6); the Ising and FK-Ising models \cite{Smi10,CheSmi12,HonSmi13,CheHonIzy15,CheDumHon12a,BenHon16}, where major advances include the analysis of fermionic observables, the energy and spin fields, and convergence to SLE and CLE with parameters 3 and 16/3; uniform spanning trees \cite{LawSchWer04}; dimer models \cite{Ken00,Ken01}; and level lines of the discrete Gaussian free field \cite{SchShe09}.

The present paper aims to change this picture. We provide the first proof of conformal invariance for a broad class of models, namely the height function of the six-vertex model with spectral parameter $\Delta$ ranging over whole intervals of possible values. In addition to this, this result opens the way toward establishing conformal invariance for other important families of lattice models, including the Ashkin-Teller model and the Potts model.

The proof follows a strategy pioneered by Kenyon in his work on dimers \cite{Ken00,Ken01}: one identifies suitable observables associated with the height function and shows that they solve a discrete analogue of an appropriate boundary value problem. In \cite{Ken00}, correlations of the height gradient were shown to be discrete holomorphic and to satisfy suitable boundary conditions. Then, the scaling limit is identified as the solution to the corresponding boundary value problem in the continuum. The same strategy has subsequently been successfully implemented in \cite{Smi01,Smi10,CheSmi12,LawSchWer04,SchShe09} for other models. In the present setting, we prove that multi-point correlations of the height function of the six-vertex model are approximately pre-harmonic and we determine their behaviour near the boundary. This enables us to pass to the limit as the lattice mesh tends to zero and thereby identify the scaling limit of the multi-point correlations of the height function.

 The proof of approximate harmonicity relies on a new argument based on the reflection positivity (RP) of the model. More precisely, RP implies analyticity properties and a Cauchy-Schwarz inequality that can be used  to transfer the approximate harmonicity of full-plane correlations established in \cite{DumKozLamMan26} to correlations in simply connected domains. In this sense, our result builds crucially on \cite{DumKozLamMan26} and on the body of work leading up to it.

 Interestingly, the paper fits the framework imagined by physicists consisting in first proving translation, rotation and scale invariance before deriving conformal invariance. Indeed, translation invariance was trivial, rotation invariance was obtained in \cite{AveCheDumHeKraManXia26,DumKozKraManOul26}, the scaling limit in the full plane (which happens to be scale invariant) was derived in \cite{DumKozLamMan26}, and finally conformal invariance is obtained in this paper.  For the height function of the six-vertex model, the present work can therefore be viewed as the final step in this program.

Let us conclude by mentioning that the approximate harmonicity result is particularly robust and, as discussed at the end of this introduction, is expected to have applications beyond the results established here. From a broader perspective, this paper provides a stepping stone toward establishing conformal invariance for other lattice models.

\subsection{Definition of the six-vertex model on discrete graphs}

The six-vertex model is a probability measure on orientations of the edges of a graph \cite{Bax82,LieWu80}. A six-vertex configuration $\omega^\times$ on a subgraph of the square lattice assigns an orientation to each edge, subject to the \emph{ice rule}: at every vertex of degree four, the number of incoming edges equals the number of outgoing edges. In particular, there are precisely six local arrow configurations compatible with this constraint, as depicted in Figure \ref{fig:6V_vertices}. The six local configurations are grouped into three pairs, to which we assign positive weights $\mathbf a$, $\mathbf b$, and $\mathbf c$, referred to below as {\em Boltzmann weights}. We also introduce the \emph{spectral parameter}
\begin{equation}
\label{eq}
\Delta=\Delta(\mathbf a,\mathbf b,\mathbf c)
:=
\frac{\mathbf a^2+\mathbf b^2-\mathbf c^2}
{2\mathbf a\mathbf b}.
\end{equation}

\begin{figure}
\centering
\includegraphics[width=0.8\linewidth]{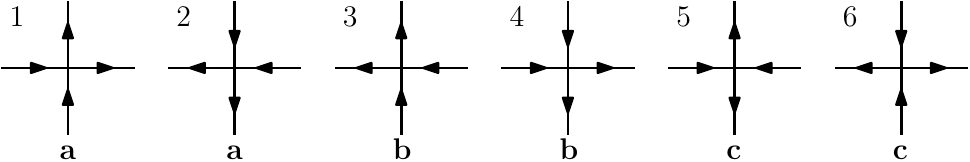}
\caption{The six possible local arrow configurations around a degree-four vertex of the six-vertex model, together with their corresponding weights.}
\label{fig:6V_vertices}
\end{figure}

\bigbreak\noindent{\bf Configurations in finite domains.}
While the six-vertex model is often considered with periodic boundary conditions, in this paper we work in finite domains. Let $F$ be a collection of faces of $\mathbb Z^2$ which is connected (here two faces are neighbours if they share a side) and whose complement is connected. We associate with $F$ the graph $G=(V,E)$ whose edge-set $E$ is composed of edges separating two neighbouring faces of $F$, and whose vertex-set $V$ is made of the endpoints of these edges. The resulting graph has vertices of degree one, two, or four. A six-vertex configuration on $G$ is an orientation of edges in $E$ satisfying the ice rule at every vertex of $V$ of  degree four.

From now on, we only consider a graph $G$ obtained from a collection of faces $F$ of the square lattice\footnote{Or as we will see below, of a rescaled version of the square lattice.} which is connected and whose complement is connected, as explained in the previous paragraph. We call this a {\em simply connected discrete domain}. 

Note that the set of faces of $G$ is not $F$. More precisely, except for the exterior face of $G$, the other faces are indeed in $F$, but not all faces in $F$ are faces in $G$. We denote by $\partial G$ the set of faces of $F$ that are not faces in $G$. It coincides with the collection of faces in $F$ that are adjacent to at least one face of the square lattice outside $F$.

\bigbreak\noindent{\bf The height function of the six-vertex model.}
 Every six-vertex configuration $\omega^\times$ on $G$ naturally gives rise to an integer-valued \emph{height function} $h$ defined on $F$. It is piecewise constant on each face\footnote{Strictly speaking, this may assign two different values to an edge of $G$, and possibly as many as four different values to a vertex of $G$. We disregard this minor ambiguity, since it has no effect on the results below. It could, for instance, be resolved by fixing in advance any deterministic rule for selecting a value at these exceptional points.} and its increments between neighbouring faces are determined by the orientation of the edge separating them: when crossing an edge of $\omega^\times$, the height increases by $1$ if the crossed arrow points to the right and decreases by $1$ otherwise.

\begin{figure}
\begin{center}
\includegraphics[width = .35\textwidth, page = 7]{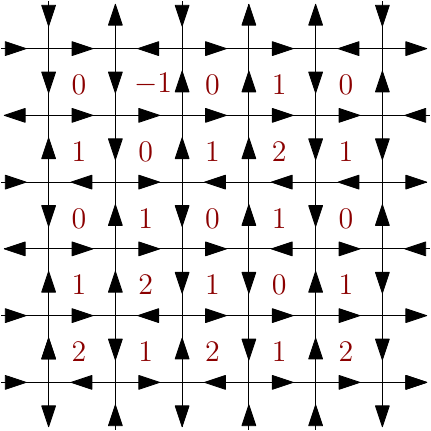}
\caption{A six-vertex configuration $\omega^\times$ and its associated height function.}
\label{fig:6Vheights}
\end{center}
\end{figure}

The ice rule ensures that these local increments are compatible around every vertex and can therefore be integrated to define a height function on $F$. The height function is determined only up to a global additive constant, whereas its gradient is uniquely determined by $\omega^\times$. Conversely, the arrow configuration can be recovered from the height increments between neighbouring faces. Thus, modulo an additive constant, the height function and the six-vertex configuration encode the same information.

\bigbreak\noindent{\bf Six-vertex measure with Dirichlet boundary conditions.}
Let $\zeta$ be a height function prescribed on $\partial G$. The six-vertex measure on $G$ with \emph{Dirichlet boundary condition $\zeta$} is defined by
\begin{equation}
\label{eq:6V_def}
\bbP^\zeta_{G,\mathbf a,\mathbf b,\mathbf c}[\omega^\times]
:=
\frac{1}{Z^\zeta_{G,\mathbf a,\mathbf b,\mathbf c}}
\mathbf a^{N_1+N_2}
\mathbf b^{N_3+N_4}
\mathbf c^{N_5+N_6}
\mathbbm{1}_{h_{|\partial G}=\zeta},
\end{equation}
where $N_i=N_i(\omega^\times)$ denotes the number of degree-four vertices whose local arrow configuration is of type $i$. The normalising constant $Z^\zeta_{G,\mathbf a,\mathbf b,\mathbf c}$ is the corresponding \emph{partition function}, chosen so that \eqref{eq:6V_def} defines a probability measure.

We emphasise that only degree-four vertices contribute to the Boltzmann weight. At such vertices, the ice rule restricts the local configuration to one of the six types displayed in Figure \ref{fig:6V_vertices}.

\subsection{Basics on the Gaussian Free Field}

We now turn to the description of the object that will ultimately arise as the scaling limit.
There are several equivalent ways to view the GFF, depending on the regularity one wishes to impose on the field.
For further background on the GFF, we refer to \cite{She07,BerPow25}.
Let $\calD\subset\mathbb R^2$ be a \emph{domain}, that is, a non-empty connected open set.
For $n\geq 1$, let
\begin{equation}
    D_n(\calD)
    :=
    \bigl\{(u_1,\dots,u_n)\in\calD^n:\ u_i\neq u_j
    \text{ for all }i\neq j\bigr\}.
\end{equation}

We begin with the GFF with zero boundary conditions.
Let
$ G_{\calD}:D_2(\calD)\longrightarrow\mathbb R
$
denote the Dirichlet Green function in $\calD$.
Thus, for each fixed $y\in\calD$, the function $x\mapsto G_{\calD}(x,y)$ is harmonic on $\calD\setminus\{y\}$, vanishes as $x$ approaches $\partial\calD$, and satisfies
\begin{equation}
    G_{\calD}(x,y)
    =
    -\frac{1}{2\pi}\log|x-y|+O(1)
    \qquad\text{as }x\to y.
\end{equation}
Equivalently, $G_{\calD}$ is symmetric and satisfies these properties in either coordinate.

By a \emph{finite-Dirichlet-energy generalised test function} we mean a finite signed measure $\varphi$ supported in $\calD$ such that
\begin{equation}
    \iint_{\calD\times\calD}
    G_{\calD}(u,v)\,{\rm d}\varphi(u)\,{\rm d}\varphi(v)
    <\infty.
\end{equation}
The GFF is a random distribution, in the sense that it can be tested against a finite-Dirichlet-energy generalised test function $\varphi$. We set
\begin{equation}
   \langle\Gamma,\varphi\rangle:= \int_{\calD}
   \Gamma(u){\rm d}\varphi(u).
\end{equation}
\begin{definition}[GFF with zero boundary conditions] The zero-boundary GFF $\Gamma^\calD$ can then be characterised as follows:
for every finite family
$
   \varphi=(\varphi_1,\ldots,\varphi_n)
$ of finite-Dirichlet-energy generalised test functions, define the $n\times n$ matrix
$\Sigma(\varphi)$ by
\begin{equation}
    \Sigma(\varphi)_{ij}
    :=
    \iint_{\calD\times\calD}
    G_{\calD}(u,v)\,
    {\rm d}\varphi_i(u)\,
    {\rm d}\varphi_j(v).
\end{equation}
Then
\begin{equation}
    \bigl(\langle\Gamma^\calD,\varphi_i\rangle\bigr)_{i=1}^n
    \sim
    \calN\bigl(0,\Sigma(\varphi)\bigr).
\end{equation}
This specifies the law of $\Gamma$ as a centred Gaussian generalised function with covariance kernel $G_{\calD}$.
\end{definition}

Away from the diagonals, the multipoint correlation functions of the zero-boundary GFF are given by Wick's formula. Namely,
\begin{equation}
    \label{eq:GFF_multi-point_cor_function}
    \Psi_n:D_n(\calD)\longrightarrow\mathbb R,
    \qquad
    (u_1,\dots,u_n)
    \longmapsto
    \sum_{\pi}
    \prod_{\{i,j\}\in\pi}
    G_{\calD}(u_i,u_j),
\end{equation}
where the sum runs over all pairings $\pi$ of $\{1,\dots,n\}$, that is, over all partitions of $\{1,\dots,n\}$ into sets of cardinality two.
In particular,
$ \Psi_n\equiv 0$
    whenever $n$ is odd,
since in that case no such pairing exists.

More generally, the GFF  in $\calD$  with boundary conditions $\zeta$ and variance $\sigma^2$, denoted by $\Gamma^{\calD,\zeta,\sigma^2}$, is obtained by setting
\begin{equation}
    \Gamma^{\calD,\zeta,\sigma^2}:=\sigma\Gamma^{\calD}+H_\zeta,
\end{equation}
where $H_\zeta$ denotes the harmonic extension of $\zeta$ to $\calD$.
Here and throughout, we implicitly assume that $\zeta$ admits a harmonic extension to $\calD$ that extends continuously to $\partial\calD$ with boundary values $\zeta$. Note that $\langle \Gamma^{\calD,\zeta,\sigma^2},\varphi\rangle$ is a multivariate Gaussian random variable with averages given by $\langle H_\zeta,\varphi\rangle:=(\langle H_\zeta,\varphi_i\rangle)_{1\le i\le n}$ and variance $\sigma^2\Sigma(\varphi)$.

\begin{proposition}[Conformal invariance of the GFF]\label{prop:conformal invariance GFF}
The family of Gaussian free fields with Dirichlet boundary conditions is conformally invariant in the following sense. Let
$f:\calD\to\calD'$ be a conformal bijection, and let $\zeta$ be a continuous function on $\partial\calD$. Then the push-forward under $f$ of the GFF in $\calD$ with boundary conditions $\zeta$ has the same law as the GFF in $\calD'$ with boundary conditions $\zeta\circ f^{-1}$.
\end{proposition}

\subsection{Main results of the paper}

This article applies to a certain range of triples $(\mathbf a,\mathbf b,\mathbf c)$, subject to several restrictions. The first, and most immediate, was already alluded to above. When $\Delta<-1$, the height function of the six-vertex model is localised, as proved in \cite{Dum21,RS22,GP23}, whereas it is expected to be delocalised for $\Delta\in[-1,1)$ (this is proved in some special cases \cite{Cha21,GP23,Lis21,DumKarManOul22,GL25}). We therefore restrict our attention to triples $(\mathbf a,\mathbf b,\mathbf c)$ for which $\Delta\in[-1,1)$.

As explained later, we require that $\mathbf a=\mathbf b$ as we are currently unable to derive the proper RSW theorems for anisotropic models. We do expect that this independent question can be addressed, which would enable us to generalise our result. In this paper, we nevertheless restrict our attention to the case $\mathbf a=\mathbf b$ and assume without loss of generality that they are equal to $1$. Below, we index the model by the parameter $\Delta$, which is in one-to-one correspondence with $\mathbf c$, the only remaining free parameter once $\mathbf a$ and $\mathbf b$ are fixed at 1. We say that the model is {\em isotropic}.

A third restriction on the range of parameters arises from our use of positive association for the so-called spin representation of the six-vertex model (see below). This requires us to assume that $\mathbf c\ge1$, i.e.~$\Delta\le 1/2$.

The fourth restriction is more delicate. We require that the scaling limit of the full-plane six-vertex height function be a Gaussian Free Field (GFF). More precisely, let $\mathbb E_{\delta\bbZ^2}$ denote the unique flip-symmetric Gibbs measure of the six-vertex model (see \cite{DumKozLamMan26} for the definition). Under this measure, a six-vertex height function $h$ can be defined up to an additive constant. We say that {\em the scaling limit of the six-vertex height function (of spectral parameter $\Delta$) is the GFF} if, for every collection of distinct points $u_1,u'_1,\dots,u_n,u'_n\in\mathbb R^2$,
\[
\mathbb E_{\delta\bbZ^2}\left[(h(u_1)-h(u'_1))\cdots(h(u_n)-h(u'_n))\right]
\]
converges, as $\delta$ tends to zero, to the corresponding multi-point correlations of the full-plane GFF with variance
\[
\sigma^2=\frac{2}{\arccos\Delta}.
\]
Such a convergence result has been established in \cite{DumKozLamMan26} for every $\Delta\in[-1,-1/2]$, and in \cite{GiuMasTon17}\footnote{In \cite{GiuMasTon17}, convergence to the full-plane GFF is established, but the variance of the limiting GFF is not identified explicitly. Nevertheless, since the limiting correlations are rotationally invariant, independently of the value of $\sigma^2$, one may use the main result of \cite{DumKozLamMan26}  to identify the variance.} for $\Delta\in[-\Delta_0,\Delta_0]$, where $\Delta_0>0$ is a sufficiently small constant.

It is noted in \cite{DumKozLamMan26} that the range of parameters for which full-plane GFF convergence could be established is potentially considerably larger, essentially the only missing ingredient being a rotational-invariance result analogous to that of \cite{AveCheDumHeKraManXia26}. In view of this, we state the main results of this paper in the following general setting.

Let
\begin{equation}\label{eq:definition I}
I_{\mathbf{CI}}=\{\Delta\le 1/2:\text{ the full-plane 6V height function's scaling limit is the GFF}\}.
\end{equation}
As explained above, $(-1,-1/2]\cup[-\Delta_0,\Delta_0]\subset I_{\mathbf {CI}}$ (where $\Delta_0$ is as in \cite{GiuMasTon17}).

\subsubsection{Harmonicity of the multi-point correlations of the isotropic six-vertex model}

The first result, which represents the core of our paper, states that multi-point correlations are approximately pre-harmonic (in the sense that the circle averages around a point are close to the value at the point) in each coordinate.

To state our result, introduce, for a reasonable function $F$ on $\mathbb R^2$, $z\in \mathbb R^2$, and $r>0$, the Laplacian
\begin{equation}
\Delta_rF(z):=F(z)-\oint_{\partial B_r(z)} F(v){\rm d}v,
\end{equation}
where $\oint$ means the average on the circle. Except if otherwise stated, we always think of $\Delta_r$ as acting on the function following it with $z$ as the variable. Recall that we are consider the model with $\mathbf a=\mathbf b=1$.

\begin{theorem}[scaling limit of the isotropic six-vertex model]\label{thm:harmonicity}
For every $K,R,r,\eps,n>0$, there exists $\delta_0=\delta_0(K,R,r,\eps,n)>0$ such that for every $\Delta\in I_{\mathbf {CI}}$, the following holds true. For every $\delta\in(0,\delta_0)$, every simply connected subgraph $G$ of $\delta\bbZ^2$ with $\Lambda_R \not\subset G$, every $n$ faces $z,u_2,\dots,u_{n}\in G \cap \Lambda_R$ at a distance at least 1 from $\partial G$, and a distance between $1$ and $K$ from each other, and any boundary conditions $\zeta$ on $\partial G$ satisfying $\sup\{|\zeta(u)|:u\in\partial G\}\le K$, we have
\begin{equation}
\big|\Delta_r \mathbb E_{G}^\zeta[h(z)h(u_2)\dots h(u_{n})]\big|\le \eps.\label{eq:harmonicity}
\end{equation}
 \end{theorem}

Note that the notion of ``approximate harmonicity'' used here does not live at the scale $\delta$ of the original graph, but rather at the macroscopic scale. This differs from exact pre-harmonicity or pre-holomorphicity obtained in \cite{Ken00,Ken01,Smi10,CheSmi12}.

The approximate harmonicity result is quite robust. More generally, it can be used to establish harmonicity (in the variable $z$) of certain observables whenever one can explore an annulus around $z$ and ``discover'' a path of heights with uniformly bounded variation. In particular, the argument should extend to cylindrical and toroidal geometries, as well as to more general boundary conditions, provided that these do not substantially alter the qualitative behaviour of the model.

\subsubsection{Carath\'eodory convergence}

From now on, the results involve certain approximations (on the lattice $\delta\mathbb Z^2$) of a domain $\calD$. We therefore discuss this in more detail now.

Our aim is to approximate a simply connected domain $\calD\subset\mathbb R^2$, possibly together with  marked boundary points
$
a^{(1)},\dots,a^{(k)}\in\partial\calD$, listed in clockwise order (below, consider the case of no marked point as $k=0$). Although our results apply in a more general setting, we restrict ourselves to simply connected domains in order to avoid unnecessary technicalities.

Our proof is compatible with a wide variety of discrete approximations, but we adopt the following setup for definiteness. We consider {\em simply connected discrete domains} defined as follows. They are subgraphs $\calD_\delta$ of $\delta\mathbb Z^2$ obtained as follows: consider an {\em edge}-self-avoiding non-self-crossing polygon, which we denote $\partial\calD_\delta$. To simplify the discussion below, we consider that when the polygon turns, the turns are infinitesimally roundish. 
Then, $\calD_\delta$ is the set of medial faces enclosed by $\partial\calD_\delta$.

Let $\calD_\delta$ be simply connected discrete domains, viewed as unions of faces of $\delta\bbZ^2$, and let
$
a^{(1)}_\delta,\dots,a^{(k)}_\delta
$
be marked boundary vertices (of the graph associated with this set of faces, see the definitions above) that lie on the boundary of $\calD_\delta$, seen as a domain of $\mathbb R^2$.

We say that
$
(\calD_\delta;a^{(1)}_\delta,\dots,a^{(k)}_\delta)
$ converges to
$(\calD;a^{(1)},\dots,a^{(k)})
$
in the {\em Carath\'eodory sense} if there exist conformal bijections
\begin{equation}
f_\delta:\calD_\delta\longrightarrow\mathbb D
\qquad\text{and}\qquad
f:\calD\longrightarrow\mathbb D
\end{equation}
such that every compact subset of $\calD$ is contained in $\calD_\delta$
for all sufficiently small $\delta$, $f_\delta$ converges to $f$
uniformly on every compact subset of $\calD$, and such that
\begin{equation}
f_\delta(a^{(j)}_\delta)\longrightarrow f(a^{(j)})
\qquad\text{for every }j=1,\dots,k.
\end{equation}
Here, the conformal maps are understood to be continuously extended to the marked boundary points (which are considered as prime ends in the case of $\calD$).

\subsubsection{Convergence of the height function in piecewise even/odd domains}

Recall that Theorem~\ref{thm:harmonicity} implies that $\mathbb E_{G}^\zeta[h(z)]$ is approximately harmonic on large graphs. Consequently, any subsequential limit, along a sequence for which the functions $\mathbb E_{\calD_\delta}^{\zeta_\delta}[h_\delta(\cdot)]$ converge uniformly on compact sets, is harmonic. In particular, for suitable choices of $\calD_\delta$ and $\zeta_\delta$, if one can establish convergence of the expected height function near the boundary, then one can obtain convergence of the six-vertex height function throughout the domain. This can be achieved in certain cases. We mention one here, which is important both for the six-vertex model and for related models, such as the Ashkin-Teller model.

 Fix a chessboard colouring, in black and white, of the vertices of $\delta\mathbb Z^2$. A simply connected discrete domain $\calD_\delta$ is called {\em piecewise even/odd} if there exist medial edges  $a^{(0)}_\delta,\dots,a^{(k)}_\delta=a^{(0)}_\delta$ partitioning $\partial\calD_\delta$ into arcs $(a^{(i)}_\delta a^{(i+1)}_\delta)$, such that every face of $\mathcal D_\delta$ adjacent to the arc $(a_\delta^{(i)} a_\delta^{(i+1)})$ is black when $i$ is even and white when $i$ is odd.

We say that boundary conditions $\zeta_\delta$ on $\calD_\delta$ are {\em piecewise constant} if $\zeta_\delta$ is constant on the faces bordering each boundary arc $(a^{(i)}_\delta a^{(i+1)}_\delta)$, for $1\leq i\leq k$. We denote by $\zeta_\delta(i)$ the value of $\zeta_\delta$ on the arc $(a^{(i)}_\delta  a^{(i+1)}_\delta)$. Notice that the boundary values on two consecutive arcs differ by $\pm1$.

We are now in a position to state the following theorem, which states the convergence of $h_\delta$  in this context.

\begin{theorem}[Scaling limit in even domains]\label{thm:even}
Let $\Delta\in I_{\mathbf {CI}}$.  Let $\calD$ be a simply connected domain and let $\zeta$ be a piecewise constant boundary function whose jumps, occurring at marked boundary points $a^{(1)},\dots,a^{(k)}$, are all equal to $\pm1$. Let $\zeta(i)$ be the value on $(a^{(i)},a^{(i+1)})$.

Consider piecewise even/odd  simply connected discrete domains $(\calD_\delta,a^{(1)}_\delta,\dots,a^{(k)}_\delta)$ on $\delta\bbZ^2$ converging to
$(\calD,a^{(1)},\dots,a^{(k)})$
in the Carath\'eodory sense, and $\zeta_\delta$ piecewise constant boundary conditions $\zeta(i)$ on faces adjacent to the arcs $(a^{(i)}_\delta a^{(i+1)}_\delta)$.

Then, $h_\delta$  on $\calD_\delta$ converges, as $\delta$ tends to zero, to the GFF on $\calD$ with boundary conditions $\zeta$ and variance
\begin{align}\label{eq:sigma_expression}
		\sigma^2
        =  \frac{2}{\arccos \Delta}.
	\end{align}
	More precisely, for any $n\geq 1$, the $n$-point correlations of $h_\delta$ converge, as $\delta$ tends to zero, uniformly on compact subsets of $D_n(\calD)$ to  the $n$-point correlations of the corresponding GFF.
            Furthermore, for any family $\varphi=(\varphi_1,\ldots,\varphi_n)$ of finite Dirichlet energy generalised test functions,
            the law of $(\langle  h_\delta,\varphi_i\rangle)_i$ converges weakly to $\calN((\langle H_\zeta,\varphi_i\rangle)_i,\sigma^2\Sigma(\varphi))$ as $\delta$ tends to zero.%
\end{theorem}

For the six-vertex model on a piecewise even/odd domain $\calD_\delta$, every vertex of $\partial\calD_\delta$ which is not one of the $a_\delta^{(i)}$ has degree two. Such vertices therefore carry no non-trivial Boltzmann weight. Even domains with piecewise constant boundary conditions arise naturally when exploring the so-called percolation representation of the six-vertex model, which we introduce in Section~\ref{sec:review_percolation_representation}.
It is also natural to ask what happens when one explores level lines of the six-vertex model, or interfaces arising from FK percolation. In these settings, the boundary conditions are constant, but non-trivial weights may be associated with the degree-two boundary vertices.

\subsection{Discussion}

\subsubsection{Generalisations of our result} One may wonder why we restrict ourselves to piecewise even/odd simply connected discrete  domains and to the associated piecewise boundary conditions. Several generalisations are conceivable.

Simple connectedness is not essential per se. The main point that would need to be adapted is the notion of Carath\'eodory convergence, which can be extended to this setting without substantial difficulty. One can also extend the harmonicity result and the GFF convergence to the cylinder, see \cite{DumKozLamMan26b} for details. A generalisation to the torus also appears within reach, but it raises the problem of estimating the leading eigenvalues of the first blocks of the so-called transfer matrix and therefore relies on delicate results from the Bethe Ansatz \cite{DumKozKraManTik22}.

The case $\mathbf a\ne \mathbf b$ is probably reachable under the constraint that $\mathbf a,\mathbf b<\mathbf c$.  In this case, one still benefits from an FKG inequality for the absolute value of the height function, which is a core ingredient in our proof. The additional difficulty compared to the isotropic case $\mathbf a=\mathbf b=1$ comes from the fact that we do not currently know the relevant RSW estimates. Still, these estimates seem within reach. Assuming such qualitative estimates, one can derive the conformal invariance once the lattice $\mathbb Z^2$ is embedded in such a way that each square becomes a rhombus with a given angle $\alpha$ depending on $\mathbf a,\mathbf b,\mathbf c$; see \cite{DumKozLamMan26} for more details.

As for the boundary conditions, Theorem~\ref{thm:harmonicity} shows that multi-point correlations are harmonic in each coordinate, provided that the boundary conditions remain bounded. The main difficulty in deriving a convergence result for general such boundary conditions is twofold. 

First, one needs to understand the behaviour of the mean height function, that is, the one-point correlation. This can be achieved for certain specific classes of boundary conditions. For instance, one may allow several of the points $a_\delta^{(i)}$ to converge to the same point in the scaling limit, while keeping $k$ fixed. Another possibility is to consider domains whose boundaries consist of horizontal and vertical segments, with boundary conditions alternating between two values along boundary arcs of macroscopic length.

Second, one would need to estimate the boundary conditions of multi-point correlations, i.e.~the behaviour of the correlations when one of the points approaches the boundary.  In order to prove this result, one should extend the boundary mixing result Proposition~\ref{prop:boundary_values_of_h}. While this can probably be done, this requires a new argument.

Notice that the preceding two points are independent. Even when the limit of the expected height cannot be identified, once boundary mixing has been established, one may still centre the height function $h_\delta$ by subtracting its expectation, namely by setting $
\widetilde h_\delta(x)
=
h_\delta(x)-\mathbb E^\zeta_{\calD_\delta}[h_\delta(x)].
$
The centred height function $\widetilde h_\delta$ then has the GFF in $\calD$ with zero boundary conditions as its scaling limit.

\subsubsection{Conformal invariance of other lattice models}

The conformal invariance of the six-vertex height function established above has direct applications to several other lattice models.

First, through the Baxter-Kelland-Wu (BKW) correspondence \cite{BKW76} and the coupling between $\mathrm{CLE}(4)$ and the GFF, our result implies conformal invariance of the critical random-cluster model with cluster weight $q=4$. Previous results for this model include the derivation of its critical exponents \cite{Che26a}, as well as the computation of certain full-plane correlation functions \cite{CheXia26}. Combined with the Edwards-Sokal coupling, this further yields conformal invariance of the critical four-state Potts model on the square lattice. The details of these two conformal invariance results are given in \cite{AveDumHeLamMan26}.

The BKW correspondence extends more generally to random-cluster models with cluster weight $q<4$. In this setting, however, the boundary behaviour induced by conditioning on a large loop of the random-cluster model are considerably more involved. Identifying the scaling limit of the corresponding field therefore remains a substantial challenge. Nevertheless, the results of the present paper constitute a substantial step toward a better understanding of the scaling limit of critical random-cluster models with $1\le q<4$; see \cite{Dum17a} for more details.

The critical four-state Potts model on the square lattice is itself a particular instance of the Ashkin-Teller model. In \cite{DumLisQia25,DumLisQia21}, the scaling limit of the double Ising model, another special case of the Ashkin-Teller model, was obtained through its coupling with the dimer height function, or equivalently with the height function of the six-vertex model at spectral parameter $\Delta=0$. Our convergence result for six-vertex models with $\Delta\in[-1,-1/2]\cup[-\Delta_0,\Delta_0]$ provides a natural route toward establishing conformal invariance of the Ashkin--Teller model over a broader range of parameters. We will pursue this direction in a forthcoming work \cite{AveDumLamMan26}.

 \section{Proof overview}

We start by explaining the main steps and ingredients in the proof of Theorem~\ref{thm:harmonicity}.

\begin{figure}
\begin{center}
\includegraphics[width = 0.4\textwidth]{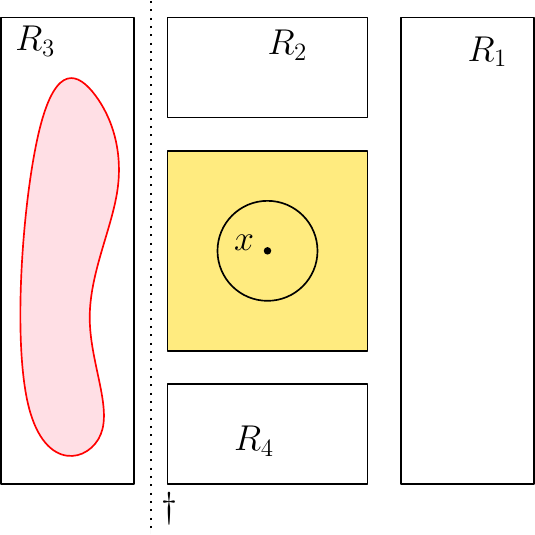}
\caption{The four rectangles $R_1, \dots, R_4$ around $\Lambda_{2r}$. 
The conditioning in Lemma~\ref{lem:R_1} is on $h_4$ on some set as the red one, contained in $R_1$. In Theorem~\ref{thm:intermediate_statement}, we similarly condition on $h_4$ in four connected sets, each contained in one rectangle $R_i$. }
\label{fig:4rectangles2}
\end{center}
\end{figure}

Recall the definition of the local observable $\Delta_r h(x)$. 
We fix henceforth $x$ and $r$. For simplicity we take $x = 0$.
%

A central point will be to condition the full plane measure $\mu_{\delta \bbZ^2}$ on realisations of the gradient in certain regions around $x$, 
and to prove harmonicity under this conditional measure. 
For precision, we will condition on the height function mod $4$; 
write $h_4:F(\delta\Z^2)\to\Z/4\Z,\,f\mapsto h(f)+4\Z$ for the height function modulo four.
Observe that, if $A$ is a connected set of faces, the information of $h_4$ on $A$ is equivalent to the gradient of $h$ on this set, together with a single bit of information that determines the rest mod $4$ of $h$ on a single face (which itself is odd or even). 

For a connected  set $A$, write $h_4|_A$ for the restriction of $h_4$ to $A$ and
\begin{align}
\osc ( h_4|_A)  = \max_{z \in A} h(z) - \min_{z \in A} h(z). 
\end{align}
Note that, since $h_4|_A$ determines the gradient of $h$ on $A$, both its maximum and minimum are defined up to common additive constant. 

\subsection{Harmonicity after conditioning}
\label{subsec:harmonicity_after_conditioning}

%
%

Let us start by a simplified setting which is already very instructive. Consider a  reflection line $\dagger$ across a vertical line which goes through centres of faces. 
\begin{lemma}\label{lem:R_1}
	For any $\eta > 0$, $M\geq 1$ and $\eps > 0$, there exists $\delta_0=\delta_0(\eta,M,\eps)>0$ such that,
	for any $\delta <\delta_0$, any $x$ at distance at least $\eta$ on the left of a vertical line which goes through centres of faces, and any connected set of faces $A$ of diameter at least $r$ 
	at a distance at least $\eta$ on the right of this vertical line, and any realisation $\zeta$ of $h_4$ on $A$ with $\osc (\zeta) \leq M$,
	\begin{align}
	\mu_{\delta \bbZ^2} \big[ \Delta_r h(x) \,\big| h_4|_{A} = \zeta\,\big] \leq \epsilon.
	\end{align}
\end{lemma}


\begin{proof}[Proof sketch for Lemma~\ref{lem:R_1}] Let $A\subset F(\Z^2)$ denote a finite connected set of faces to the left of the
reflection line, and let $\zeta$ be a realisation of $h_4$ on $A$.
We assume that the event $\{h_4|_{A}=\zeta\}$ has positive probability.
Fix $x\in\R^2$ and $r\geq 1$ such that the faces determining $\Delta_rh(x)$ lie to the
right of the reflection line.
See Figure~\ref{fig:4rectangles2} for an illustration; the reflection line is dotted. 

Reflection positivity (see~\cite[Theorems~2.2 and~4.15]{DumKozLamMan26}
or Proposition~\ref{lemma:spectral_representation} below)
yields:
\begin{multline}
    \textstyle\mu_{\Z^2}\left[
        \ind_{h_4|_A=\zeta}
        \cdot
        \Delta_rh(x)
    \right]^2
    \textstyle
    \leq
    \mu_{\Z^2}\left[
        \{h_4|_A=\zeta\}\cap\{h_4|_{A^\dagger}=\zeta^\dagger\}
    \right]\cdot
    \mu_{\Z^2}\left[
        \Delta_rh^\dagger(x)
        \Delta_rh(x)
    \right],\label{eq:ref_pos1}
\end{multline}
In terms of conditional expectations, this may be rewritten as
\begin{multline}
    {\textstyle
        \mu_{\Z^2}\left[
            \Delta_rh(x)
            \middle|
            \{h_4|_A=\zeta\}
    \right]^2}
    \leq
    \frac{
        \mu_{\Z^2}[
            \{h_4|_A=\zeta\}\cap\{h_4|_{A^\dagger}=\zeta^\dagger\}
        ]
    }{
        \mu_{\Z^2}[
        \{h_4|_A=\zeta\}]^2
    }
    \cdot
    \mu_{\Z^2}[
        \Delta_rh^\dagger(x)
        \Delta_rh(x)
  ].
\end{multline}
The last factor is of order $o(1)$ as $r$ tends to infinity by the full-plane Gaussian Free Field convergence of \cite{DumKozLamMan26}.
Thus, once we can uniformly bound the ratio of probabilities, it follows that
the \emph{conditional expectation} of the random variable tends to zero as $\delta$ tends to zero.
The uniform bound on the ratio is called \emph{quasi-independence}, and is derived in Lemma~\ref{lemma:quasi_independence}. 
\end{proof}
The proof of Lemma~\ref{lemma:quasi_independence} is based on Russo--Seymour--Welsh-type percolation arguments,
which also play an important role in other parts of the proof.
The quasi-independence statement is one of the
technical ingredients
in the proof of Theorem~\ref{thm:harmonicity}.

Lemma~\ref{lem:R_1} provides some form of asymptotic harmonicity
of the height expectation in the full-plane measure $\mu_{\Z^2}$ conditioned
on a ``boundary-conditioning-type'' event.
The main difficulty in the proof of Theorem~\ref{thm:harmonicity} comes down to
extending from boundary conditions \emph{on one side of a reflection plane}
to boundary conditions \emph{on a full circuit enclosing $\partial B_r(x)$}.
This will require another argument, which is outlined in the next section.

\subsection{Harmonicity via analyticity}

Certain regions of the plane will be of particular importance, such as the rectangles $R_1,\dots, R_4$ of Figure~\ref{fig:4rectangles2}.
These depend on a parameter $\eta > 0$ that will be fixed later.  
Formally, they are
\begin{align}
R_1 &:= [2r + \eta, 4r - \eta] \times [-4r,4r],\\
R_2 &:= [-2r +\eta , 2r - \eta] \times [2r + \eta,4r] \\
R_3 &:= [ -4r + \eta, -2r - \eta] \times [-4r,4r],\\
R_4 &:= [-2r +\eta , 2r - \eta] \times [-4r,-2r - \eta]. 
\end{align}
\im{Maybe call these $R_N$, $R_S$, $R_E$ and $R_W$. }

As mentioned above, in the proof Theorem~\ref{thm:harmonicity} we would rather like to condition on $h_4$ 
in all rectangles $R_1,\dots, R_4$, rather than in just one of them. 
The equivalent result is listed below. 

\begin{theorem}[Harmonicity under conditioning]\label{thm:intermediate_statement}
	For any $\eta > 0$, $M\geq 1$, $0<r<1/2$, and $\eps > 0$ there exists $\delta_0 = \delta_0(\eta,M,r,\eps)> 0$ such that,
	for any $\delta <\delta_0$, any connected sets of faces $A_i\subset R_i$ for $i = 1,\dots, 4$, each of diameter at least $r$,  
	and any realisations $\zeta_i$ of $h_4$ on $A_i$ with $\osc (\zeta_i) \leq M$ for each $i$, we have
	\begin{align}\label{eq:h1}
	\mu_{\delta \bbZ^2} \big[ \Delta_r h(x) \,\big|\, \{ h_4|_{A_1} = \zeta_1\}\cap\dots\cap\{ h_4|_{A_4} = \zeta_4\} \big] \leq \epsilon.
	\end{align}
\end{theorem}

Observe that the same reasoning as in Lemma~\ref{lem:R_1} may not be applied here. 
Indeed, when trying to separated $R_1$ from $x$ and $R_2,\dots, R_4$ by a vertical line to apply \eqref{eq:ref_pos1},
the second observable multiplied by its conjugate become more complicated than the one we wish to bound.
Therefore, a second crucial idea from \cite{DumKozLamMan26} needs to be used: the analyticity of certain observable with respect to horizontal, vertical and diagonal pushes. 

\hide{
\begin{theorem}[Harmonicity after exploration]
    \label{thm:intermediate_statement}
    Assume GFF convergence for our choice of six-vertex parameters $\a=\b=1$ and $\c\in[1,2]$.

    We may find a family of constants $\cinter(M,\eta,r)>0$
    indexed by $(M,\eta,r)\in\Z_{\geq 1}\times(0,1/100)\times\R_{\geq 1}$
    such that $\lim_{r\to\infty}\cinter(M,\eta,r)=0$ for every fixed $(M,\eta)$,
    and
    such that the following statement holds.
    Fix $(M,\eta,r)$ with $r\geq 5000M/\eta$, and fix $x\in\R^2$.
    Let $\gamma_1,\ldots,\gamma_4\subset F(\Z^2)$ denote four connected
    sets, with $\gamma_i\subset F(R_i)$ for the respective rectangles $R_i$ in
    Figure~\ref{fig:4rectangles2}, and which each have a diameter of at least $r$.
    Let $\zeta_i:\gamma_i\to\Z$ denote four suitable height functions with $\osc(\zeta_i)\leq M$,
    such that the joint conditioning event below has positive probability.
    Then
    \begin{equation}\label{eq:intermediate_statement}
        \textstyle
        \left|\mu_{\Z^2}\left[
            \Delta_rh(x)
            \middle|
            \bigcap_{i=1}^4\{h_4|_{\gamma_i}=\zeta_i\}
        \right]\right|
        \leq \cinter(M,\eta,r).
    \end{equation}
    Here $h_4|_{\gamma_i}=\zeta_i$ means that the functions are equal modulo four.
\end{theorem}
}

\begin{proof}[Proof sketch for Theorem~\ref{thm:intermediate_statement}]
    Fix $M,\eta,r,\eps$ as in the statement. 
    For $i\in\{1,2,3,4\}$, set
	$ E_i :=\{ h_4|_{A_i} = \zeta_i \}$
	and for $z\in\mathbb R^2$, let $E_i(z)$ denote the translation of $E_i$ by $z$.  
  For fixed $r$,   we can shift the connected set $A_1$ to the right by $2\delta\lfloor ar/\delta\rfloor$ for some $a>0$,
    and analyse the behaviour of the functions
    \begin{equation}
        f_\delta(a):=
        \frac{
            \mu_{\delta\bbZ^2}\left[
                \Delta_r h(x)
                \ind_{E_1((2\delta\lfloor ar/\delta\rfloor,0))
               \cap E_2\cap E_3\cap E_4 }           \right]
        }
        {
            \mu_{\delta\bbZ^2}\left[
               E_1
            \right]\cdots  \mu_{\delta\bbZ^2}\left[
               E_4
            \right]
        }.
    \end{equation}
    It turns out that the transfer matrix formalism (Proposition~\ref{lemma:spectral_representation})
    yields a particular (spectral) representation of the function $f_\delta$.
   Among other things, any subsequential limit $f(a)$ of $f_\delta(a)$ as $\delta$ tends to zero is
    analytic in $a>0$.
    In particular, if $f(a)$ is zero for large enough $a$,
   it is zero for every $a>0$. It therefore suffices to prove that $f_\delta(a)$ tends to zero as $\delta$ tends to zero when $a$ is very large.

    The same procedure can be performed repeatedly. Combining \emph{several} translations in
    horizontal, vertical, and diagonal directions, we may arrive at a situation where the translates of $A_1,\dots,A_4$ are all below an horizontal line, and $\Lambda_{2r}(x)$ above it.
    See Figure~\ref{fig:sq2}. Doing so, we may apply Lemma~\ref{lem:R_1} (all events depend on faces lying one side of a reflection line), thus concluding that the limit as $\delta$ tends to zero of the relevant function is zero, and therefore that it is also zero without the different translations, i.e.~that $f_\delta(0)$ tends to zero as $\delta$ tends to zero.
    
    The result follows by another use of quasi-independence (Lemma~\ref{lemma:quasi_independence}) to show that the product of the probabilities of the events $E_i$ is of the order of the probability of the intersection of the events $E_i$. As a consequence, $f_\delta(0)$ and the left-hand side of \eqref{eq:h1} can be compared up to uniform constants. 
\end{proof}

\begin{remark}[Quasi-independence requires spatial separation of the conditioning]
    For quasi-independence to hold, it is important that the sets $(A_i)_i$ are
    geometrically separated.
    Indeed, to apply the Russo--Seymour--Welsh-type arguments, we need some ``space''
    to show that the mixing between the different conditioning events occurs.
    This is why we cannot simply ``close'' the circuit, and obtain the harmonicity
    theorem immediately, and that we work with the H shape instead.
\end{remark}

\begin{remark}\label{remark:diagonal_transfer_matrix}The proof strategy is reminiscent of the one used in~\cite[Section~7]{DumKozLamMan26} to prove
    harmonicity of multi-point correlation functions,
    although it is complicated by the lack of \emph{a priori} rotational invariance of
    the observables. As a consequence, the pushing in diagonal directions plays an important role.
    In fact, we view the construction of a transfer matrix in the diagonal direction
    as one of the innovations of this article.
  Transfer matrices in the horizontal and vertical directions are well established in
    the literature,
    and there is no special restriction on the model's parameters.
   For a transfer matrix in the diagonal direction (with
    a probabilistic interpretation),
    a certain monotonicity property is required.  Consequently, this formalism only makes sense for the parameters $\a=\b=1$ and $\c\geq 1$. The same condition is also at the core of the
    percolation representation
    used in~\cite[Section~14]{DumKozLamMan26}, in earlier work, and in the present article.
    \commentPiet{Supply the historical references for this
        monotonicity/percolation representation (the papers mentioned in the original
        papers by Marcin and by Sasha with Piet).}
   
\end{remark}

\subsection{From infinite volume to finite volume: strong mixing}

Theorem~\ref{thm:intermediate_statement} will be complemented by a strong mixing result. It will yield the following exploration procedure. 
We will explore part of the configuration in $R_1,\dots, R_4$ in such a way that, 
when the exploration is successful,
\begin{itemize}
	\item[(a)] Theorem~\ref{thm:intermediate_statement} applies and yields a small expectation of $\Delta_r h(x)$ 
	under the infinite volume measure conditioned on the exploration and
	\item[(b)] the exploration acts in an almost Markovian way: conditionally on it, the law of the configuration in $\Lambda_{2r}(x)$ is almost 
	independent of that outside $\Lambda_{4r}(x)$. We call this property {\em strong mixing}; see Proposition~\ref{prop:strong_mixing}.
\end{itemize}
The exploration will also be proved to be successful  with high probability under $\mu_{\delta\bbZ^2}$. 
The above will ultimately lead to the following result.

\begin{theorem}\label{thm:explored_harmo}
	For any $\eps > 0$ and $r > 0$, there exists $\delta_0=\delta_0(\eps,r)> 0$ 
	such that, for all $0 < \delta < \delta_0$ the following holds. 
	On $\delta \bbZ^2$, there exists an exploration $\calE$ of a connected set of faces in $\Lambda_{4r}(x) \setminus \Lambda_{2r}(x)$ that separates $\Lambda_{2r}$ from infinity and a notion of a ``successful outcome'' for $\calE$ such that 
	\begin{align}\label{eq:explored_harmo}
		\Big|	\mu_{\delta \bbZ^2} \big[ \Delta_rh(x) \big| \calE \text{ that is successful} \big]  \Big|< \eps
	\end{align}
	and 
	\begin{align}\label{eq:explored_good}
		\mu_{\delta \bbZ^2} \big[ \text{$\calE$ is successful}\big]  > 1- \eps.
	\end{align}
\end{theorem}

We stated the theorem in a someone loose way since the exploration $\calE$ will depend on $h$, but also on a percolation representation $\omega^+$ that has not yet been defined. Its essential feature is that $\calE$ is measurable in terms of the configuration inside what was explored and that it renders the measure in $\Lambda_{2r}(x)$ independent of the configuration on the outside of $\calE$. 

The above may be understood as (and may be show to imply that)
\begin{align}\label{eq:harmonicity_is_local}
	\mu_{\delta \bbZ^2} \big[ \Delta_rh(x) \big| h \text{ on }\partial \Lambda_{4r}(x) \big] \xrightarrow{\bbP} 0 \qquad \text{ as $\delta \to 0$.}
\end{align}

Using weak-mixing results (see Lemma~\ref{lemma:weak_mixing}) the above convergence will be transported 
to finite-domain measures $\mu_{G}^\zeta$, as in Theorem~\ref{thm:harmonicity}. 
Finally, in Section \ref{sec:L2bounds_proof}, we develop moment bounds for the conditional expectation above, 
that will allow us to upgrade the convergence in probability to a convergence in $L^1$. 
This will ultimately lead to Theorem~\ref{thm:harmonicity}.

\begin{remark}Theorem~\ref{thm:explored_harmo} and \eqref{eq:harmonicity_is_local} state that the harmonicity of the full-plane 
multi-point function is a local property and is not due to global cancellations. \end{remark}

\subsection{Organisation}

The rest of the paper is organised as follows. 
Sections~\ref{sec:review_percolation_representation} and \ref{sec:mix_reg} introduce the percolation representation of the six-vertex model 
and derive relatively direct consequences of it, such as $L^2$ bounds and weak mixing properties.
Section~\ref{sec:almost_markov_proof} states and proves the strong mixing result mentioned above. 
In Section~\ref{sec:spectral_representation_proof}, we prove certain properties of the spectral representation of the full-plane six-vertex measure; some of these already appeared in \cite{DumKozLamMan26}. 
They are then used in Section~\ref{sec:intermediate_proof} to prove Theorem~\ref{thm:intermediate_statement}. 
In Section \ref{sec:harmonicity_proof} we combine Theorem~\ref{thm:intermediate_statement} and the strong mixing of Section~\ref{sec:almost_markov_proof} to prove Theorem~\ref{thm:explored_harmo}. We then deduce Theorem~\ref{thm:harmonicity}.
Finally, Sections~\ref{sec:boundary_conditions_h} and~\ref{sec:GFF_convergence} deduce Theorem~\ref{thm:even}, 
that is the convergence of the height function to the Gaussian Free field.

\section{Basics of the percolation representation}
\label{sec:review_percolation_representation}

In this section, we define a graphical representation associated to the six-vertex model. The representation appeared already in \cite{GL25} and was used extensively in~\cite[Part~D]{DumKozLamMan26}. We will limit here our description to the minimal set of properties used in this paper. 

\subsection{From heights to percolation}

Let $\bbZ^2_{\even}$ denote the graph with vertex-set $F_\even := \{(x+y +\tfrac12, x-y + \tfrac12) : x,y \in \bbZ\}$ and edges $E_\even$ between the 
vertices at distance $\sqrt 2$ of each other. Note that the vertices of $\bbZ^2_{\even}$ are the centres of the even faces of $\bbZ^2$; we will identify the two. 

Let $h$ be a height function defined on $F(\Z^2)$ or on a subset of it.
We couple $h$ with two percolation configurations $\omega = (\omega^+, \omega^-)$ on the edges of $\bbZ^2_\even$ as follows. 

Let $(X_e)_{e \in  E_\even}$ be i.i.d. Bernoulli variables of parameter $ 1/\c\in[0,1]$ (recall that $\c\ge1$).
For any edge $e= (uv) \in E_\even$, if we write $(u^* v^*)$ for its dual edge (with $u^*$ and $v^*$ therefore being centres of odd faces of $\mathbb Z^2$), 
\begin{itemize}
	\item if $h(u) \neq h(v)$, set $\omega^+(e)  = \omega^-(e) = 0$; 
	\item if $h(u) = h(v) \in 4\bbZ$ and $h(u^*) \neq h(v^*)$, set $\omega^+(e)= 1$ and $\omega^-(e) = 0$; 
	\item if $h(u) = h(v) \in 4\bbZ$ and $h(u^*) = h(v^*)$, set $\omega^+(e)= X_e$ and $\omega^-(e) = 0$; 
	\item if $h(u) = h(v) \in 4\bbZ+2$ and $h(u^*) \neq h(v^*)$, set $\omega^+(e)= 0$ and $\omega^-(e) = 1$; 
	\item if $h(u) = h(v) \in 4\bbZ +2$ and $h(u^*) = h(v^*)$, set $\omega^+(e)= 0$ and $\omega^-(e) =  X_e$.
\end{itemize}
From this point onward,  $\mu_G^\zeta$ and $\mu_{\Z^2}$ denote the joint laws of $(h,\omega)$.

Notice that no edge of $\omega^+$ is incident to an edge of $\omega^-$.
Furthermore, each connected component of $\omega^+$ or $\omega^-$ has constant height value.  
We will write $(\omega^+)^*$ for the dual configuration of $\omega^+$. 
Then, each connected component of $(\omega^+)^*$ has total height oscillation at most $3$. 

Finally,  if $A,B$ are connected sets of faces and $h$ is a height function with 
\begin{align}
\max_{z \in A} h(z) \leq \min_{z \in B} h(z) - 4(k+1), 
\end{align}
then there exists $2k$ circuits of $\bbZ^2_\even$ separating $A$ from $B$, contained alternatively in $\omega^+$ and $\omega^-$.
Indeed, it suffices to consider the cluster in $\bbZ^2_\even$ of different even heights between $\max_{z \in A} h(z)$ and $\min_{z \in B} h(z)$.
Their boundaries are part of $\omega^+$ or $\omega^-$, depending on the parity of their half-height.

\subsection{Basic properties of the percolation configuration}

An even domain is the set of faces of $\bbZ^2$ contained on or inside a circuit of $\bbZ^2_\even$. 
\im{See if this is coherent with the intro.}

\begin{lemma}[Spatial Markov property for the joint law]\label{lem:measures}
	Let $\calD$ be an even domain and write $\gamma$ for the circuit of $\bbZ^2_\even$ surrounding it. 
	Then, there exist laws $\mu_\calD^{+}$ and $\mu_\calD^-$ for $(h,\omega^+,\omega^-)$ on $\calD$ such that, 
	for any $G \supset \calD$ and any boundary condition $\zeta$ on $G$, 
	\begin{align}
		\mu_G^\zeta [\blank | (h,\omega) \text{ on and outside $\gamma$ such that $\omega^+ \equiv 1$ on $\gamma$}] &=  \mu_\calD^{+},\\
		\mu_G^\zeta [\blank | (h,\omega) \text{ on and outside $\gamma$ such that $\omega^- \equiv 1$ on $\gamma$}] &=  \mu_\calD^{-},
	\end{align}
	where in the above we refer to the law of $(h,\omega^+, \omega^-)$ in $\calD$, up to an additive constant in $4\bbZ$.
\end{lemma}

The proof is contained in \cite[Part~D]{DumKozLamMan26}. The marginal of $h$ for the measure $\mu_\calD^{+}$ and $\mu_\calD^{-}$ is the six-vertex measure with constant boundary conditions and no diagonal interactions across the boundary defined in the introduction. 
We will refer to these as the measures with $+$ or $-$ boundary conditions for $\omega$. 

\begin{remark}
Similar measures may be defined on even domains with $+$ boundary conditions for certain arcs and $-$ boundary conditions for other arcs. 
In these definitions, the heights of the different arcs is not prescribed. \im{This may be needed for the quad RSW. Check with  Emile}
\end{remark}

\begin{proposition}[FKG for $\omega^+$]\label{prop:FKG_omega+}
	Let $\calD$ be an even domain. Under both  $\mu_\calD^{+}$ and $\mu_\calD^{-}$, 
	the couple $(\omega^+, 1-\omega^-)$ satisfies the FKG lattice condition. 
	The same holds under the infinite volume measure $\mu_{\bbZ^2}$.

	More generally, if $V$ is a set of {\em even} faces and $E$ is a set of edges between even faces, in the measure $\mu_\bbZ^2$, conditionally on $h_4$ on $V$ and the configuration of $\omega$ on $E$, the pair $(\omega^+, 1-\omega^-)$ satisfies the FKG lattice condition.
	\im{This is quite confusing, since we later mention that conditioning on heights, FKG may fail. If not needed, we should not mention it.
	That we may condition on $\omega$ on $E$ is a direct consequence of the FKG lattice condition, so that can even be omitted.} 
	
	In particular, if $\calD' \subset \calD$ is an even domain,  $\omega^+$ sampled according to $\mu_{\calD'}^{+}$ stochastically dominates 
	the restriction to $\calD'$ of $\omega^+$ sampled according to $\mu_{\calD}^{+}$.
\end{proposition}

The proof of the above may be found in~\cite[Part~D]{DumKozLamMan26}.
We draw the attention of the reader to the fact that the FKG property for $\omega^+$ may fail under boundary conditions involving the full height function, or the height function on both even \emph{and} odd faces. Let us mention one last result which will be useful later.

\begin{proposition}[Uniform odd heights from the percolation]
\label{prop:odd_spins_uniform}
    Under $\mu_{\bbZ^2}$ conditioned on the pair $((h_4(z))_{z \in \bbZ^2_{\even}},\omega)$, 
    the odd height function $(h_4(z))_{z \in \bbZ^2_{\odd}}$ is constant on each cluster of $\omega^*$, and is obtained by i.i.d. uniform $\pm 1$ random variables on each such cluster.
\end{proposition}

The previous result also holds under measures $\mu_G^+$ and $\mu_G^-$, since, 
by Lemma \ref{lem:measures}, these measures are obtained from $\mu_{\bbZ^2}$ by conditioning on certain edges of $\omega$.

\subsection{RSW theory for $\omega^+$}

A {\em tube of quality} $\eta$ is a union of distinct boxes $\Lambda_{\eta}(x_i)$ for $i = 1,\dots, K$ with $K \leq \eta^{-2}$ 
and 
$x_1,\dots, x_K \in 2\eta \bbZ^2$ with $\|x_{i+1} - x_i \|= 2\eta$ for each $i$. 
The tube is said to be {\em closed} if $\|x_{K} - x_1 \|= 2\eta$. 

Fix two simply connected finite domains formed of unions of faces of $\delta\bbZ^2$, $\calD_{\rm out}$ and $\calD_{\rm in}$. 
Assume that $\calD_{\rm in} \subset \calD_{\rm out}$.
The region $\calD_{\rm out} \setminus \calD_{\rm in}$ will be called a {\em discrete topological annulus}. 
It will be called even if both $\calD_{\rm in}$ and $ \calD_{\rm out}$ are even domains. 
We will say it is {\em  of quality} $\eta$ if there exists a closed tube of quality $\eta$ in $\calD_{\rm out}$ surrounding $\calD_{\rm in}$.

The measure $\mu_{\calD}^+$ (and similarly $\mu_{\calD}^-$) from Lemma~\ref{lem:measures} may also be defined on annuli $\calD_{\rm out} \setminus \calD_{\rm in}$ where $ \calD_{\rm in} \subset \calD_{\rm out}$ are even domains. 
Under these measures, the height difference between the inner and outer circuit is not specified.

\begin{proposition}[RSW with $\omega$ b.c.]\label{prop:RSW+}
	For each $\eta > 0$ there exists $c  = c(\eta)> 0$ such that, for any $\delta > 0$ small enough, the following holds. 
	For any  even discrete topological annulus  $\calD_{\rm out} \setminus \calD_{\rm in}\subset\delta \bbZ^2$ of quality $\eta$  
	\begin{align}\label{eq:RSW+}
	\mu_{\calD_{\rm out} \setminus \calD_{\rm in}}^{-} [ \omega^+ \text{ contains a circuit around $\calD_{\rm in}$} ] \geq c. 
	\end{align}
\end{proposition}

A more elaborate version of the above may be expressed as follows. 
An even quad is an even domain $\calQ$ with four marked points (i.e.~vertices of $\bbZ^2_\even$) $a,b,c,d$ on its boundary, in counter-clockwise order. 
The conformal length of $(\calQ,a,b,c,d)$ is the unique $\rho$ such that there exists a conformal map from $(\calQ,a,b,c,d)$  to the rectangle $[0,1] \times [0,\rho]$ 
that extends by continuity to the boundary of $\calQ$ and maps $a,b,c,d$ to the four corners of the rectangle, starting with the top-left one. 

Write $\mu_{\calQ}^{\rm alt}$ for the measure on $\calQ$ with boundary conditions $+$ for $\omega$ on the arcs $(ab)$ and $(cd)$ of the boundary 
and $-$ on the rest of the boundary. Here, the heights of the different arcs is not prescribed. 

\begin{corollary}[RSW in terms of conformal length]\label{prop:RSW_quads}
	For each $\rho > 0$ there exists $c  = c(\rho)> 0$ such that, for any quad $(\calQ,a,b,c,d)$ of conformal length at most $\rho$, 
	\begin{align}
	\mu_{\calQ}^{\rm alt} [ \omega^+ \text{ connects the arcs $(ab)$ and $(cd)$} ] \geq c. 
	\end{align}
\end{corollary}

\begin{proof}
The proof follows from Proposition~\ref{prop:RSW+} using the technique \im{Probably too vague} of \cite{DumManTas21}.
\end{proof}

Let us also mention a direct corollary of Proposition~\ref{prop:RSW+}, which will be useful to us later. 

\begin{corollary}\label{cor:heigh_diff_annuli}
	For any $\eta > 0$, there exists $c > 0$ such that, for any $\delta > 0$ small enough, 
	any  even discrete topological annulus  $\calD_{\rm out} \setminus \calD_{\rm in}\subset\delta \bbZ^2$ 
	with $\calD_{\rm in}$ having diameter at least $\eta$ and $\calD_{\rm out}^c$ intersecting  $\Lambda_{1/\eta}$, for every $k\ge0$,
	\begin{align}\label{eq:bound_increments}
	\mu_{\calD_{\rm out} \setminus \calD_{\rm in}}^{+} \big[\, \big| h|_{\partial \calD_{\rm out}} -h|_{\partial \calD_{\rm in}} \big| \geq 4k\big] \leq \exp(-ck). 
	\end{align}
\end{corollary}

Note that, due to the boundary conditions, $h$ is constant on $\partial \calD_{\rm out}$ and on $\partial \calD_{\rm in}$, with the difference between these constants being a multiple of $4$. 

\begin{proof}
	We first observe that, as a consequence of Proposition~\ref{prop:RSW+} and the FKG property for $\omega^+$ 	(Proposition~\ref{prop:FKG_omega+}), there exists $c_0 = c_0 (\eta) > 0$ such that 
	\begin{align}\label{eq:diff_is_0}
	\mu_{\calD_{\rm out} \setminus \calD_{\rm in}}^{+} \big[\partial \calD_{\rm in} \xleftrightarrow{\omega^+}\partial \calD_{\rm out}\big] \geq c_0.
	\end{align}
	The event above guarantees in particular that the height increment between $\partial \calD_{\rm in} $ and $\partial \calD_{\rm out} $ is zero. 
	
	We turn to proving \eqref{eq:bound_increments}. If the height difference between  $\partial \calD_{\rm in} $ and $\partial \calD_{\rm out} $ is at least $4k$, then there exist $2k$ concentric circuits in  $ \calD_{\rm out} $, alternating between $\omega^+$ and $\omega^-$, surrounding $\partial \calD_{\rm in} $.
	By exploring successively the circuits from the inside towards the outside, and applying \eqref{eq:diff_is_0} at every stage, we conclude that 
	\begin{align}\label{eq:alternating_circuits_bound}
	\mu_{\calD_{\rm out} \setminus \calD_{\rm in}}^{+} \big[ \text{ $2k$ concentric alternating circuits of $\omega^+$ and $\omega^-$ around $\partial \calD_{\rm in} $}\big] \leq c_0^k.
	\end{align}
	This provides the desired bound. 
\end{proof}

\subsection{Circuit estimates with height boundary conditions}\label{sec:double_circ}

An  RSW statement also holds under boundary conditions expressed in terms of $h_4$. 
Note that such boundary conditions are quite restrictive for $\omega$ as the FKG inequality for $\omega^+$ may fail in such situations. 
Furthermore, additional complications arise when working with boundaries with several connected components. 
We will use here a mixture of height function estimates, namely the FKG inequality for $|h|$ as well as results of \cite{DumKarManOul22}, 
and properties of the percolation configuration $\omega$. 

For technical reasons, we will directly state an estimate for  {\em double} circuits. 
Let $\calD_{\rm out} \setminus \calD_{\rm in}$ be a discrete topological annulus of quality $\eta > 0$ on $\delta \bbZ^2$ with $\delta > 0$ sufficiently small. 
Then one may identify simply connected domains 
$$\calD_{\rm in} = \calD_0 \subset \calD_1 \subset \calD_2  \subset \calD_3  = \calD_{\rm out}$$ 
such that the concentric  annuli $\calA_1 = \calD_1 \setminus \calD_0$, 
$\calA_2 = \calD_2 \setminus \calD_1$ and  
$\calA_3 = \calD_3 \setminus \calD_2$ are each of quality $\eta / 4$. 
We assume that this is done in a fixed, arbitrary way for any choice of $\calD_{\rm out} \setminus \calD_{\rm in}$. 

We say that $\omega^+$ contains a double circuit in $\calD_{\rm out} \setminus \calD_{\rm in}$ if 
both $\calA_1$ and $\calA_3$ contain circuits of $\omega^+$ separating the inside of the annulus from the outside. 
Write ${\rm DbCirc} ( \calD_{\rm out} \setminus \calD_{\rm in})$ for the event that such a double circuit exists. 

\begin{lemma}\label{lemma:double_circ}
	For each $\eta >0$, $M \geq 1$ and $n\geq 1$, there exists a constant $c  = c(\eta, M, n)> 0$ such that, for any $\delta > 0$ small enough the following holds. 
	For any discrete  topological annulus  $\calD_{\rm out} \setminus \calD_{\rm in}\subset\delta \bbZ^2$ of quality $\eta$, 
	any connected sets of faces $R_1,\dots, R_n$ and any realisations $\xi_1, \dots, \xi_n$ of $h_4$ on them such that 
	\begin{itemize} 
	\item no $R_i$ intersects the annulus $\calD_{\rm out} \setminus \calD_{\rm in}$; 
	\item each $R_i$ has diameter at least $\eta$ and is contained in $\Lambda_{1/\eta}$;  
	\item $\osc(\xi_i) \leq M$ for all $i$,
	\end{itemize}
	we have 
	\begin{align}\label{eq:RSWhbc}
	\mu_{\delta\bbZ^2} \Big[ {\rm DbCirc} ( \calD_{\rm out} \setminus \calD_{\rm in}) \,\Big|\,
	\bigcap_{i = 1}^n\{h_4 = \xi_i \text{ on } R_i\} \Big] \geq c. 
	\end{align}
\end{lemma}

\begin{proof}
	Fix all quantities as in the statement. 
	
	Note that the conditioning in \eqref{eq:RSWhbc} is not in terms of the height, but only in terms of its gradient in each region.  
	Indeed, to transform it into a condition on the values $h$ on the regions  $R_i$, 
	one needs to fix the offsets of the height between the regions. 
	Write $O_i \in 4\bbZ$ for the offset of the height in $R_i$: it is defined by identifying a face $x_i$ in each $R_i$ and imposing that 
	$h(x_i) - O_i \in \{0,1,2,3\}$ for each $i$.  
	Since the height function is defined up to global additive constant, we fix $O_1 = 0$. 
	
	Then
	\begin{align*}
		&\mu_{\delta\bbZ^2} \Big[ {\rm DbCirc} \,\Big|\,\bigcap_{i = 1}^n\{h_4|_{R_i} = \xi_i\} \Big] \\
		&= \sum_{o_2,\dots, o_n \in 4\bbZ} 
		\mu_{\delta\bbZ^2} \Big[ {\rm DbCirc} \,\Big|\,\bigcap_{i = 1}^n\{h|_{R_i} = \xi_i + o_i\} \Big]
		\mu_{\delta\bbZ^2} \Big[	\bigcap_{i = 2}^n\{O_i = o_i\} \,\Big|\,\bigcap_{i = 1}^n\{h_4|_{R_i} = \xi_i \} \Big].
	\end{align*}
	To be precise, by $\{h|_{R_i} = \xi_i + o_i\}$ we mean that $h_4 = \xi_i$ on $R_i$ and that $h(x_i) - o_i \in \{0,1,2,3\}$. 
	
	\begin{figure}
	\begin{center}
	\includegraphics[width = 0.6\textwidth]{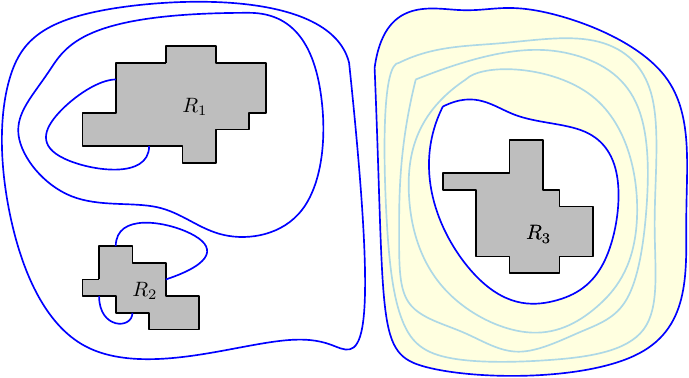}
	\caption{An instance with three sets $R_i$, where the offset of $R_3$ is larger than those of $R_1$ and $R_2$. 
	The blue paths are part of $\omega^+$.  
	This imposes the existence of a number of loops of $\omega^+$ separating $R_1$ and $R_2$ from $R_3$, 
	which in turn forces the existence of an annulus (yellow) with $\omega^+$ boundary conditions
	and a large height difference between its inner and outer boundaries. }
	\label{fig:Rs}
	\end{center}
	\end{figure}
	
	We will first explain that there exists $C = C(\eta, M,n)$ such that
	\begin{align}\label{eq:offset_bound}
		\mu_{\delta\bbZ^2} \Big[	\max_i |O_i | \leq C  \,\Big|\,
		\bigcap_{i = 1}^n\{h_4|_{R_i} = \xi_i \} \Big] \geq \tfrac12.
	\end{align}
	We will then bound the probability of ${\rm DbCirc}$ when offsets satisfy the condition above. 
	\medskip 
	
	\noindent
	\textbf{Bounding the offsets: proof of \eqref{eq:offset_bound}.}
	If we suppose that $\max_i |O_i | > C$ and we take $C$ sufficiently large, a 
	direct computation states that there must exits an  even topological annulus 
	$\calB_{\rm out} \setminus \calB_{\rm in}$ with $\calB_{\rm in}^c$ intersecting $\Lambda_{1/\eta}$, 
	disjoint from all regions $R_i$, separating some of the regions $R_i$ from the other,
	with $\partial \calB_{\rm in} \subset \omega^+$, $\partial \calB_{\rm out} \subset \omega^+$ and 
	$$\big| h|_{\partial \calB_{\rm out}} -h|_{\partial \calB_{\rm in}}\big| \geq \frac{C}{4nM}.$$
	Furthermore, $\calB_{\rm in}$ may be explored from the inside and $\calB_{\rm out}$ from the outside. 
	Indeed, if we suppose that $R_1,\dots, R_j$ are the regions in $\calB_{\rm in}$ and the other are outside, 
	it suffices to take $\partial \calB_{\rm in}$ and $\partial \calB_{\rm out}$ as 
	the innermost and outermost, respectively, circuits of  $\omega^+$ 	separating $R_1,\dots, R_j$ from the other regions. 
	When conditioning on $\calB_{\rm in}$, $\calB_{\rm out}$ and on $h_4|_{R_i} = \xi_i$ for each $i$, 
	the measure in $\calB_{\rm out} \setminus\calB_{\rm in}$ is $\mu_{\calB_{\rm out} \setminus\calB_{\rm in}}^+$. 
	See Figure~\ref{fig:Rs}.

	Thus, accounting for the combinatorial factor which consists in selecting the regions surrounded by the annulus, we find 
	\begin{align}
		\mu_{\delta\bbZ^2} \Big[	\max_i |O_i | > C  \,\Big|\,\bigcap_{i = 1}^n\{h_4|_{R_i} = \xi_i \} \Big] 
		&\leq 2^n \sup_{\Lambda_{\rm in}, \Lambda_{\rm out}} \mu_{\Lambda_{\rm out} \setminus \Lambda_{\rm in}}^+\big[\,\big| h|_{\partial \Lambda_{\rm out}} -h|_{\partial  \Lambda_{\rm in}}\big| \geq \tfrac{C}{4nM}\big]\nonumber\\
		&\leq 2^n\cdot \exp(-c_0 \tfrac{C}{4nM}),\label{eq:offset_bound2}
	\end{align}
	where the $\sup$ is over all possible selections $\Lambda_{\rm in}$ and $\Lambda_{\rm out}$ of $\calB_{\rm in}$ and $\calB_{\rm out}$. Note that, since $\calB_{\rm in}$ must contain at least one of the regions $R_i$, it has diameter at least $\eta$. 
	Corollary~\ref{cor:heigh_diff_annuli} provides the second inequality, with $c_0$ depending only on $\eta$. 
	Taking now $C=C(\eta, M,n)$ large enough, we conclude that the above is smaller than $1/2$. 
	This concludes the proof of~\eqref{eq:offset_bound}.

	\medskip 
	
	\noindent
	\textbf{Bound the probability of ${\rm DbCirc}$ when offsets are small.}
	Fix any  offsets $o_2,\dots, o_n \in 4\bbZ$ with $\max_i |o_i| \leq C$, where $C$ is the constant appearing in \eqref{eq:offset_bound}. 
	Write $\zeta$ for the actual heights on each of the regions $R_i$ defined by $\xi_i + o_i$. 
	Then, we have 
	$$ \max_{z \in \cup_i R_i} |\zeta(z)|\leq C+2M.$$
	
		\begin{figure}
	\begin{center}
	\includegraphics[width = 0.7\textwidth]{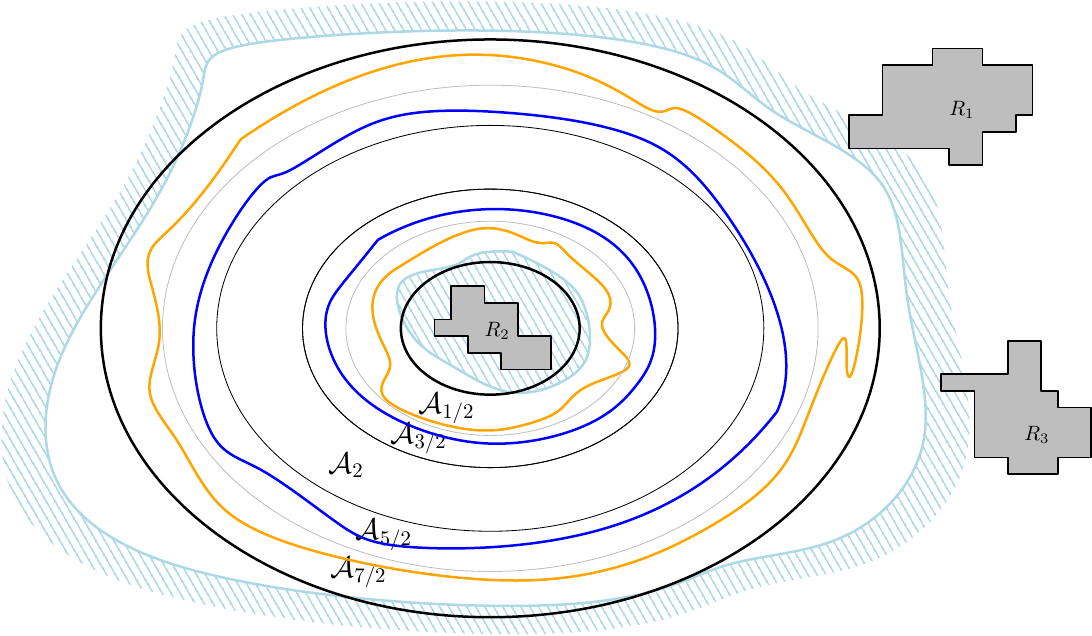}
	\caption{The concentric annuli $\calA_{1/2}, \dots, \calA_{7/2}$ separating $R_1$ from $R_2$ and $R_3$. 
	When $\calA_{1/2}$ and $\calA_{7/2}$ contain circuits of $|h| \geq C+2M+4$ (orange), 
	these force the existence of $\omega^+$ circuits (light blue) separating $\calA_{3/2}$ from $R_1$ and $\calA_{5/2}$ from  $R_2$ and $R_3$. The latter may exit $\calA_{1/2}$ and $\calA_{7/2}$. By exploring such circuits, we may then create with positive probability circuits of $\omega^+$ in $\calA_{3/2}$ and $\calA_{5/2}$ (blue). }
	\label{fig:double_circ}
	\end{center}
	\end{figure}

	Assume that $R_1,\dots, R_\ell$ are contained in $\calD_{\rm in}$ and $R_{\ell+1},\dots, R_n$ are outside $\calD_{\rm out}$, with $0 < \ell < n$;
	we will treat the degenerate cases $\ell = 0$ and $\ell = n$ after. 
	For technical reasons, we further split the annuli $\calA_1$ and $\calA_3$ into concentric disjoint annuli $\calA_{1/2}$, $\calA_{3/2}$ and $\calA_{5/2}$, $\calA_{7/2}$, respectively, each of quality $\eta/16$. 
	Applying \cite[Thm. 1.4]{DumKarManOul22} and the FKG inequality for $|h|$  \cite[Prop. 2.4]{DumKarManOul22} 
	we conclude that 
	\begin{align}\label{eq:DbCirc|h|}
		\mu_{\delta\bbZ^2} \Big[\calA_{1/2} \text{ and }\calA_{7/2} \text{ contain circuits of $|h| \geq C+2M+4$} \,\Big|\ h=\zeta \text{ on } \bigcup_i R_i \Big] \geq c_1,
	\end{align}
	with $c_1 > 0$ depending only on $\eta$, $C$ and $M$. Write ${\rm DbCirc}(|h|)$ for the event in the probability above. 
	
	When ${\rm DbCirc}(|h|)$ occurs, there exist circuits $\Gamma_{\rm in}$ and $\Gamma_{\rm out}$ of $\omega^+$
	separating $R_1,\dots, R_\ell$ from $\calA_{3/2}$ and $\calA_{5/2}$ from $R_{\ell+1},\dots, R_n$, respectively. 
	Furthermore, such circuits may be explored from the inside and outside, respectively. 
	See Figure~\ref{fig:double_circ}.

	Thus, applying the monotonicity in boundary conditions for $\omega^+$ (Proposition~\ref{prop:FKG_omega+}), we find 
\begin{align}
	\mu_{\delta\bbZ^2} \Big[ 	{\rm DbCirc} ( \calD_{\rm out} \setminus \calD_{\rm in}) \Big|{\rm DbCirc}(|h|) \text{ and } h=\zeta \text{ on } \bigcup_i R_i \Big]& \\
	\geq \mu_{\delta\bbZ^2}\big[ \omega^+ \text{ has circuits in $\calA_{3/2}$ and in $\calA_{5/2}$}\big]
	&\geq c_2, 
\end{align}
with $c_2> 0$ provided by Proposition~\ref{prop:RSW+} and depending only on $\eta$. 
The same conclusion is reached when $\ell = 0$, by exploring only $\Gamma_{\rm out}$ and when $\ell=n$ by exploring only $\Gamma_{\rm in}$.

	Combining the above with \eqref{eq:DbCirc|h|} and summing over all values of $o_2,\dots, o_n$ bounded by $C$, we find 
	\begin{align}
		\mu_{\delta\bbZ^2} \Big[ 	{\rm DbCirc} ( \calD_{\rm out} \setminus \calD_{\rm in}) \,\Big| 
		\bigcap_{i = 1}^n\{h_4|_{R_i} = \xi_i \} \text{ and } 	\max_i |O_i | \leq C \Big] \geq c_1c_2.
	\end{align}
	Finally, \eqref{eq:offset_bound} leads us to the desired bound. 
\end{proof}

\subsection{Flower domain exploration}

Fix a discrete topological annulus $\calD_{\rm out} \setminus \calD_{\rm in}$ of faces of $\delta \bbZ^2$. 
Define the exploration procedure from outside to inside in $\calD_{\rm out} \setminus \calD_{\rm in}$ as follows. 

Consider the percolation $\omega^+$ and its dual $(\omega^+)^*$. 
Starting from the boundary of  $\calD_{\rm out}$, explore all interfaces between $\omega^+$ and $(\omega^+)^*$ up to the first time they exit $\calD_{\rm out}$ or enter $\calD_{\rm in}$. 
By exploring the interfaces, we mean revealing the state of the percolation edges and the values of $h_4$ on the faces on either side of the exploration. 
See Figure~\ref{fig:flower}  for an example. 

	\begin{figure}
	\begin{center}
	\includegraphics[width = 0.55\textwidth]{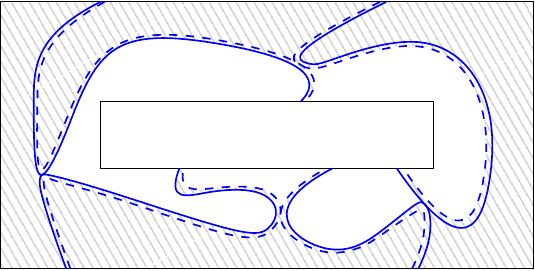}
	\caption{A topological annulus and the interfaces (in blue) between the primal percolation $\omega^+$ and its dual, explored from the outer boundary towards the inside.  The explored region is hashed and the white middle region is the flower domain explored from outside to inside. }
	\label{fig:flower}
	\end{center}
	\end{figure}

Write ${\rm Exp}$ for the set of faces resulting from this exploration and $h_{\rm exp}$ and $\omega_{\rm exp}$ for the resulting height and percolation configuration. 
Write $\calF$ for the connected component of $\calD_{\rm in}$ in $\delta \bbZ^2\setminus {\rm Exp}$. 
Note that the boundary of $\calF$ is formed of a number of arcs $P_1,\dots, P_k$ of $\omega^+$ or $(\omega^+)^*$. 
The height on each of the arcs of $\omega^+$ is constant, while the arcs of $(\omega^+)^*$ have oscillations at most $3$.

We call $\calF$ the flower domain explored from $\calD_{\rm out}$ to $\calD_{\rm in}$ and $P_1,\dots, P_k$ its {\em petals}. 
We write $\osc( \calF)$ for the oscillation of the explored height function on the boundary of~$\calF$. 

Note that the measure inside $\calF$ conditionally on the explored part has complicated boundary conditions. 
Indeed, they are a mix of heights and percolation information. 
To obtain boundary conditions induced exclusively by heights, one may explore the {\em thickened} flower domain $\calF^+$, which is obtained by further exploring the height in each face of $\calF$ adjacent to its boundary. 
The notation $\osc(\calF^+)$ extends trivially. 

In a similar way, we define the flower domain explored from the inside towards the outside, or from $\calD_{\rm in}$ to $\calD_{\rm out}$. 
Its petals face outwards and $\calF^c$ is a simply connected set of faces.

The following result will allow us to control the total oscillation of $h$ on $\partial \calF$.
Write $A_k(\calD_{\rm in }, \calD_{\rm out})$ for the event that there exist two connected paths $\gamma_1$ and $\gamma_2$ of faces of $\delta\bbZ^2$ between $\partial \calD_{\rm in}$ and $\partial\calD_{\rm out}$ 
with 
\begin{align}
 \min_{z \in \gamma_2} h(z)  \geq \max_{z \in \gamma_1} h(z) +k.
\end{align}
We call this the $k$-arm event, due to its similarity with $k$-arm percolation events. 

\im{I believe that this lemma contains the same info as the one written before - Lemma~\ref{lemma:general_rsw} (except that I changed $k^2$ to $k$. }

\begin{lemma}[Bound on the flower domain oscillation]\label{lemma:arm_event_bound}
    For any $\eta > 0$, there exists $c>0$ such that, 
for any topological annulus $\calD_{\rm out} \setminus \calD_{\rm in}$ of quality $\eta$
and any $k \geq 1$, 
        \begin{equation}\label{eq:arm_event_bound1}
            \mu_{\Z^2}\big[A_k(\calD_{\rm in }, \calD_{\rm out})\big]\leq \tfrac1c e^{-c  k}.
        \end{equation}
	Furthermore,
        for any measure $\mu_G^\zeta$ with $G$ simply connected \im{We do need $G$ to be simply connected, right? We should not write b.c. $\zeta$ in terms of $h_4$ on non-simply connected graphs.}, 
             \begin{equation}\label{eq:arm_event_bound2}
            \mu_{G}^\zeta\left[A_k(\calD_{\rm in }, \calD_{\rm out})\right]
            \leq \tfrac1c e^{-c (  k - \osc(\zeta))}.            
        \end{equation}
\end{lemma}

In the second point, $G$ does not need to contain $\calD_{\rm out}$. 
In such cases, the arms need to be contained both in $\calD_{\rm out}$ and in $G$.

\begin{proof}
%
	Let us start by proving \eqref{eq:arm_event_bound1}.
	We may cover $\calD_{\rm out} \setminus \calD_{\rm in }$ by $C \eta^{-1}$ 
	annuli of the form $\Lambda_\eta \setminus \Lambda_{\eta/2}$, 
	such that, if $A_k(\calD_{\rm in }, \calD_{\rm out})$ occurs, then $A_{\frac{k}{C\eta}}(\Lambda_{\eta/2},\Lambda_\eta)$ 
	occurs in at least one of these annuli. Here $C \geq1$ is some universal constant. 
	Applying~\cite[Theorem~4.9]{DumKozLamMan26} and a union bound, we find~\eqref{eq:arm_event_bound1}.
	
	\medskip 
	
	We turn to  \eqref{eq:arm_event_bound2}. 
	Fix $G$ and $\zeta$. It suffices to consider the case where $k \geq \osc(\zeta) + 8$. 
	We start off by exploring the union $\calC_{\partial G}(\omega^+)^*$  of the clusters of the boundary of $G$ in $(\omega^+)^*$. Recall that each cluster of $(\omega^+)^*$ has total oscillation at most $3$. 
	Thus, the total oscillation of the explored part is at most $\osc(\zeta) + 8$.

    Under the conditional measure $\mu_G^\zeta[\blank| \calC_{\partial G}(\omega^+)^*]$,
    the marginal in each component of the the unexplored part 
 $G\setminus    \calC_{\partial G}(\omega^+)^*$  has $+$ boundary conditions for the percolation. 
 Furthermore, the different connected components are independent. 
 
 Now, if $A_k(\calD_{\rm in }, \calD_{\rm out})$ occurs, then 
 $A_{k - \osc(\zeta) - 8}(\calD_{\rm in }, \calD_{\rm out})$ occurs in the complement of $\calC_{\partial G}(\omega^+)^*$. 
 In particular, $A_{\frac{k - \osc(\zeta) - 8}2}(\calD_{\rm in }, \calD_{\rm out})$ occurs in one of the connected components of 
 $G\setminus    \calC_{\partial G}(\omega^+)^*$.
Thus, we find
 \begin{align}\label{eq:HHDDA}
	 \mu_{G}^\zeta\left[A_k(\calD_{\rm in }, \calD_{\rm out})\right]
             \leq \max_{H}  \mu_{H}^+\left[ A_{\frac{k - \osc(\zeta) - 8}2}(\calD_{\rm in }, \calD_{\rm out})\right],
 \end{align}
where the maximum is over all even domains $H$. 

We will now argue that there exits $c = c(\eta) > 0$ such that, for any even domain $H$ and $\ell \geq 1$,
\begin{align}\label{eq:HHDD}
	\max_{H}  \mu_{H}^+\left[ A_{8\ell}(\calD_{\rm in }, \calD_{\rm out})\right] \leq \tfrac{1}c e^{-c \ell}.
\end{align}
\im{In the previous version it was stated that we adapt the proof of~\cite[Theorem~4.9]{DumKozLamMan26}, but this is not what was done. A new proof was done, giving just the $e^{-ck}$ bound.}

The event above implies the existence of $\ell$ concentric circuits $\chi_1\dots, \chi_{\ell}$ alternating between $\omega^+$ and $\omega^-$, all crossing the annulus $\calD_{\rm out}\setminus\calD_{\rm in}$ from the inside to the outside. 
These may be explored from outside to inside. Given that $\chi_1,\dots, \chi_j$ exist, the existence of $\chi_{j+1}$ has probability uniformly bounded away from $1$ due to \eqref{eq:RSW+}. This is the same procedure as in the proof of Corollary~\ref{cor:heigh_diff_annuli}.

This proves \eqref{eq:HHDD}, which, when inserted in \eqref{eq:HHDDA} proves the lemma. 
\end{proof}

\section{Consequences of the RSW theory: mixing and regularity}\label{sec:mix_reg}

We now state a series of relatively standard consequences of the percolation representation, 
and most importantly of the RSW property, Proposition~\ref{prop:RSW+}.

\subsection{$L^2$-bounds}
\label{sec:L2bounds_proof}

For a face $x$ and $r >0$, write 
\begin{align}\label{eq:Delta_tilde_def}
\widetilde \Delta_r h(x) = \mu_{\delta \bbZ^2} \big[\Delta_r h(x) \,\big|\, h_4 \text{ on }\partial \Lambda_{2r}(x) \big].
\end{align}
This is a variable that has same expectation as $\Delta_r h(x)$, but will be shown below to have much smaller variance. 
It is the average value of the Laplacian given the configuration further away from where the Laplacian is measured. 
The same definition applies in domains, as long as they contain $\Lambda_{2r}(x)$.

\begin{lemma}[$L^2$-bound of the averaged Laplacian]
    \label{lemma:L2bounds}
    There exists a universal constant $C > 0$ such that
    for any $x$, $r >0$ and $\delta > 0$ small enough, 
    \begin{equation}\label{eq:L2_bound_infinite}
        \mu_{\delta \Z^2}[\widetilde\Delta_rh(x)^2]\leq C.
    \end{equation}
    Furthermore, for any $G \subset \delta\bbZ^2$ containing $\Lambda_{2r}(x)$ and any boundary condition $\zeta$, 
    \begin{equation}\label{eq:L2_bound_domain}
            \mu_G^\zeta[\widetilde\Delta_rh(x)^2]\leq C \cdot (\osc(\zeta))^2. 
   \end{equation}
\end{lemma}

The following fact, which follows immediately from the FKG inequality and the spatial Markov property for $h$ \cite{DumKarManOul22}, will be useful below. 

\begin{lemma}
    \label{lemma:boundedness}
    Let $G$ be a domain and $\zeta$ a boundary condition expressed in terms of heights. 
    Then for any face $x \in G$
    \begin{equation}
	 \min_{z \in \partial G} \zeta(z)\leq \mu_{G}^{\zeta}[h(x)]\leq \max_{z \in \partial G} \zeta(z).
    \end{equation}
\end{lemma}

\im{Might need to give a ref for the proof. I cited \cite{DumKarManOul22} just to put something. Better citations probably exist}

Let us start by illustrating the idea behind Lemma~\ref{lemma:L2bounds} in the full-plane case.
By Lemma~\ref{lemma:boundedness}, we have
\begin{align}
    \label{eq:L2bound_oscillation}
    \mu_{\Z^2}[\widetilde\Delta_rh(x)^2]
    &=
    \mu_{\Z^2}\left[
        \mu_{\Z^2}\left[
            \Delta_rh(x)\middle| h_4|_{\partial \Lambda_{2r}(x)}
        \right]^2
        \right]\\
    &\leq
    \mu_{\Z^2}\left[\osc(h|_{\partial \Lambda_{2r}(x)})^2\right].
\end{align}
The bound on the right is not satisfactory,
as the typical height oscillation on $\partial \Lambda_{2r}(x)$ is of order $\sqrt{\log r/\delta}$,
while we are looking for a bound of order $O(1)$.
By standard $L^2$-martingale theory, the middle term is increasing in the conditioning
$\sigma$-algebra.
The idea is to run an exploration procedure leading to a \emph{richer}
$\sigma$-algebra,
but in a way that \emph{improves} the rightmost term coming from the oscillation on the boundary.

\begin{proof}[Proof of Lemma~\ref{lemma:L2bounds}]
    Fix $r >0$ and $x$. 
    Let $\calF$  denote the fattened flower domain explored from $\Lambda_{2r}(x)$ to $\Lambda_{3r/2}(x)$
    and let $\xi$ denote the boundary conditions induced by the exploration. 
     We start with the case of the infinite volume. 
    By the tower property and Jensen's inequality
    \begin{align}
       \mu_{\delta\Z^2}\big[\widetilde\Delta_rh(x)^2\big]
 	&= \mu_{\delta\Z^2}\left[ \mu_{\delta\Z^2}\left[\Delta_rh(x)\middle| h_4|_{\partial \Lambda_{2r}(x)} \right]^2 \right]\\
&\leq 
\mu_{\delta\Z^2}[  \mu_{\delta\Z^2}\big[ \Delta_rh(x)| \calF, \xi ]^2\big] \\
&\leq  \mu_{\delta\Z^2} [\osc(\xi)^2 ].
\end{align}
	The last inequality follows from Lemma~\ref{lemma:boundedness}.
	
    Now, Lemma~\ref{lemma:arm_event_bound}
    implies a uniform exponential bound on the distribution of $\osc(\xi)$ under $\mu_{\delta\bbZ^2}$.
    This proves that the right-hand side above is bounded uniformly in $r$, provided $\delta$ is small enough. 
    \medskip 
    
    The same procedure in finite domains implies that 
        \begin{align}
       \mu_{G}^\zeta [\widetilde\Delta_rh(x)^2]
	&\leq  \mu_{G}^\zeta [\osc(\xi)^2 ].
	\end{align}
        Observe now that Lemma~\ref{lemma:arm_event_bound} (specifically \eqref{eq:arm_event_bound2})
provides a uniform exponential bound on the distribution of $\osc(\xi) -\osc(\zeta)$ under $\mu_{G}^\zeta$.
	This implies \eqref{eq:L2_bound_domain}.
\end{proof}

We will also need a variant of Lemma~\ref{lemma:L2bounds} for the height function rather than the Laplacian. 
As above, set 
\begin{align}\label{eq:htilde_def}
\widetilde h_r(x) = \mu_G^\zeta \big[h(x) \,\big|\, h_4 \text{ on } \partial \Lambda_{2r}(x) \big].
\end{align}


\begin{lemma} \label{lem:htilde}
	For any $M\geq 1$ and $n\ge 2$, there exists a decreasing function $C:(0,\infty)\to(0,\infty)$ such that,
	for any $\delta>0$ small enough, the following holds. 
	Take $G\subset\delta\bbZ^2$ simply connected, $\zeta$ a boundary condition with $\max_{z\in\partial G}|\zeta(z)|\le M$ 
	and points $u_1,\ldots,u_n\in G$. 
	For each $i$, set $d_i$ to be the distance from $u_i$ to $\partial G$, 
	and choose $0 <r_i\leq d_i/8$ be such that $u_i$ is at distance at least $8r_i$ from all the other $u_j$. Then
	\begin{align}\label{eq:htilde}
    	  \mu_G^\zeta\big[\widetilde \Delta_{r_1} h(u_1)^2\cdot \widetilde h_{r_2}(u_2)^2 \cdot \dots\cdot \widetilde h_{r_n}(u_n)^2\big]&\leq \prod_{i=1}^n C(r_i/d_i)\quad\text{and} \\
            \mu_G^\zeta\big[\widetilde h_{r_2}(u_2)^2 \cdot \dots\cdot \widetilde h_{r_n}(u_n)^2\big]&\leq \prod_{i=2}^n C(r_i/d_i).
   \end{align}
\end{lemma}

\begin{proof}
	The proof is almost identical to that of Lemma~\ref{lemma:L2bounds}, 
	but explorations need to be performed in annuli around each point $u_i$. These annuli are disjoint due to our assumption. 
	The difference of heights between the annuli are controlled as in the proof of Lemma~\ref{lemma:double_circ}.
\end{proof}

\im{I find that the version above needs a more detailed proof.}

\subsection{Quasi-independence}
\label{sec:quasi_independence_proof}

Connected sets of faces $R_1,\dots, R_n$ of $\delta \bbZ^2$ are said to be $\eta$-separated 
if there exist  closed tubes $\Xi_1 \dots, \Xi_n$ of quality $\eta$, with each $\Xi_i$ surrounding $R_i$ but leaving all other $R_j$ outside or vice versa.

\begin{lemma}[Quasi-independence]
    \label{lemma:quasi_independence}
  For any $\eta >0$, $M\geq 1$  and $n\geq 2$, there exists a constant $c = c(\eta,M,n) > 0$ such that, 
  for any $\delta > 0$ small enough,
   any connected sets $R_1,\dots,R_n \subset \Lambda_{1/\eta}$ of faces of $\delta \bbZ^2$
   that are $\eta$-separated and have diameter at least $\eta$, 
   and any assignments $\xi_1, \dots, \xi_n$ of $h_4$  on $R_1,\dots, R_n$ with $\osc(\xi_i)\leq M$ for each $i$,
     \begin{equation}
	c
        \leq
        \frac{\mu_{\delta\bbZ^2}\big[ \bigcap_{i = 1}^n \{h_4|_{R_i}=\xi_i\}\big] }{ \prod_{i = 1}^n\mu_{\delta\bbZ^2}[h_4|_{R_i}=\xi_i]}
    	\leq c^{-1}   .
        \end{equation}
\end{lemma}

\subsubsection{Swapping techniques}\label{subsec:geometric_mixing_annuli}

We start by stating two ``swapping'' lemmas, expressed in terms of partition functions of the model with different boundary conditions. 
These will be instrumental in proving weak and strong mixing statements, respectively.  

Recall that even domains $\calD$ are interiors of loops on $\bbZ^2_\even$
which we denote by $\partial \calD$. \im{See if this was introduced before, Move this to previous definitions. }
Consider four simply connected, even domains 
$\calD_{\rm in}\subset\calD_{\rm out}$ and $\calD_{\rm in}'\subset \calD_{\rm out}'$. 
We will say that the domains form {\em compatible annuli of quality $\eta$} 
if there exists a common closed tube of quality $\eta$ contained in both 
$\calD_{\rm out}$ and $\calD_{\rm out}'$, surrounding both $\calD_{\rm in}, \calD_{\rm in}'$.
See Figure~\ref{fig:swap_ann}.

	\begin{figure}
	\begin{center}
	\includegraphics[width = 0.55\textwidth]{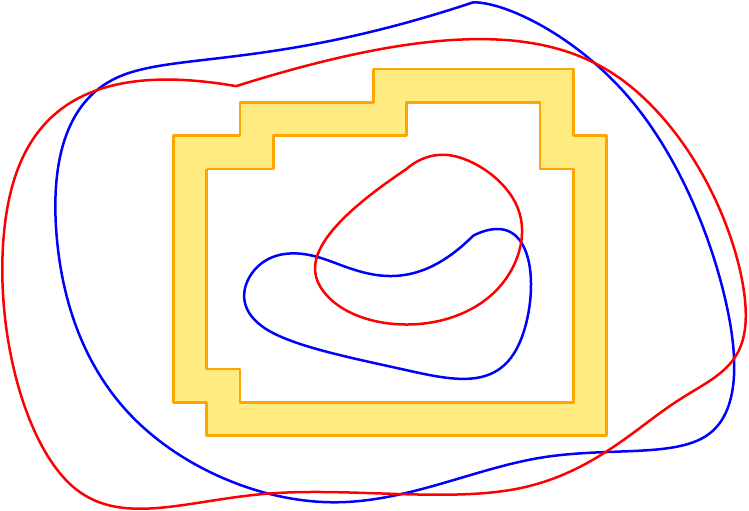}
	\caption{Two compatible annuli $\calD_{\rm out} \setminus \calD_{\rm in}$ (blue) and $\calD_{\rm out}' \setminus \calD_{\rm in}'$ (red) of quality $\eta$. The common tube separating the inside from outside is orange. }
	\label{fig:swap_ann}
	\end{center}
	\end{figure}

Write 
$Z^+(\calD_{\rm out} \setminus \calD_{\rm in})$
for the partition function with $+$ boundary conditions for $\omega$ in $\calD_{\rm out} \setminus \calD_{\rm in}$.
Then, 
\begin{equation}\label{eq:Zdom_Zann}
    Z^+(\calD_{\rm out} \setminus \calD_{\rm in})
    =
    \frac{    Z^+(\calD_{\rm out})}{    Z^+(\calD_{\rm in})}
    \mu_{\calD_{\rm out}}^+[\partial \calD_{\rm in} \subset \omega^+].
\end{equation}
Also, recall that 
\begin{align}
\mu_{\calD_{\rm out} \setminus \calD_{\rm in}}^+ =     \mu_{\calD_{\rm out}}^+[\blank | \partial \calD_{\rm in } \subset \omega^+].
\end{align}


\begin{lemma}[Swapping in annuli]\label{lem:geometric_mixing_annuli}
    For any $\eta > 0$, there exists $c >0$ such that, 
    if the domains $\calD_{\rm in}, \calD_{\rm in}',\calD_{\rm out}, \calD_{\rm out}'$ form compatible annuli of quality $\eta$, 
    \begin{equation}
        \label{eq:weak_mixing_crossing_estimate}
        c\leq
        \frac{Z^+(\calD_{\rm out} \setminus \calD_{\rm in}) Z^+(\calD_{\rm out}' \setminus \calD_{\rm in}') }
        {Z^+(\calD_{\rm out}' \setminus \calD_{\rm in}) Z^+(\calD_{\rm out} \setminus \calD_{\rm in}') }
        \leq c^{-1}.
    \end{equation}
\end{lemma}

\begin{proof}
	Using \eqref{eq:Zdom_Zann}, we find
	 \begin{align}\label{eq:4Z}
 \frac{Z^+(\calD_{\rm out} \setminus \calD_{\rm in})\cdot Z^+(\calD_{\rm out}' \setminus \calD_{\rm in}') }
        {Z^+(\calD_{\rm out}' \setminus \calD_{\rm in}) \cdot Z^+(\calD_{\rm out} \setminus \calD_{\rm in}') }
         = 
\frac{\mu_{\calD_{\rm out}}^+[\partial \calD_{\rm in} \subset \omega^+]\cdot
    \mu_{\calD_{\rm out}'}^+[\partial \calD_{\rm in}' \subset \omega^+] }{
    \mu_{\calD_{\rm out}'}^+[\partial \calD_{\rm in} \subset \omega^+]\cdot
    \mu_{\calD_{\rm out}}^+[\partial \calD_{\rm in}' \subset \omega^+]}.
	 \end{align}
	
	Write $\calT_{\rm out} \setminus \calT_{\rm in}$ for a tube of quality $\eta$ that surrounds both $\calD_{\rm in}$ and $\calD_{\rm in}'$ and is contained in both $\calD_{\rm out}$ and $\calD_{\rm out}'$. 
	By the monotonicity of boundary conditions for $\omega^+$ (see Proposition~\ref{prop:FKG_omega+}), we have 
	\begin{align}
	 \mu_{\bbZ^2}[\partial \calD_{\rm in} \subset \omega^+] 
	 \leq \mu_{\calD_{\rm out}}^+[\partial \calD_{\rm in} \subset \omega^+] 
	 \leq \mu_{\calT_{\rm out}}^+[\partial \calD_{\rm in} \subset \omega^+].
	\end{align}
	The same holds for $\partial \calD'_{\rm in} $ instead of $\partial \calD_{\rm in} $
	and also for the measure $\mu_{\calD_{\rm out}'}^+$.
	
	Moreover, by exploring the outermost circuit of $\omega^+$ in $\calT_{\rm out} \setminus \calT_{\rm in}$ and using again the monotonicity in boundary conditions, 
	we also find 
	\begin{align}
	 \mu_{\bbZ^2}[\partial \calD_{\rm in} \subset \omega^+] 
	& \geq 	 \mu_{\bbZ^2}[ \omega^+ \text{ contains circuit in $\calT_{\rm out}$ around $\calT_{\rm in}$}] \cdot 
	  \mu_{\calT_{\rm out}}^+[\partial \calD_{\rm in} \subset \omega^+]\\
	  &\geq c\mu_{\calT_{\rm out}}^+[\partial \calD_{\rm in} \subset \omega^+],
	\end{align}
	with the constant $c = c(\eta) > 0$ given by \eqref{eq:RSW+}, provided that $\delta$ is small enough. 
	 
	Injecting the above into \eqref{eq:4Z}, we obtain the desired bound. 
\end{proof}

A second swapping will be done in quads.
Recall that a {quad} $\calQ$ is an even domain whose boundary contains four points $a,b,c,d$ in counter-clockwise order. 
Call the four resulting arcs $(ab)$, $(bc)$, $(cd)$ and $(da)$ of its boundary the \emph{left}, \emph{bottom}, \emph{right} and \emph{top} of the quad.
Let $\Horizontal_{\calQ}$ denote the event that there exists a path of $\omega^+$ in $\calQ$ connecting the left and right sides.
Similarly, $\Vertical_\calQ$ denotes the event of that its top and bottom are connected in $\omega^+$.

Consider the following setting. 
Two non-intersecting paths $\gamma_{L}$ and $\gamma_R$ of $\delta\bbZ^2_\even$
crossing the rectangle $[0,r] \times [0,R]$ vertically, with $\gamma_R$ to the right of $\gamma_L$. 
Four paths $\gamma_T^1, \gamma_T^2$ and $\gamma_B^1, \gamma_B^2$ of $\delta\bbZ^2_\even$, contained in  $[0,r] \times [0,R]$, 
each connecting $\gamma_L$ to $\gamma_R$.
These paths are necessarily stuck between $\gamma_L$ and $\gamma_R$ and we will assume that 
each of the former is above each of the latter. 

Define the quads $\calQ_{ij}$ for $i,j \in \{1,2\}$ as having the bottom and top boundaries formed of $\gamma_B^i$ and $\gamma_T^j$, 
and the left and right boundaries formed of segments of $\gamma_L$ and $\gamma_R$.

Let $\calQ_{\rm in}$ be the quad formed of the set of faces between $\gamma_L$ and $\gamma_R$,
above both $\gamma_B^1$ and $\gamma_B^2$, and 
below both $\gamma_T^1$ and $\gamma_T^2$.
Write $\gamma_{T}^{\rm in}$ and $\gamma_{B}^{\rm in}$ for its top and bottom boundaries.

Similarly, define $\calQ_{\rm out}$ as the quad formed of the faces between $\gamma_L$ and $\gamma_R$,
above $\gamma_B^1$ or $\gamma_B^2$, and 
below $\gamma_T^1$ or $\gamma_T^2$.
Write $\gamma_{T}^{\rm out}$ and $\gamma_{B}^{\rm out}$ for its top and bottom boundaries.
See Figure~\ref{fig:swapp_quads} for an illustration.
	\begin{figure}
	\begin{center}
	\includegraphics[width = 0.3\textwidth]{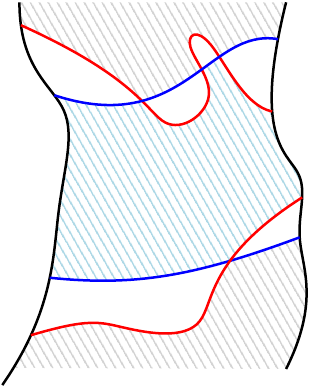}
	\caption{The quads $\calQ_{i,j}$  enclosed laterally by $\gamma_L$ and $\gamma_R$ (black) and vertically by  $\gamma_B^i$ and $\gamma_T^j$ (red for $1$ and blue for $2$). 
	The hashed interior region is $\calQ_{\rm in}$ and the hashed outer region is $\calQ_{\rm out}^c$. }
	\label{fig:swapp_quads}
	\end{center}
\end{figure}

For any of the quads $\calQ$ above,  write $Z^+(\calQ)$ for the partition function of the model with~$+$ boundary conditions in  said quad.

\begin{lemma}[Swapping in quads]\label{lem:geometric_mixing_quadrilaterals}
 With the notation above, 
       \begin{equation}\label{eq:geometric_mixing_quadrilaterals_extension}
        \mu_{\Z^2}[\Horizontal_{\calQ_{\rm in}}]^2
        \leq 
        \frac{Z^+(\calQ_{11})\cdot Z^+(\calQ_{22})}
        {Z^+(\calQ_{12})\cdot Z^+(\calQ_{21})}
        \leq 
        \mu_{\Z^2}[ \Horizontal_{\calQ_{\rm in}}]^{-2}.
    \end{equation}
\end{lemma}

\begin{proof}
	We have 
	\begin{align}
	Z^+(\calQ_{ij}) = Z^+(\calQ_{\rm out})\mu_{\calQ_{\rm out}}^+[\gamma_B^i, \gamma_T^j \in \omega^+].
	\end{align}
	Thus, the ratio in \eqref{eq:geometric_mixing_quadrilaterals_extension} becomes

	\begin{align}\label{eq:geometric_mixing_quadrilaterals_extension2} 
	\frac{\mu_{\calQ_{\rm out}}^+[\gamma_B^1, \gamma_T^1 \in \omega^+]
	\mu_{\calQ_{\rm out}}^+[\gamma_B^2, \gamma_T^2 \in \omega^+]}
	{\mu_{\calQ_{\rm out}}^+[\gamma_B^1, \gamma_T^2 \in \omega^+]
	\mu_{\calQ_{\rm out}}^+[\gamma_B^2, \gamma_T^1 \in \omega^+]}.
	\end{align}
	
	Due the FKG inequality (Proposition~\ref{prop:FKG_omega+}), we find that, for each $i,j \in \{1,2\}$,
	\begin{align*}
	\mu_{\calQ_{\rm out}}^+[\gamma_B^i \in \omega^+] \mu_{\calQ_{\rm out}}^+[\gamma_T^j \in \omega^+]  
	&\leq 	\mu_{\calQ_{\rm out}}^+[\gamma_B^i, \gamma_T^j \in \omega^+] \\
	&\leq 
	\mu_{\calQ_{\rm out}}^+[\gamma_B^i \in \omega^+ \,|\,  \Horizontal_{\calQ_{\rm in}} ] \cdot	\mu_{\calQ_{\rm out}}^+[\gamma_T^j \in \omega^+] \\
	&\leq \frac{\mu_{\calQ_{\rm out}}^+[\gamma_B^i \in \omega^+] \mu_{\calQ_{\rm out}}^+[\gamma_T^j \in \omega^+]}{\mu_{\Z^2}[ \Horizontal_{\calQ_{\rm in}}]}.
	\end{align*}
	Inserting this into \eqref{eq:geometric_mixing_quadrilaterals_extension2}, we obtain the desired bounds. 
\end{proof}

\subsubsection{Proof of Lemma~\ref{lemma:quasi_independence}}

\begin{proof}[Proof of Lemma~\ref{lemma:quasi_independence}]
Fix $\eta$ and $M$ as in the statement. 
	We will do the proof by induction over $n$, the number of regions. 
	The result is immediate for~$n = 1$. 
	
	Fix $n \geq 2$ and assume the result proved for $n-1$. Let $R_1,\dots,R_n$ be $n$ regions and $\xi_1,\dots, \xi_n$ be assignments of heights as in the statement. 
	In light of the induction hypothesis, our goal is to prove that 
	\begin{align}\label{eq:quasiH1}
	c \leq \frac{\mu_{\delta\bbZ^2}\Big[ \bigcap_{i = 1}^n \{h_4|_{R_i}=\xi_i\}\Big]}
	{\mu_{\delta\bbZ^2}\big[ h_4|_{R_1}=\xi_1\big]\mu_{\delta\bbZ^2}\Big[ \bigcap_{i = 2}^n \{h_4|_{R_i}=\xi_i\}\Big]}
	\leq c^{-1},
	\end{align}
	for some constant $c> 0$ depending only on $\eta$,  $M$ and $n$.
	
	Let $\calD_{\rm out} \setminus\calD_{\rm in}$ be an annulus of quality $\eta$ that separates $R_1$ from $R_2,\dots, R_n$, 
	which is to say that $R_1 \subset \calD_{\rm in}$ and $R_2,\dots, R_n \subset \calD_{\rm out}^c$ or vice-versa. 
	We assume the former for simplicity. 
	Set $$F_{\rm in} := \{h_4|_{R_1}=\xi_1\}\quad  \text{ and } \quad F_{\rm out} := \bigcap_{i = 2}^n \{h_4|_{R_i}=\xi_i\}.$$

	Recall the double circuit event ${\rm DbCirc} = {\rm DbCirc}(\calD_{\rm out} \setminus \calD_{\rm in})$ and the annuli $\calA_1, \calA_2$ and $\calA_3$ defined in Section~\ref{sec:double_circ}.
	By Lemma~\ref{lemma:double_circ},
	\begin{align}
	\mu_{\delta\bbZ^2}\big[{\rm DbCirc}\big|  F_{\rm in} \cap F_{\rm out} \big] &\geq c_0, \quad &
	 \mu_{\delta\bbZ^2}\big[{\rm DbCirc}\big] &\geq c_0,\\
	\mu_{\delta\bbZ^2}\big[{\rm DbCirc} \big|  F_{\rm in}\big] &\geq c_0, \qquad &
	\mu_{\delta\bbZ^2}\big[{\rm DbCirc} \big|  F_{\rm out}\big] &\geq c_0,
	\end{align}
	where $c_0 > 0$ depends only on $\eta$, $M$ and $n$. 
	Thus, the bounds in \eqref{eq:quasiH1} become equivalent to 
	\begin{align}\label{eq:quasiH2}
	c_1 \leq \frac{\mu_{\delta\bbZ^2}\big[ {\rm DbCirc} \cap F_{\rm in} \cap F_{\rm out}\big]\mu_{\delta\bbZ^2}\big[ {\rm DbCirc}\big]}
	{\mu_{\delta\bbZ^2}\big[{\rm DbCirc}  \cap F_{\rm in}\big] \mu_{\delta\bbZ^2}\big[{\rm DbCirc}  \cap F_{\rm out}\big]}
	\leq c_1^{-1}.
	\end{align}

	Call $\Gamma_{\rm in}$ the inner-most circuit of $\omega^+$ in $\calA_1$ and $\Gamma_{\rm out}$ the outermost circuit  in $\calA_3$.
	For two realisations $\gamma_{\rm in}, \gamma_{\rm out}$ of these paths, write $Z^+(\gamma_{\rm in}, \gamma_{\rm out})$ 
	for the partition function of the model with $+$ boundary conditions in the annulus with inner boundary $\gamma_{\rm in}$ and outer boundary $\gamma_{\rm out}$.
	By partitioning the events in \eqref{eq:quasiH2} on the possible values of $\Gamma_{\rm in}$ and $\Gamma_{\rm out}$
	and using the short-hand $\gamma_{\rm in}$ for $\{\Gamma_{\rm in} = \gamma_{\rm in}\}$ and similarly for $\gamma_{\rm out}$, we find 
	\begin{align*}
		\frac{\mu_{\delta\bbZ^2}\big[ {\rm DbCirc} \cap F_{\rm in} \cap F_{\rm out}\big]\mu_{\delta\bbZ^2}\big[ {\rm DbCirc}\big]}
	{\mu_{\delta\bbZ^2}\big[{\rm DbCirc}  \cap F_{\rm in}\big]\mu_{\delta\bbZ^2}\big[{\rm DbCirc}  \cap F_{\rm out}\big]}
	=
		\frac{\displaystyle \sum
		\mu_{\delta\bbZ^2}\big[ \gamma_{\rm in},\,\gamma_{\rm out},\, F_{\rm in},\, F_{\rm out}\big]
		\mu_{\delta\bbZ^2}\big[\gamma_{\rm in}' ,\,\gamma_{\rm out}'\big]}
		{\displaystyle \sum
		\mu_{\delta\bbZ^2}\big[\gamma_{\rm in} ,\,\gamma_{\rm out}',\,   F_{\rm in} \big]
		\mu_{\delta\bbZ^2}\big[\gamma_{\rm in}' ,\,\gamma_{\rm out},\,  F_{\rm out} \big]} &\\
		=
	\frac{\displaystyle\sum
		\mu_{\delta\bbZ^2}\big[\gamma_{\rm in},\, F_{\rm in}\big]
		\mu_{\delta\bbZ^2}\big[\gamma_{\rm out},\, F_{\rm out}\big]
		\mu_{\delta\bbZ^2}\big[\gamma_{\rm in}'\big]
		\mu_{\delta\bbZ^2}\big[\gamma_{\rm out}' \big]
		Z^+(\gamma_{\rm in}, \gamma_{\rm out})Z^+(\gamma_{\rm in}', \gamma_{\rm out}')
		}
		{\displaystyle\sum
		\mu_{\delta\bbZ^2}\big[\gamma_{\rm in} ,\, F_{\rm in}\big]
		\mu_{\delta\bbZ^2}\big[\gamma_{\rm out}' \big]
		\mu_{\delta\bbZ^2}\big[\gamma_{\rm in}'\big]
		\mu_{\delta\bbZ^2}\big[\gamma_{\rm out},\, F_{\rm out}\big]
		Z^+(\gamma_{\rm in}, \gamma_{\rm out}')Z^+(\gamma_{\rm in}', \gamma_{\rm out})},&
		\end{align*}
		where the sums are over possible realisations $\gamma_{\rm in}, \gamma_{\rm in}'$ of $\Gamma_{\rm in}$ and $\gamma_{\rm out}, \gamma_{\rm out}'$ of $\Gamma_{\rm out}$.
		
		In the above we have used the fact that $F_{\rm in}$ is measurable in terms of the region inside $\Gamma_{\rm in}$ and 
		$F_{\rm out}$ in terms of the one outside $\Gamma_{\rm out}$. 
		To obtain the last expression, one would  need to work first in finite volume, then take the volume to infinity. 
		
		The bounds of Lemma~\ref{lem:geometric_mixing_annuli} allow us to conclude. 
		Recall that all annuli appearing in the partition functions above are of quality $\eta/4$. 
\end{proof}

\subsection{Weak mixing}
\label{sec:weak_mixing_proof}

\begin{lemma}[Weak mixing]
    \label{lemma:weak_mixing}
	For any $\epsilon > 0$ and $M\geq1$, 
	there exists a constant $c = c(\eps, M) > 0$ such that, for $\delta > 0$ small enough,
	for every simply connected  $G \subset \delta \bbZ^2$ containing $\Lambda_1$, 
	any boundary conditions $\zeta$ on $G$ with $\osc(\zeta)\leq M$, 
	and every event $E$ measurable in terms of $h_4$ on $\Lambda_{1/2}$,
    \begin{align}
        \label{eq:weak_mixing_1}
        \mu_{\delta\Z^2}[E]\leq c
        \qquad
        &
        \text{implies}
        \qquad
        \mu_{G}^\zeta[E]\leq \epsilon,
        \\
        \label{eq:weak_mixing_2}
        \mu_{G}^\zeta[E]\leq c\qquad&\text{implies}\qquad
        \mu_{\delta \Z^2}[E]\leq\epsilon
        .
    \end{align}
\end{lemma}

\begin{proof}[Proof of Lemma~\ref{lemma:weak_mixing}]
	Fix $\eps > 0$, $M\geq 1$, $G$, $\zeta$ and $E$ as in the statement. 
	We first focus on Equation~\eqref{eq:weak_mixing_1}.
    We shall prove the contraposition
    \begin{equation}
        \mu_{G}^\zeta[E]> \epsilon
        \qquad \text{implies} \qquad
        \mu_{\Z^2}[E]> c,
    \end{equation}
    where $c >0$  is a constant that will be determined and will be shown to only depend on $M$ and $\eps$. 

    Assume that $\mu_{G}^\zeta[E]> \epsilon$.
    Let $\calF$ denote the thickened flower domain explored from inside outwards between $\Lambda_{1/2}$ and $\Lambda_{3/4}$.
	By Lemma~\ref{lemma:arm_event_bound}, there exists a value $K = K(\eps)$ such that 
	\begin{align}
		\mu_{G}^\zeta\big[E \text{ and } \osc(\calF) \leq K \big]> \epsilon/2.
	\end{align}
	
	We may write 
	\begin{align}\label{eq:sum_zeta_F}
	\mu_{G}^\zeta\big[E \text{ and } \osc(\calF) \leq K \big]
	= \sum_{F, \xi} 	\mu_{G}^\zeta[h_4 = \xi \text{ on }F^c], 
	\end{align}
	where the sum runs over all realisations $F$ or $\calF$ and all realisations $\xi$ of $h_4$ on $\calF^c$ that correspond to the event on the right-hand side. 
	
	Now, due to the quasi-independence of Lemma~\ref{lemma:quasi_independence}, 
	for any $F,\xi$ as above, \im{check that the use of $\partial G$ is correct here}
	\begin{align}
	 	\mu_{G}^\zeta[h_4 = \xi \text{ on }F^c]
		=\mu_{\bbZ^2}[h_4 = \xi \text{ on }F^c \,|\, h_4 = \zeta \text{ on } \partial G]
		\geq c_0 \mu_{\delta\bbZ^2}[h_4 = \xi\text{ on }F^c],	
	\end{align}
	where $c_0 > 0$ is a constant that depends only on $K$. 
	
	Summing the above as in \eqref{eq:sum_zeta_F}, we conclude that 
	\begin{align}
		\mu_{\delta\bbZ^2}[E] \geq  c_0	\mu_{G}^\zeta\big[E \text{ and } \osc(\calF) \leq K \big] \geq \tfrac{c_0}2\eps.
	\end{align}
	This concludes the proof with $c=  \frac{c_0}2\eps$.
\end{proof}

\section{Strong mixing}\label{sec:almost_markov_proof}

For $\eta > 0$, recall the rectangles $R_1,\dots, R_4$ of Figures~\ref{fig:4rectangles2} and~\ref{fig:4rectangles}. 
For $1 \leq  i \leq 4$, let $R_i^{\rm in}$ denote the set of faces at $L^\infty$ distance at least $\eta$ from $R_i^c$. 
Let $\calF_i$ be the flower domain explored from $R_i^{\rm in}$ to $R_i$. 
Additionally, let $\calU$ denote the annulus $\Lambda_{4r} \setminus \Lambda_{4r - \eta}$. 
Write $\mathcal{E}$ for 
\begin{itemize}
	\item the realisations $F_1,\dots, F_4$ of the flower domains $\calF_1,\dots, \calF_4$ and the realisations $(\xi, \tau)$ of $h_4$ and $\omega^+$ on each $\calF_i^c$;
	\item the realisation $\xi$ of $h_4$ on $\calU$.
\end{itemize}
Note these notions depend implicitly on $\eta$.

\begin{proposition}[Strong mixing]\label{prop:strong_mixing}
	For any $\eps > 0$, there exist $\tau > 0$ and $\eta>0$ such that for all $\delta > 0$ small enough, 
	there exists a notion of successful outcome of $\mathcal{E}$ such that, 
	if ${\rm Exp}$ is successful, 
	\begin{align}\label{eq:strong_mixing0}
	\Big\|\mu_{\delta \bbZ^2} \big[ \cdot \big|\mathcal{E} = {\rm Exp}  \big] - \mu_{\delta \bbZ^2} \big[\cdot \big| \calF_i = F_i \text{ and } (h_4,\omega^+)|_{\cup_i F_i} = (\xi, \tau^+)\big]  \Big\|_{\TV,\Lambda_{2r}(x)} \leq \eps.
	\end{align}
	Moreover, $\mu_{\delta \bbZ^2} [\calE \text{ is successful } ] \geq 1-\eps$.
\end{proposition}

Above, $\|\blank\|_{\TV,\Lambda_{2r}(x)}$ is the total variation metric between the restrictions of $h_4$ to $\Lambda_{2r}(x)$.

\subsection{Preparatory Russo--Seymour--Welsh estimates}

	\begin{figure}
	\begin{center}
	\includegraphics[width = 0.65\textwidth]{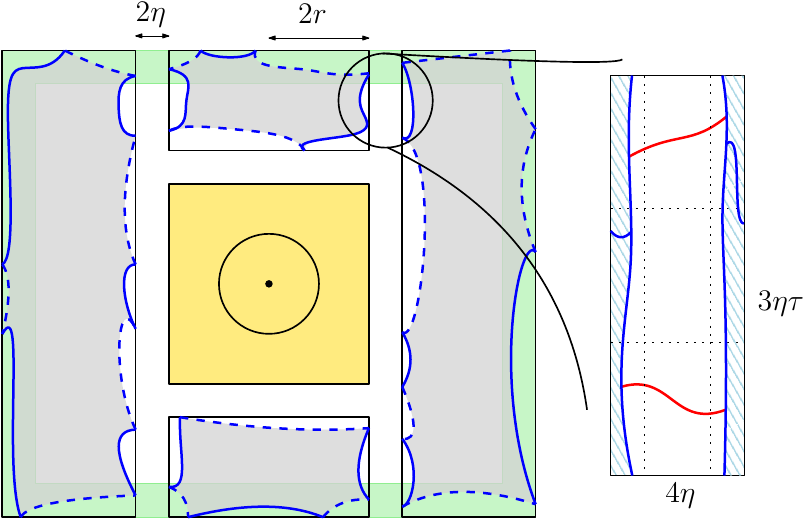}
	\caption{{\em Left:} the flower domains $\calF_i$ provide alternative primal/dual petals of $\omega^+$ (actually $\calF_i^c$ are grey; they are contained in the rectangles $R_i$). The Laplacian $\Delta_r(h)$ is computed in the orange central block.
	To completely isolate the orange block from the outside, we will also explore the height in the green annulus $\calU$. 
	{\em Right: } A tube $T_{\tau,\eta}$ contained in $S_{NE}$. It is swap-susceptible due to the blue paths of $\omega^+$, and swap-compatible if these are connected via the red paths of $\omega^+$.}
	\label{fig:4rectangles}
	\end{center}
	\end{figure}

To prepare for the proof of Proposition~\ref{prop:strong_mixing}, 
we define a set of notations and state certain consequences of  Proposition~\ref{prop:RSW+} (RSW). 
See Figure~\ref{fig:4rectangles} for relevant illustrations. 
For an aspect ratio $\tau\geq 1$, scale $\eta > 0$, and point $y\in\R^2$,
let $T_{\tau,\eta}(y)\subset\R^2$ be the vertical rectangle, or \emph{tube}, defined by
\begin{equation}
    T_{\tau,\eta}(y):=y+[-2\eta,2\eta]\times[-3\tau\eta,3\tau\eta].
\end{equation}
We slice this rectangle into four vertical strips (closed rectangles) of width $\eta$ each,
and denote the $i$-th slice by
$T_i$, from left to right.
Similarly, we slice the rectangle into three rectangles of heights $\tau \eta$, and denote the $j$-th slice by $T^j$. See Figure~\ref{fig:4rectangles}.

 A tube $T_{\tau,\eta}(y)$ is called \emph{swap-susceptible} if
\begin{itemize}
	 \item in $T_1$, there exists a connected component of $\omega^+$ connecting its top, bottom and left sides;
	\item in $T_4$, there exists a connected component of $\omega^+$ connecting its top, bottom and right sides.
\end{itemize}

A tube is called \emph{swap-compatible} if \emph{both}
 $T^1$ and $T^3$ contain an $\omega^+$-crossings 
between the left-most vertical crossing of $\omega^+$ in $T_1$ and 
the right-most  vertical crossing of $\omega^+$ in $T_4$.
See again Figure~\ref{fig:4rectangles} (right) for an illustration.

We state two standard consequences of the RSW theory for $\omega^+$. 

\begin{lemma}[Existence of swap-susceptible tubes]  \label{lem:swap_susceptible}
	For any $\tau >0$, there exists $c > 0$ such that, for any $\eta >0$ and $\delta >0$ small enough, 
	and any family of tubes $(T_{\tau,\eta}(y_i))_{1 \leq i \leq n}$ that are $\eta$ separated,
	the events $\{T_{\tau,\eta}(y_i) \text{ is swap-susceptible}\}$ dominate a family of i.i.d.~Bernoulli variables of parameter $c$. 
\end{lemma}

\begin{proof}
	All events considered are measurable in terms of $\omega^+$. 
	Thus, by Proposition~\ref{prop:FKG_omega+}, the probability that $T_{\tau,\eta}(y_i)$ is swap-susceptible, 
	given the configuration in all tubes $T_{\tau,\eta}(y_j)$ with $j\neq i$ is bounded below by 
	$$\mu_{T_{\tau,\eta}(y_i) + \Lambda_{\eta}}^- \big[ T_{\tau,\eta}(y_i) \text{ is swap-susceptible}\big], $$
	where $T_{\tau,\eta}(y_i) + \Lambda_{\eta}$ is the rectangle centred on $T_{\tau,\eta}(y_i)$, with side lengths $6\eta$ and $(6\tau +2)\eta$.
	The displayed probability is bounded below due to Propositions~\ref{prop:RSW+} and \ref{prop:FKG_omega+}.
\end{proof}

For a tube $T_{\tau,\eta}(y)$, write $\calS = \calS(T_{\tau,\eta}(y))$ for the exploration consisting of revealing $\omega^+$ in $T_{\tau,\eta}(y)$, 
left of and on the left-most vertical crossing of $T_1$ 
and right of and on the right-most vertical crossing of $T_4$. 
We will consider $\calS$ to contain the information of both the sets of edges left and right of said paths, and their states in $\omega^+$. 
Then,  $T_{\tau,\eta}(y)$ being swap-compatible is determined by $\calS$. 
When $T_{\tau,\eta}(y)$ is not swap-compatible, let $\calS$ contain the states of all edges of $\bbZ^2_\even$ in $T_{\tau,\eta}(y)$.

\begin{lemma}[Swap-susceptible tubes are swap-compatible] \label{lem:swap_compatible}
	Fix $\eps > 0$. There exists $\tau > 0$ such that the following hods for any $\eta > 0$ and $\delta  >0$ sufficiently small.
	For any tube $T_{\tau,\eta}(y)$, any realisation $S$ of $\calS$ that makes $T_{\tau,\eta}(y)$ swap-susceptible 
	and any event $B$ expressed in terms of $\omega^+$ outside of $T_{\tau,\eta}(y)$, 
	\begin{align}\label{eq:swap_compatible}
	\mu_{\delta \bbZ^2} \big[\text{$T_{\tau,\eta}(y)$ is swap-compatible} \,|\, \calS = S \text{ and } B\big] > 1- \eps.
	\end{align}	
\end{lemma}

\begin{proof}
	Fix $y$ and remove it from notation. 
	Due to the monotonicity in boundary conditions for $\omega^+$ (Proposition~\ref{prop:FKG_omega+}), we find 
	\begin{align}
	\mu_{\delta \bbZ^2} \big[\text{$T_{\tau,\eta}$ swap-compatible} \,|\, \calS = S \text{ and } B\big] 
	\geq 
	\mu_{T_{\tau,\eta}}^0 \big[\text{$T_{\tau,\eta}$ swap-compatible} \,|\, \calS = S\big] \\
	\geq
	\mu_{T_{\tau,\eta}}^{\rm alt} \big[ T^1 \text{ and }T^3 \text{ crossed horizontally in $\omega^+$}\big],
	\end{align}
	where ${\rm alt}$ are the boundary conditions on $T_{\tau,\eta}$ that are $+$ on the left and right sides and $-$ on the top and bottom sides. 
	
	It is a standard consequence of Proposition~\ref{prop:RSW+} that the above may be rendered arbitrarily close to $1$ by taking $\tau$ sufficiently large.
	\end{proof}

We now state a corollary of the above lemmas in terms of the explorations of interest to us. 
Recall the definition of the flower domains $\calF_1,\dots, \calF_4$ explored in $R_1,\dots, R_4$.
Let 
\begin{align}
S_{NE} &:= [-2r -2\eta, -2r + 2\eta]  \times [2r + \eta,4r-\eta];  \\
S_{NW} &:= [2r -2\eta, 2r + 2\eta] \times [2r + \eta,4r-\eta]; \\
S_{SE} &:= [-2r -2\eta, -2r + 2\eta]  \times [-4r + \eta, - 2r - \eta];  \\
 S_{SW} &:= [2r -2\eta, 2r + 2\eta] \times[-4r + \eta, - 2r - \eta],
\end{align}
be the vertical tubes of width $4\eta$ centred on the  spaces between the rectangles $R_i$. 
In each such tube, place $K := \lfloor \frac{2r-2\eta}{(3\tau + 2)\eta}\rfloor$ tubes $T_{\tau, \eta}(y_i)$ that are $\eta$-separated. 
\im{Might want a figure here}

Write $\calH= \calH(\tau,\eta)$ for the event that, in each $S_{*}$ for $ * \in \{NE,NW, SE, SW\}$ there exists at least one swap-susceptible tube $T_{\tau, \eta}(y_i)$ and, if we write  $T_{\tau, \eta}(y_*)$ for the outer-most swap-susceptible tube in $S_*$, 
each $T_{\tau, \eta}(y_*)$  is swap-compatible.

%

\begin{corollary}\label{cor:Hexp}
	For any $\eps > 0$, there exist $\eta > 0$ and $\tau > 0$ such that,
	\begin{align}\label{eq:Hexp}
	\mu_{\delta \bbZ^2} \big[ \calH(\tau,\eta) \big] \geq 1-\eps.
	\end{align}
\end{corollary}

\begin{proof}
	Fix $\eps> 0$ and consider a choice of $\tau,\eta > 0$ to be fixed below. 
	We simply write $T(y)$ for $T_{\tau,\eta}(y)$. 
 	Recall the exploration $\calS(T(y))$ that consists in revealing the left-most and right-most vertical crossings of $T_1$ and $T_4$, 
	and thus determine if $T(y)$ is swap-compatible. Also recall that $\calS(T(y))$ only refers to $\omega^+$. 
	
	By taking $\tau$ large enough, we may ensure \eqref{eq:swap_compatible} with the value of $\eps$ fixed above. 
	Then, by taking $\eta = \eta (r,\tau)$ small enough, we may ensure that 
	\begin{align}\label{eq:all_ss}
	\mu_{\delta \bbZ^2} \big[\text{each $S_{*}$ contains at least one swap-susceptible $T(y)$}\big] \geq 1-\eps. 
	\end{align}
	
	Apply the exploration $\calS(T(y))$ in each of the tubes $T(y)$ in $S_{NE},\dots, S_{SW}$. 
	This allows one to determine which tubes are swap-susceptible, and in particular determined $T(y_{NE}),\dots, T(y_{SW})$ when the event in \eqref{eq:all_ss} occurs. 
	Then, due to \eqref{eq:swap_compatible}, 
	\begin{align}
		\mu_{\delta \bbZ^2} \big[\text{$T_{\tau,\eta}(y_{*})$ are all swap-compatible} \,\big|\, \calS(T(y))_y \text{ s.t.  event in \eqref{eq:all_ss} occurs}\big] \geq (1-\eps)^4. 
	\end{align}
	The above, together with \eqref{eq:swap_compatible} imply the desired conclusion, after adapting $\eps$. 
\end{proof}

\subsection{Strong mixing under favourable explorations}

\begin{proposition}[Strong mixing under favourable explorations]\label{prop:strong_mixing_favor}
	For any $\eps > 0$, there exists $\tau > 0$ such that for any $\eta>0$ and $\delta > 0$ small enough, the following holds. 
	 If ${\rm Exp}$ is a realisation of $\mathcal{E}$ such that 
	 \begin{itemize}
	 	\item $\mu_{\delta \bbZ^2} \big[ \calH(\tau,\eta) \big|\mathcal{E} = {\rm Exp}  \big] \geq 1-\eps $ and
		\item $\mu_{\delta \bbZ^2} \big[ \calH(\tau,\eta) \big| \calF_i = F_i \text{ and } (h_4,\omega^+)|_{\cup_i F_i} = (\xi, \tau^+)\big]  \geq 1-\eps$, 
	 \end{itemize}
	then 
	\begin{align}\label{eq:strong_mixing_favor}
		\Big\|\mu_{\delta \bbZ^2} \big[ \cdot \big|\mathcal{E} = {\rm Exp}  \big] - 
		\mu_{\delta \bbZ^2} \big[\cdot \big| \calF_i = F_i \text{ and } (h_4,\omega^+)|_{\cup_i F_i} = (\xi, \tau^+)\big]  \Big\|_{\TV,\Lambda_{2r}(x)} 
		\leq 3\eps.
	\end{align}
\end{proposition}

\begin{proof}
	Recall the vertical splitting of $T_{\tau,\eta}$ into rectangles $T^1,T^2,T^3$ of sizes $4\eta \times\tau\eta$. 
	By taking $\tau$ sufficiently large and using Proposition~\ref{prop:RSW+}, we may ensure that
	\begin{align}\label{eq:T_2_cross}
	 	\mu_{T^2}^{\rm alt}[\Horizontal_{T^2}]^{8}  \geq 1-\eps. 
	\end{align}
	where ${\rm alt}$ are the boundary conditions on $T^{2}$ that are $+$ on its left and right sides and $-$ on the top and bottom sides. 
	
	For a configuration $\omega^+$, let
	$\Gamma_{\rm out}^{NE}$ and $\Gamma_{\rm in}^{NE}$ denote the top-most and bottom-most $\omega^+$ paths in $T(y_{NE})$
	as in the definition of swap-compatible. 
	The same notation applies in each of the other corners $NW, SE, SW$, with $\Gamma_{\rm out}^{*}$ always denoting the outer-most path. 
	Also set ${\rm Outside}$ as the region of $\Lambda_{4r}$ formed of the explored regions $\calU, \calF_1^c,\dots, \calF_4^c$
	as well as all faces in the unexplored regions that are above $\Gamma_{\rm out}^{NE}$ or $\Gamma_{\rm out}^{NW}$ or 
	below $\Gamma_{\rm out}^{SE}$ or $\Gamma_{\rm out}^{SW}$.
	Let ${\rm Inside}$ denote the unexplored region separated from infinity by $\Gamma_{\rm in}^{NE} \cup \Gamma_{\rm in}^{NW} \cup \Gamma_{\rm in}^{SE} \cup \Gamma_{\rm in}^{SW}$.
	See Figure~\ref{fig:4rectangles3}. 
	
	\begin{figure}
	\begin{center}
\includegraphics[width = 0.4\textwidth]{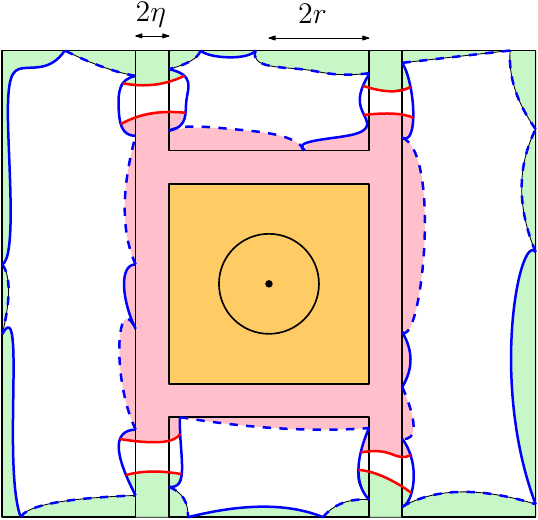}
\caption{The petals of the $\calF_i$ are in blue and the paths $\Gamma_{\rm in}^{*}$ and $\Gamma_{\rm out}^{*}$ are depicted in red. The regions ${\rm Inside}$ and ${\rm Outside}$ are pink and green, respectively. }
\label{fig:4rectangles3} 
	\end{center}
	\end{figure}
	
	We will use the short-hand $\calF = F$ for $\{  \calF_i = F_i \text{ and } (h_4,\omega^+)|_{\cup_i F_i} = (\xi, \tau^+)\}$.
	Fix two configurations $(h^1, \omega^1)$ and $(h^2, \omega^2)$ such that $\calH(\tau, \eta)$ occurs. 
	They will ultimately be sampled according to 
	$\mu_{\delta \bbZ^2} \big[ \cdot \big|\mathcal{E} = {\rm Exp}  \big]$ and $\mu_{\delta \bbZ^2} \big[\cdot \big| \calF = F\big]$, respectively.
	
	Write $(h_{\rm out}^1, \omega_{\rm out}^1)$ and $(h_{\rm in}^1, \omega_{\rm in}^1)$  
	for the restrictions of  $(h^1, \omega^1)$ to ${\rm Outside}(\omega^1)$ and ${\rm Inside}(\omega^1)$, respectively. 
	Do the same for $(h^2, \omega^2)$.
	For $i,j \in \{1,2\}$, define the short-hand 
	\begin{align*}
	&\mu\big[ (h_{\rm in}^i,  \omega_{\rm in}^i) \big|(h_{\rm out}^j,  \omega_{\rm out}^j)\big] \\
	&\qquad := \mu_{\delta\bbZ^2}\big[(h,\omega^+) = (h_{\rm in}^i,  \omega_{\rm in}^i) \text{ on } {\rm Inside}(\omega^i)\big|(h,\omega^+) = (h_{\rm out}^j,  \omega_{\rm out}^j) \text{ on } {\rm Outside}(\omega^j)\big].
	\end{align*}
	Let us now bound the cross ratio
	\begin{align}\label{eq:cross_ratio}
	\frac{\mu\big[ (h_{\rm in}^1,  \omega_{\rm in}^1) \big|(h_{\rm out}^1,  \omega_{\rm out}^1)\big]\times
	\mu\big[ (h_{\rm in}^2,  \omega_{\rm in}^2) \big|(h_{\rm out}^2,  \omega_{\rm out}^2)\big]}
	{\mu\big[ (h_{\rm in}^2,  \omega_{\rm in}^2) \big|(h_{\rm out}^1,  \omega_{\rm out}^1)\big]\times
	\mu\big[ (h_{\rm in}^1,  \omega_{\rm in}^1) \big|(h_{\rm out}^2,  \omega_{\rm out}^2)\big]}.
	\end{align}
	
	For $i,j \in \{1,2\}$, write $\calQ_{NE}^{i,j}$ for the quad formed of the faces that are
	\begin{itemize}
		\item right of the left-most $\omega^+$ crossing of $T_{\tau,\eta}(y)$,
		\item left of the right-most $\omega^+$ crossing of $T_{\tau,\eta}(y)$,
		\item below $\Gamma_{\rm out}^{NE} (\omega^j)$ and	
		\item above $\Gamma_{\rm in}^{NE}(\omega^i)$.
	\end{itemize}
	See Figure~\ref{fig:4rectangles} for an illustration. 
	Use a similar procedure to define $\calQ_{NW}^{i,j}$,$ \calQ_{SE}^{i,j}$ and $\calQ_{SW}^{i,j}$. 
	Recall the partition functions $Z^+(\calQ)$ defined in Section~\ref{subsec:geometric_mixing_annuli}, specifically those used in Lemma~\ref{lem:geometric_mixing_quadrilaterals}. 
	Then, the cross ratio in \eqref{eq:cross_ratio} becomes 
	\begin{align}\label{eq:cross_ratio_bounds}
        1- \eps	\leq \prod_{* \in \{NE,NW,SE, SW\}} \frac{Z^+(\calQ_{*}^{1,1})Z^+(\calQ_{*}^{2,2})}{Z^+(\calQ_{*}^{2,1})Z^+(\calQ_{*}^{1,2})} \leq 
             \frac{1}{1- \eps}, 
	\end{align}
	with the bounds due to Lemma~\ref{lem:geometric_mixing_quadrilaterals}, the monotonicity of boundary conditions (Proposition~\ref{prop:FKG_omega+}) and our choice of $\tau$ \eqref{eq:T_2_cross}. 
	
	Now, for any event $E$ depending on $h$ and $\omega^+$ in $\Lambda_{2r}$, we have 
	\begin{align*}
	\mu_{\delta \bbZ^2} \big[ E \cap \calH \big|\mathcal{E} = {\rm Exp}  \big]&= 
	\sum_{\substack{(h_{\rm out}^1,  \omega_{\rm out}^1 )\\
	(h_{\rm in}^1,  \omega_{\rm in }^1) \in E }}
	 \mu\big[ (h_{\rm in}^1,  \omega_{\rm in}^1) \big|(h_{\rm out}^1,  \omega_{\rm out}^1)\big]\times \mu\big[ (h_{\rm out}^1,  \omega_{\rm out}^1) \big|\mathcal{E} = {\rm Exp}  \big],
\\	 
\mu_{\delta \bbZ^2} \big[  \calH \big|\mathcal{E} = {\rm Exp}  \big] &= 
	\sum_{\substack{(h_{\rm out}^1,  \omega_{\rm out}^1) \\
	(h_{\rm in}^2,  \omega_{\rm in }^2)}}
	 \mu\big[ (h_{\rm in}^2,  \omega_{\rm in}^2) \big|(h_{\rm out}^1,  \omega_{\rm out}^1)\big]\times \mu\big[ (h_{\rm out}^1,  \omega_{\rm out}^1) \big|\mathcal{E} = {\rm Exp}  \big],
	\end{align*}
	where the sums run over configurations as featured in \eqref{eq:cross_ratio}. 
	Similar expressions exist for $\mu_{\delta \bbZ^2} \big[ E \cap \calH \big| \calF = F \big]$
	and $\mu_{\delta \bbZ^2} \big[ \calH \big| \calF = F \big]$, 
	with $ (h_{\rm out}^1,  \omega_{\rm out}^1) $ being replaced by $ (h_{\rm out}^2,  \omega_{\rm out}^2)$.
	Combining these expressions with the bounds in \eqref{eq:cross_ratio_bounds} we find
	\begin{align*}
		1 - \eps \leq &\frac{\mu_{\delta \bbZ^2} \big[ E \cap \calH \big|\mathcal{E} = {\rm Exp}  \big]
		\mu_{\delta \bbZ^2} \big[\calH \big| \calF = F  \big] }
		{\mu_{\delta \bbZ^2} \big[ \calH \big|\mathcal{E} = {\rm Exp}  \big]
		\mu_{\delta \bbZ^2} \big[ E \cap \calH \big| \calF = F  \big] }		
		\leq 
		\frac{1}{1-\eps}. 
	\end{align*}
	Finally, since we assumed 
	\begin{align}\label{eq:boundsH}
		1 - \eps \leq \mu_{\delta \bbZ^2} \big[ \calH \big|\mathcal{E} = {\rm Exp}  \big] \leq 1 \quad \text{ and }\quad 
		1 - \eps \leq \mu_{\delta \bbZ^2} \big[\calH \big| \calF = F  \big]  \leq 1,
	\end{align}
	we conclude that, 
	\begin{align*}
		(1 - \eps)^2 \leq &\frac{\mu_{\delta \bbZ^2} \big[ E \cap \calH \big|\mathcal{E} = {\rm Exp}  \big] }
		{\mu_{\delta \bbZ^2} \big[ E \cap \calH \big| \calF = F  \big] }	\leq 	
	(1-\eps)^{-2}. 
	\end{align*}
	Using again \eqref{eq:boundsH}, this leads to 
\begin{align}
\big|\mu_{\delta \bbZ^2} \big[ E  \big|\mathcal{E} = {\rm Exp}  \big] - \mu_{\delta \bbZ^2} \big[ E  \big| \calF = F  \big]	\big| \leq 3\eps,
\end{align}
as claimed. 
\end{proof}

\subsection{Proof of Proposition~\ref{prop:strong_mixing}}

\begin{proof}
	Fix $\eps > 0$. Let $\tau > 0$ be given by Proposition~\ref{prop:strong_mixing_favor}. 
	
	Let $\eta > 0$ be given by Corollary~\ref{cor:Hexp} for $\eps^2$. Then, by the tower property, \eqref{eq:Hexp} may be written as 
	\begin{align}\label{eq:Hexp_conditional}
		&\mu_{\delta \bbZ^2} \Big[	\mu_{\delta \bbZ^2} \big[ \calH(\tau,\eta)^c \big| \calE \big]\Big] \leq \eps^2 \quad \text{ and} \\
		&\mu_{\delta \bbZ^2} \Big[	\mu_{\delta \bbZ^2} \big[ \calH(\tau,\eta)^c \big| \calF_1,\dots, \calF_4 , (h_4,\omega^+)|_{\cup_i F_i}   \big]\Big] \leq \eps^2.
	\end{align}
By Markov's inequality, it follows that, with probability at least $1 - 2\eps$, $\calE$ is such that 
\begin{align}
		\mu_{\delta \bbZ^2} \big[ \calH(\tau,\eta) \big|\mathcal{E} \big] \geq 1-\eps \quad \text{ and }\quad 
 \mu_{\delta \bbZ^2} \big[ \calH(\tau,\eta) \big|  \calF_1,\dots, \calF_4 , (h_4,\omega^+)|_{\cup_i F_i}   \big]  \geq 1-\eps. 
\end{align}
	We call $\calE$ successful when the above occurs. 
	Then, for $\calE$ successful, \eqref{eq:strong_mixing_favor} implies \eqref{eq:strong_mixing0}, with $3\eps$ instead of $\eps$. 
	Modifying $\eps$ provides the desired conclusion. 
\end{proof}

\section{Proof of the spectral representation}
\label{sec:spectral_representation_proof}

The goal of this section is Proposition~\ref{lemma:spectral_representation}.

\subsection{Spectral representation of local observables}
\label{subsec:transfer_matrix_formalism}

Recall that $\mu_{\Z^2}$ denotes the full-plane law of the height
function modulo four, $h_4$ (see~\cite[Sections~14--17]{DumKozLamMan26}).
The following statements are written in the spirit of~\cite[Section~4.5]{DumKozLamMan26},
although we state everything directly in terms of the full-plane measure $\mu_{\Z^2}$.

\begin{definition}[Local observables, translations, and reflections]
    \label{def:local_observables}
    Let $\mathsf D:=\{\rightarrow,\nearrow,\uparrow,\nwarrow,\leftarrow,\swarrow,\downarrow,\searrow\}$
    be the set of eight directions.
    We set
    \begin{equation}
        \begin{aligned}
            v_\rightarrow&:=(2,0), & v_\leftarrow&:=(-2,0),\\
              v_\uparrow&:=(0,2), & v_\downarrow&:=(0,-2),\\
            v_\nearrow&:=(2,2), & v_\swarrow&:=(-2,-2), \\
            v_\nwarrow&:=(-2,2), & v_\searrow&:=(2,-2).
        \end{aligned}
    \end{equation}

    The associated reflection lines are
    \begin{equation}
        \begin{aligned}
            \calR_\rightarrow=\calR_\leftarrow&:=\{1/2\}\times\R, \\
             \calR_\nearrow=\calR_\swarrow&:=\{(x,y)\in\R^2:x+y=1\},\\
            \calR_\uparrow=\calR_\downarrow&:=\R\times\{1/2\}, \\
             \calR_\nwarrow=\calR_\searrow&:=\{(x,y)\in\R^2:x-y=0\}.
        \end{aligned}
    \end{equation}
    For $\star\in\mathsf D$, let $\dagger_\star$ denote the
    reflection with respect to $\calR_\star$, acting on faces through their centres, and set
    $h^{\dagger_\star}(f):=h(f^{\dagger_\star})$.

    A random variable $\bfX=\bfX(h)$ is called a \emph{local observable} if
    it is bounded and real-valued, and is measurable in terms of $h_4$ on
    finitely many faces.  Its support is the union of those faces.
    For $v\in\Z^2$ with $v_1+v_2$ even, write $\tau_v(f):=f+v$ for every face $f$ of $\mathbb Z^2$. The translated observable
    $\tau_v(\bfX)$ is defined by $\tau_v(\bfX)(h):=\bfX(h\circ\tau_v)$,
    so its support is translated by $v$.
    \commentPiet{14.09.2026: I don't think it's necessary to explicitly define how boundary height functions are transformer;
    this is obvious from the rest.}

    Reflection acts on observables by
    \begin{equation}
        \bfX^{\dagger_\star}(h)
        :=\bfX\bigl(h^{\dagger_\star}\bigr).
    \end{equation}

    We write $\frakA_\star^-$ and $\frakA_\star^+$ for the spaces of local
    observables supported, respectively, in the two closed half-planes on the
    negative and positive sides of $\calR_\star$, where the positive side is
    the one pointed to by $\star$.
\end{definition}

The horizontal spectral representation below is the full-plane height-modulo-four analogue of
\cite[Theorem~4.15]{DumKozLamMan26}; its transfer matrix construction is in
\cite[Sections~11--12]{DumKozLamMan26}, and the full-plane limiting procedure is \cite[Sections~14--17]{DumKozLamMan26}.
The extension to the result below is discussed in Section~\ref{sec:spectral_representation_proof}.

\begin{proposition}[Spectral representation of local observables]
    \label{lemma:spectral_representation}
    For every $\star\in\mathsf D$ and observables $\bfX\in\frakA_\star^-$ and
    $\bfY\in\frakA_\star^+$, there exists a finite complex-valued measure
    $\mu_{\bfX,\bfY,\star}$ on $[0,1]$ such that, for every
    $n\in\Z_{\geq0}$,
    \begin{equation}\label{eq:pushing_spectral_representation}
        \mu_{\Z^2}\bigl[\bfX\,\tau_{nv_\star}(\bfY)\bigr]
        =\int_{[0,1]}\lambda^n\diff\mu_{\bfX,\bfY,\star}(\lambda).
    \end{equation}
    %
    Moreover, $\mu_{\bfX,\bfX^{\dagger_\star},\star}$ and
    $\mu_{\bfY^{\dagger_\star},\bfY,\star}$ are positive measures, and
    \begin{equation}\label{eq:pushing_reflection_positivity}
        \bigl|\mu_{\Z^2}[\bfX\bfY]\bigr|^2
        \leq \lVert\mu_{\bfX,\bfY,\star}\rVert_{\mathrm{TV}}^2
        \leq
        \lVert\mu_{\bfX,\bfX^{\dagger_\star},\star}\rVert_{\mathrm{TV}}
        \lVert\mu_{\bfY^{\dagger_\star},\bfY,\star}\rVert_{\mathrm{TV}}
        =\mu_{\Z^2}[\bfX\bfX^{\dagger_\star}]
        \mu_{\Z^2}[\bfY^{\dagger_\star}\bfY].
    \end{equation}
\end{proposition}

The goal of the rest of this section is to prove this proposition.
\subsection{The horizontal transfer matrix}

In this part, given some even integer $L>0$, define the cylinder $\text{Cyl}_L:=\bbZ\times (\bbZ/L\bbZ)$ of width $L$ as being infinite in the horizontal direction.

We define a column configuration $\kappa$ of the six-vertex model as an element of $(\bbZ/4\bbZ)^{\bbZ/L\bbZ}$ which alternates parity, in the sense that for all $i\in\bbZ/L\bbZ$, one has $\kappa_{i+1}\equiv \kappa_i+1\pmod 2$. One should think of $\kappa$ as locally encoding the height function modulo $4$ on a column of the cylinder $\text{Cyl}_L$. Let us point out now the following Markov property: conditionally on two column configurations $\kappa$ and $\kappa'$ on columns $j$ and $j'$ respectively, the six-vertex configuration between these two columns is independent from that outside the two columns. Let $S_{\rightarrow}$ be the set of all column configurations $\kappa$, and let $\Omega:=\ell^2(S_\rightarrow)$ be the space on which the transfer matrix acts. We denote by $e_\kappa$ the basis vector associated to a configuration $\kappa\in S_\rightarrow$.

Define the transfer matrix $t_\rightarrow$ to be the linear map from $\Omega$ to itself which is defined on the basis $e_\kappa$ in the following way. Given $\kappa,\kappa'\in S_\rightarrow$, start by inserting $\kappa'$ to the right of $\kappa$. From here, each vertex between the two columns gets a weight:
\begin{itemize}
    \item $0$ if the four heights around the vertex do not satisfy the ice rule,
    \item $\mathbf c$ if the vertex is a $\mathbf c$-type vertex,
    \item $1$ otherwise,
\end{itemize}
and let $e_{\kappa'}^\dagger t_\rightarrow e_\kappa$ be the product of such weights. Notice that one automatically gets $e_{\kappa'}^\dagger t_\rightarrow e_\kappa=e_{\kappa}^\dagger t_\rightarrow e_{\kappa'}$, so $t_\rightarrow$ is a symmetric operator on $\Omega$.

For $j\in\bbZ$, let $E_j$ be the $j$-th column of $\text{Cyl}_L$, to which we associate a copy $\Omega_j:=\Omega$ of the vector space. Then, for all $j$, $t_\rightarrow$ can be seen as an operator from $\Omega_j$ to $\Omega_{j+1}$. For each $j\le j'$, let $E_{jj'}$ be the union of columns $\cup_{j\le k\le j'}E_k$. For two configurations $\kappa,\kappa'$ and two indices $j,j'$, we may define $\bbE_{\kappa\kappa'}^{jj'}[\cdot]$ to be the six-vertex measure inside $E_{jj'}$ conditioned on column $E_j$ having configuration $\kappa$ and column $E_{j'}$ having configuration $\kappa'$, assuming this conditioning is non-degenerate. Notice that the Markov property of the configurations ensures that additionally conditioning on any configuration outside $E_{jj'}$ does not change the measure.

For $\mathbf X$ a random variable measurable in terms of the heights modulo $4$ inside $E_{jj'}$, define the linear operator $\mathfrak o_{\mathbf X}^{jj'}:\Omega_j\to\Omega_{j'}$ in the following way. For each basis vectors $e_\kappa\in\Omega_j$ and $e_{\kappa'}\in\Omega_{j'}$, set
\begin{equation*}
    e_{\kappa'}^\dagger\mathfrak o_{\mathbf X}^{jj'}e_\kappa := \bbE_{\kappa\kappa'}^{jj'}[\mathbf X](e_{\kappa'}^\dagger t_\rightarrow^{j'-j}e_\kappa),
\end{equation*}
The measure $\bbE_{\kappa\kappa'}^{jj'}[\cdot]$ makes sense precisely when $e_{\kappa'}^\dagger t_\rightarrow^{j'-j}e_\kappa$ is non-zero, so $\mathfrak o_{\mathbf X}^{jj'}$ is well-defined. These observables follow the same properties as in \cite{DumKozLamMan26}.

\subsection{The diagonal transfer matrix}

We now wish to define a transfer matrix $t_{\nearrow}$ in the diagonal direction. The main obstacle in this case is that the naive construction of $t_\nearrow$ does not yield a symmetric operator.

For $L>0$, let us start by defining the cylinder $\text{Cyl}_L$ in the diagonal direction as the set $\bbZ^2/\sim$, where $\sim$ is the equivalence relation generated by $(x,y)\sim (x+L,y-L)$ for all $(x,y)\in\bbZ^2$. We define a diagonal configuration $(\kappa,\eps)$ to be a pair constituted of some $\kappa\in (2\bbZ/4\bbZ)^{\bbZ/L\bbZ}$ (so all entries are even), and for each $i$, a quantity $\eps_{i,i+1}\in\{-1,0,1\}$, in such a way that whenever $\eps_{i,i+1}=0$, one has $\kappa_i=\kappa_{i+1}$. One should think of $\kappa$ as encoding the height function modulo $4$ on a diagonal. At each $i$ such that $\kappa_i=\kappa_{i+1}$, the value $\eps_{i,i+1}$ encodes some additional information: either $\eps_{i,i+1}=0$, in which case this means that there is an edge of the percolation $\omega^\even$ between the two faces on which $\kappa_i$ and $\kappa_{i+1}$ lie, otherwise there is no such edge, so the two odd faces between the values $\kappa_i$ and $\kappa_{i+1}$ must have the same value, equal to $\eps_{i,i+1}$. See Figure \ref{fig:diagonal_transfer_matrix} (left) for an illustration.

We claim that this notion of diagonal configuration allows for the following Markov property, as in the horizontal case: given two diagonal configurations $(\kappa,\eps)$ and $(\kappa',\eps')$ on two diagonals $k$ and $k'$ respectively, the six-vertex configuration between the two diagonals $k$ and $k'$ is independent from the configuration outside the diagonals. This is due to the fact that the only place the inside and the outside could interact along a diagonal $k$ is at vertices where $\eps_{i,i+1}=0$ (that is, the odd heights at each side of the diagonal are undetermined), but in this case, the existence of the edge of $\omega^\even$ implies that this vertex will get weight $1$, independently of the heights neighboring it, as desired.

To a configuration $(\kappa,\eps)$, let us assign a weight
\begin{equation*}
    w(\kappa,\eps):=(\mathbf c-1)^{\#\{i:\kappa_i=\kappa_{i+1},\ \eps_{i,i+1}\neq0\}}.
\end{equation*}

Once again, set $S_{\!\nearrow}$ to be the set of all diagonal configurations $(\kappa,\eps)$, and let $\Omega:=\ell^2(S_{\!\nearrow})$ be the vector space on which the transfer matrix acts. Define $(e_{\kappa,\eps})$ to be the canonical basis of $\Omega$.  For $j\in2\bbZ/L\bbZ$, we define $E_j$ to be the $j$-th diagonal of $\text{Cyl}_L$, so that the $E_j$ form a partition of the even faces of $\text{Cyl}_L$, and we associate to it a vector space $\Omega_j$ which is a copy of $\Omega$. For each $j$, we interpret the weight $w(\cdot)$ as a diagonal operator $W:\Omega_j\to\Omega_j$ with diagonal entries $(w(\kappa,\eps))_{\kappa,\eps}$.

Let us now define the tilted transfer matrix $\tilde t_{\nearrow}$. Given two configurations $(\kappa,\eps)$ and $(\kappa',\eps')$, we define the value $e_{\kappa',\eps'}^\dagger\tilde t_{\nearrow}e_{\kappa,\eps}$ in the following way. As in Figure \ref{fig:diagonal_transfer_matrix}, place both diagonal configurations next to one another. If the values do not align then give weight $0$. Otherwise, sum over all possible heights in the diagonal in between (which are not determined at indices where $\eps_{i,i+1}=\eps'_{i,i+1}=0$, see the question marks in Figure \ref{fig:diagonal_transfer_matrix}), and for each, compute the six-vertex weights at all vertices marked by $\bullet$ in the figure. For each $j$, one can interpret $\tilde t_\nearrow$ as a linear map from $\Omega_j$ to $\Omega_{j+1}$. Finally, one defines the transfer matrix $t_\nearrow:\Omega_j\to\Omega_{j+1}$ by $t_\nearrow := W^{1/2}\tilde t_{\nearrow}W^{1/2}$, where the two $W$ act on $\Omega_{j+1}$ and $\Omega_j$ respectively.

\begin{figure}
    \centering
    \includegraphics[width=0.49\textwidth]{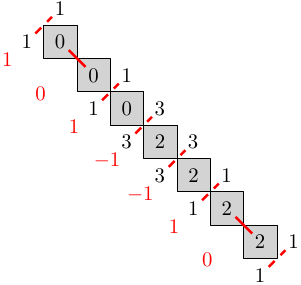}
    \includegraphics[width=0.49\textwidth]{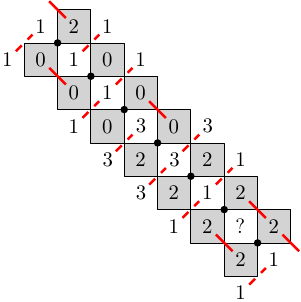}
    \caption{Left: A diagonal configuration $(\kappa,\eps)$. The data of $\kappa$ is given in the grey squares, the data of $\eps$ is given by the numbers in red, at the bottom left of each vertex between two grey squares. The data of $\eps$ is also represented by the configuration of $\omega$ (in red) between the grey squares, and corresponds to the vertices where $\eps=0$. The additional heights on the figure are those given by non-zero values of $\eps$ (recall that $-1=3$ since we are modulo $4$). This configuration has weight $(\mathbf c-1)^2$. Right: When putting two diagonal configurations next to one another, they interact on the vertices marked in black, by summing over all possible configurations of the ``$?$'' squares. The matrix element associated to these two configurations is $e_{\kappa',\eps'}^\dagger\tilde t_\nearrow e_{\kappa,\eps}=\mathbf c+\mathbf c^3$.}
\label{fig:diagonal_transfer_matrix}
\end{figure}

It is fairly easy to see that by definition, the matrix $\tilde t_\nearrow$ is symmetric, and therefore, the transfer matrix $t_\nearrow$ is also symmetric.

Let us explain why $t_\nearrow$ can be considered to be a transfer matrix. Given two indices $j\le j'$ and two configurations $(\kappa,\eps)$ on $E_j$ and $(\kappa',\eps')$ on $E_{j'}$, we may write
\begin{equation*}
    e_{\kappa',\eps'}^\dagger t_{\nearrow}^{j'-j}e_{\kappa,\eps}=w(\kappa',\eps')^{1/2}w(\kappa,\eps)^{-1/2}\cdot e_{\kappa',\eps'}^\dagger (\tilde t_{\nearrow}W)^{j'-j}e_{\kappa,\eps}.
\end{equation*}
This quantity expands into a sum over all configurations of heights between $E_j$ and $E_{j'}$ of a certain weight. One may check that every vertex gets a weight, either from the matrix $\tilde t_\nearrow$ or from $W$. In the second case, the weight is $1$ if the vertex corresponds to an edge of $(\omega^\even)^\star$, and $\mathbf c-1$ if it corresponds to an edge of $\omega^\even$. Therefore, modulo the extra boundary weights $w(\kappa',\eps')^{1/2}w(\kappa,\eps)^{1/2}$, one recovers exactly the partition function of the six-vertex model conditioned on having configurations $(\kappa,\eps)$ on $E_j$ and $(\kappa',\eps')$ on $E_{j'}$.

As in the horizontal case, define $E_{jj'}$ to be the set of faces between diagonals $E_j$ and $E_{j'}$ (note that this includes all odd faces as well, so it is not simply the union of the $E_k$ for $j\le k\le j'$). By the Markov property, given configurations $(\kappa,\eps)$ on $E_j$ and $(\kappa',\eps')$ on $E_{j'}$, there is a natural measure $\bbE_{(\kappa,\eps)(\kappa',\eps')}^{jj'}[\cdot]$ defined as the six-vertex measure conditioned on having configurations $(\kappa,\eps)$ and $(\kappa',\eps')$ on $E_j$ and $E_{j'}$ respectively. If $\mathbf X$ is a random variable measurable in terms of the heights modulo $4$ inside $E_{jj'}$, define the linear operator $\mathfrak o_{\mathbf X}^{jj'}:\Omega_j\to\Omega_{j'}$ by setting, for each configurations $(\kappa,\eps),(\kappa',\eps')$,
\begin{equation*}
    e_{\kappa',\eps'}^\dagger\mathfrak o_{\mathbf X}^{jj'}e_{\kappa,\eps}:=\bbE_{(\kappa,\eps)(\kappa',\eps')}^{jj'}[\mathbf X](e_{\kappa',\eps'}^\dagger t_{\nearrow}^{j'-j}e_{\kappa,\eps}),
\end{equation*}
which once again makes sense since the conditioning is non-degenerate precisely when $e_{\kappa',\eps'}^\dagger t_{\nearrow}^{j'-j}e_{\kappa,\eps}$ is non-zero.

\subsection{Spectral theorem and Perron-Frobenius eigenvector}

From now on, take $\star\in\{\rightarrow,\nearrow\}$. The proof works in the same way for the horizontal and diagonal transfer matrices. Since the transfer matrix $t_\star$ is symmetric, it has real eigenvalues. Furthermore, its entries are nonnegative, so it has an eigenvalue $\lambda_0>0$ such that all other eigenvalues $\lambda$ satisfy $|\lambda|\le \lambda_0$. Ideally, we would want $\lambda_0$ to be a simple eigenvalue, and for other eigenvalues $\lambda$ to satisfy the strict inequality $|\lambda|<\lambda_0$. In order to achieve this, a simple way would be for $t_\star$ to be irreducible, that is to say that for $k$ large enough, all entries of $t_\star^k$ are positive. However, this is not the case, since the transfer matrix elements $e_{\kappa'}^\dagger t_\star e_\kappa$ are zero unless $\kappa$ and $\kappa'$ have the same winding around the cylinder, i.e. if one goes around the cylinder and integrates the height function modulo $4$ into a real integer-valued height function, the height difference obtained after a full turn around the cylinder is the same in $\kappa$ and $\kappa'$. Let us note that in the diagonal case, one cannot integrate the height function off of $\kappa$ alone, but it is possible with the additional data of $\eps$.

For this reason, we restrict ourselves to the \emph{balanced} block, consisting of all configurations $\kappa$ which integrate to real height functions around the cylinder, denote by $\Omega^b$ the space generated by balanced configurations. The transfer matrix $t_\star$ acts on $\Omega^b$ and it is now irreducible, so it has a largest eigenvalue $\lambda_0$, of Perron-Frobenius eigenvector $v_0$. Denote $\lambda_0>|\lambda_1|\ge|\lambda_2|\ge\ldots$ the eigenvalues of $t_\star$, with corresponding orthogonal eigenvectors $v_0,v_1,\ldots$.

In probabilistic terms, the restriction to the balanced block corresponds to conditioning on the event $\texttt{Bal}$ that the configuration is balanced on the cylinder. The following result allows us to pass to the full-plane. The horizontal case follows from~\cite[Lemma~18.7]{DumKozLamMan26}, and the diagonal case is analogous.

\begin{theorem}
    For each $\star\in\{\rightarrow,\nearrow\}$, each even $L>0$, let $\text{Cyl}_L$ be the corresponding cylinder. Then, as $L$ tends to infinity, the measures $\mu_{\text{Cyl}_L}[\cdot|\mathtt{Bal}]$ converge, for the local convergence of measures, to $\mu_{\bbZ^2}$.
\end{theorem}

\subsection{Proof of Proposition~\ref{lemma:spectral_representation}}

We place ourselves in the corresponding setup, and set $\star\in\{\rightarrow,\nearrow\}$ to be the direction we consider. Assume that $\mathbf X\in\frakA_{\star}^-$ and $\mathbf Y\in\frakA_\star^+$. Since $\mathbf X$ and $\mathbf Y$ are local, there is some $i$ such that $\mathbf X$ depends only on the height modulo $4$ inside $E_{(-i)0}$ and $\mathbf Y$ depends only on the height modulo $4$ inside $E_{0i}$. Now, take some even $L$ large enough so that $\mathbf X$ and $\mathbf Y$ make sense as observables on $\text{Cyl}_L$, and consider the six-vertex model on $\text{Cyl}_L$.

In order to express $\bbE_{\text{Cyl}_L}[\mathbf X\cdot \tau_{nv_\star}\mathbf Y]$, one would need to first consider the model on a long torus of length $M$, and take $M$ to infinity. As in \cite{DumKozLamMan26}, the result is that at the limit, only the Perron-Frobenius vector $v_0$ survives at infinity. More precisely, this gives the following equation
\begin{equation*}
    \bbE_{\text{Cyl}_L}[\mathbf X\cdot \tau_{nv_\star}\mathbf Y]=\frac{v_0^\dagger\mathfrak o_{\mathbf Y}^{(2n)(2n+i)}t_\star^{2n}\mathfrak o_{\mathbf X}^{(-i)0} v_0}{v_0^\dagger t_\star^{2n+2i+2}v_0}.
\end{equation*}
Expanding $t_\star^{2n}=\sum_k \lambda_k^{2n}v_kv_k^\dagger$, we get (up to identifying the $v_k$ in $\Omega_0$ and $\Omega_{2n}$)
\begin{align*}
    \bbE_{\text{Cyl}_L}[\mathbf X\cdot \tau_{nv_\star}\mathbf Y]&=\lambda_0^{-2n-2i-2}\sum_k \lambda_k^{2n}v_0^\dagger\mathfrak o_{\mathbf Y}^{0i}v_kv_k^\dagger\mathfrak o_{\mathbf X}^{(-i)0} v_0\\
    &=\int \lambda^n d\mu_{\mathbf X,\mathbf Y,L,\star}(\lambda)
\end{align*}
where
\begin{equation*}
    \mu_{\mathbf X,\mathbf Y,L,\star}=\lambda_0^{-2i-2}\sum_kv_0^\dagger\mathfrak o_{\mathbf Y}^{0i}v_kv_k^\dagger\mathfrak o_{\mathbf X}^{(-i)0}v_0\cdot \delta_{\lambda_k^2/\lambda_0^2}
\end{equation*}
is a (signed) measure on $[0,1]$. In the case where $\mathbf Y=\mathbf X^\dagger$, the measure becomes
\begin{equation*}
    \mu_{\mathbf X,\mathbf X^\dagger,L,\star}=\lambda_0^{-2i-2}\sum_k(v_k^\dagger\mathfrak o_{\mathbf X}^{(-i)0}v_0)^2\cdot \delta_{\lambda_k^2/\lambda_0^2}
\end{equation*}
so it is a positive measure, of weight
\begin{align*}
    ||\mu_{\mathbf X,\mathbf X^\dagger,L,\star}||_{\text{TV}}&=\lambda_0^{-2i-2}\sum_k(v_k^\dagger\mathfrak o_{\mathbf X}^{(-i)0}v_0)^2\\
    &=\int d\mu_{\mathbf X,\mathbf X^\dagger,L,\star}=\bbE_{\text{Cyl}_L}[\mathbf X\mathbf X^\dagger].
\end{align*}
By Cauchy-Schwarz, one may write
\begin{align*}
    ||\mu_{\mathbf X,\mathbf Y,L,\star}||_{\text{TV}}^2&\le\Big(\lambda_0^{-2i-2}\sum_k\big|v_0^\dagger\mathfrak o_{\mathbf Y}^{0i}v_kv_k^\dagger\mathfrak o_{\mathbf X}^{(-i)0}v_0\big|\Big)^2\\
    &\le\Big(\lambda_0^{-2i-2}\sum_k\big(v_0^\dagger\mathfrak o_{\mathbf Y}^{0i}v_k\big)^2\Big)\Big(\lambda_0^{-2i-2}\sum_k\big(v_k^\dagger\mathfrak o_{\mathbf X}^{(-i)0}v_0\big)^2\Big)\\
    &=\bbE_{\text{Cyl}_L}[\mathbf X\mathbf X^\dagger]\bbE_{\text{Cyl}_L}[\mathbf Y\mathbf Y^\dagger].
\end{align*}

As $L$ tends to infinity, the values $\bbE_{\text{Cyl}_L}[\mathbf X\mathbf X^\dagger]$ and $\bbE_{\text{Cyl}_L}[\mathbf Y\mathbf Y^\dagger]$ converge to their full-plane counterparts. Therefore, the measures $\mu_{\mathbf X,\mathbf Y,L,\star}$ have uniformly bounded norm, and therefore one can extract convergent subsequences. \eav{Will need to check some topology stuff here but I think it should be alright.} Any subsequential limit $\mu_{\mathbf X,\mathbf Y,\star}$ satisfies, for all $n\ge 0$,
\begin{equation*}
    \int \lambda^n d\mu_{\mathbf X,\mathbf Y,\star}(\lambda)=\lim_{L\to\infty}\bbE_{\text{Cyl}_L}[\mathbf X\cdot \tau_{nv_\star}\mathbf Y]=\bbE_{\bbZ^2}[\mathbf X\cdot \tau_{nv_\star}\mathbf Y]
\end{equation*}
and
\begin{equation*}
    ||\mu_{\mathbf X,\mathbf Y,\star}||_{\text{TV}}^2\le \lim_{L\to\infty}\bbE_{\text{Cyl}_L}[\mathbf X\mathbf X^\dagger]\bbE_{\text{Cyl}_L}[\mathbf Y\mathbf Y^\dagger]=\bbE_{\bbZ^2}[\mathbf X\mathbf X^\dagger]\bbE_{\bbZ^2}[\mathbf Y\mathbf Y^\dagger].
\end{equation*}



\section{Proof of the intermediate statement (Theorem~\ref{thm:intermediate_statement})}
\label{sec:intermediate_proof}
\label{sec:pushing}

\boxstatus{Tiancheng 03.10: this section comes from HDC section 5.}

\boxstatus{Tiancheng 30.09: this section is ready.}

\subsection{Setting of the proof and the two fundamental propositions}
\label{subsec:pushing-reduction}

\boxstatus{Tiancheng 30.09: this section is ready.}

We work with fixed $\eta > 0$ and $M \geq 1$ that will be implicit in the notations. 
 Fix a mesh size $\delta >0$, a realisation of $A_1,\dots,A_4$, and realisations $\zeta_1,\dots,\zeta_4$ of $h_4$ on $A_1,\dots,A_4$.
 For $i\in\{1,2,3,4\}$, set
$
        E_i :=\{ h_4|_{A_i} = \zeta_i \}
$
and for $z\in\mathbb R^2$, let $E_i(z)$ denote the translation of $E_i$ by $z$.  
Write $\calO_\eta$ for the set of quadruples of points 
 $\mathbf z=(z_1,\dots,z_4)\in(\mathbb R^2)^4$
such that the translates $R_1+z_1,\dots,R_4+z_4$ are separated by a distance (at least) $\eta r/2$
and do not intersect $\Lambda_{2r}(0)$. In particular $(0,\dots,0) \in \mathcal O_\eta$. 
\commentTiancheng{Check the notation $B$ versus $\Lambda$, the meaning of $\mathbf z\in\Lambda_R$, and $\mu_{\bfX,\bfY}$.}
For any $\mathbf z\in \mathcal O_\eta$, write 
\[
        \mathbf E(\mathbf z):=\bigcap_{i\in\{1,2,3,4\}} E_i(z_i).
\]

\begin{figure}[htbp]
    \centering
    \includegraphics[width=\textwidth]{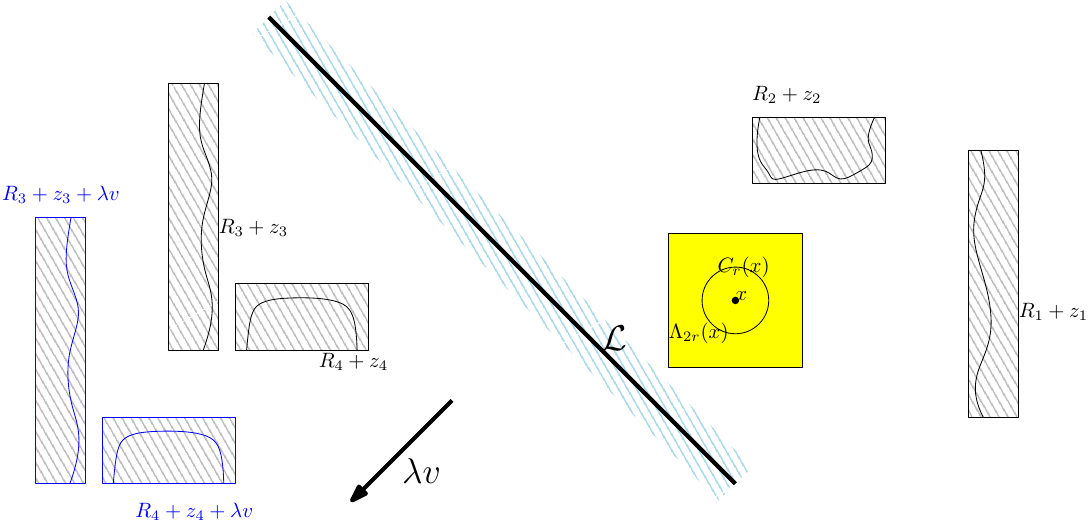}
    \caption{Here is one example of one step pushing. In this case $I=\{3,4\}$, $v=(-2,-2)$ and $\mathbf v=(0,0,v,v)$. In blue the events after pushing.}
    \label{fig:one-step-pushing}
\end{figure}
%
%
%

For $\eps,\delta_0>0$, set $\calN_\eps(\delta_0)$ to be the set of points $\mathbf z \in \calO_\eta$
such that, for any $0 < \delta < \delta_0$, 
any realisation of $(A_i)_{i\in \{1,2,3,4\}}$ and $\zeta$ with $\osc(\zeta_i) \le M$ for each $i\in\{1,2,3,4\}$, 
\begin{align}
	\Big|\mu_{\delta \bbZ^2}\left[ \Delta_r h(0)\,\middle|\, \mathbf E (\mathbf z)\right]\Big| \leq \eps.
\end{align}

Our goal is to prove that for every $\eps>0$, there exists $\delta_0$ such that $(0,\dots,0)\in \calN_\eps(\delta_0)$. The proof will rely on three results as an input. 

To describe them, the notion of separation line will be useful. 
Fix $\star\in\mathsf D$ and $\mathbf z\in\calO_\eta$, and write $v=v_\star$, with the notation of Section~\ref{subsec:transfer_matrix_formalism}. 
Below, for fix $\eta>0$, $0<r<1/2$. A line $\mathcal L$ perpendicular to $v$ is called a \emph{separation line} for $\mathbf z$ if there exists a set $I = I(\mathcal L) \subset \{1,2,3,4\}$ 
so that the support of $\Delta_r h(0)$ and $(R_i +z_i)_{i\in \{1,2,3,4\} }$ are all at a distance at least $\eta r/2$ from $\mathcal L$, 
with $(R_i +z_i)_{i\in I}$ on the positive side\footnote{The positive side of $\mathcal L$ is the half-plane whose points have larger scalar product with $v$ than those of~$\mathcal L$. The negative side is the other half-plane.} of $\mathcal L$ and the rectangles $R_i+z_i$ with $i\notin I$ and the support of $\Delta_r h(0)$ on the negative side of $\mathcal L$. 

The first proposition is essentially a rescaled version (by a factor $\delta$) of Lemma~\ref{lem:R_1}. 
\begin{proposition}[Separated conditioning]\label{prop:rp-separated-smallness}
Fix $\eta>0$, $M\ge1$, $0<r<1/2$, $R \geq 1$ and $\eps > 0$.
There exists $\delta_0  >0$ such that any $\mathbf z \in \calO_\eta\cap\Lambda_R$ \commentIoan{Not sure it's necessary that $\mathbf z$ is bounded}for which there exists a horizontal separation line $\mathcal L$ with $I(\mathcal L)=\{1,2,3,4\}$ is in $\calN_\eps(\delta_0)$. 
\end{proposition}
The proof is the same as for Lemma~\ref{lem:R_1}. We briefly recall it here for completeness.

\begin{proof}
 Fix $\eta>0$, $M\ge1$, $0<r<1/2$, $R \geq 1$ and $\eps > 0$.
Consider $(A_i)_{i\in \{1,2,3,4\}}$, $\zeta$ and $\mathbf z$ as in the statement of the proposition.
%
%
%

Equation~\eqref{eq:pushing_reflection_positivity} applied to $ \bfX:=\Delta_r h(x)$ and $ \bfY:=1_{\bf E(\mathbf z)}$ gives that 
\begin{equation}\label{eq:separated_reflection_bound}
        \left|\mu_{\delta \bbZ^2}[\Delta_r h(x)1_{\bf E(\mathbf z)} ]\right|^2
        \le
        \mu_{\delta \bbZ^2}[\Delta_r h(x) \Delta_r h^\dagger(x)]\,
        \mu_{\delta \bbZ^2}[h,h^\dagger\in\mathbf  E(\mathbf z) ].
\end{equation}
Since the height function is (approximately) harmonic in the full plane \cite[Proposition~7.2]{DumKozLamMan26}, one has
\begin{equation}\label{eq:separated_harmonicity}
\mu_{\delta \bbZ^2}[\Delta_r h(x) \Delta_r h^\dagger(x)] 
= o_\delta  (1),
\end{equation}
The four rectangles and their reflections are pairwise at distance at least $\eta r/2$ since $\mathbf z\in\calO_\eta$, each $\zeta_i$ and its reflection has oscillation at most $M$.
Lemma~\ref{lemma:quasi_independence} (and reflection invariance) applied to the eight events $E_i(z_i)$, $i=1,\dots,4$ and their reflections shows that for some constant $C=C(\eta,M,R)$, 
\begin{align}
    \mu_{\delta\bbZ^2}[h,h^\dagger\in\mathbf  E(\mathbf z)]
    &\le C\prod_{i=1}^4\mu_{\delta\bbZ^2}[E_i(z_i)]^2,\label{eq:separated_eight_events}\\
 \mu_{\delta\bbZ^2}[\mathbf  E(\mathbf z)]
    &\ge C^{-1}\prod_{i=1}^4\mu_{\delta\bbZ^2}[E_i(z_i)],\label{eq:separated_four_events}
\end{align}uniformly in $\zeta$ and $\mathbf z$ and all sufficiently small $\delta$ in the assumptions.
Together, this gives that
\begin{equation}\label{eq:separated_event_bound}
        \mu_{\delta \bbZ^2}[h,h^\dagger\in\mathbf  E(\mathbf z) ]
        \le C\,\mu_{\delta \bbZ^2}[\mathbf  E(\mathbf z) ]^2.
\end{equation}
When plugged in \eqref{eq:separated_reflection_bound}, and after dividing by  $\mu_{\delta\bbZ^2}[\mathbf  E(\mathbf z)]$, this concludes the proof.
\end{proof}

In order to bring back ourselves into the situation of the previous proposition, we use a pushing argument, in the spirit of \cite[Section 7.1]{DumKozLamMan26}.  
In \cite{DumKozLamMan26}, pushing was done for observables of the scaling limit of the field, which were shown to be analytic in the ``pushing variable''.
In this article, statements are given in the discrete. 

\begin{proposition}[Pushing perpendicularly to separation lines]\label{prop:one-step-pushing}
Fix $\eta>0$, $M\ge1$, $0<r<1/2$, $0 < a< b$ and $\eps>0$. 
Then, there exists $\delta_0$ and $\eps'$ such that the following holds true for every $\mathbf z\in\calO_\eta$ with a separation line $\calL$ perpendicular to $v$. If for all $\lambda \in (a,b)$,
$\mathbf z + \lambda \mathbf v \in  \calN_{\eps'}(\delta_0)$, where $\mathbf v=(v_1,\dots,v_4)$ with $v_i=v$ if $i\in I(\calL)$ and $v_i=0$ otherwise,  then $\mathbf z  \in  \calN_{\eps}(\delta_0)$.
\end{proposition}

We postpone the proof to Section~\ref{sec:push2} and jump directly to the proof of Theorem~\ref{thm:intermediate_statement}.
\subsection{Proof of Theorem~\ref{thm:intermediate_statement}}

Let us start by the following lemma.

\begin{lemma}\label{lem:finite-push-geometry}
	Fix $\eta > 0$. There exists $\eta' > 0$ and a sequence of lines $\mathcal L^{(1)}, \dots, \mathcal L^{(k)}$
	orthogonal to horizontal, vertical or diagonal directions $v^{(1)},\dots, v^{(k)}\in\{v_\star:\star\in\mathsf D\}$
	and values $\lambda^{(1)},\dots, \lambda^{(k)} > \eta'$ such that the following holds. 
	For any $\lambda_1,\dots, \lambda_k$ with $|\lambda_j - \lambda^{(j)}| < \eta'$, 
	if we defined recursively
	\begin{align}\label{eq:zj_def}
	\mathbf z^{(0)} = (0,\dots,0) \quad \text{ and }\quad 
	\mathbf z^{(j)} = \mathbf z^{(j-1)} + \lambda_j \mathbf v^{(j)}, 
	\end{align}
	where $\mathbf v^{(j)}_i$ is equal to $v^{(j)}$ when $i\in I(\calL^{(j)})$, and zero otherwise, then
	\begin{itemize}
	\item  
	$\mathbf z^{(k)}\in \calO_\eta\cap\Lambda_R$ and there exists a separation line $\mathcal L$ with $I(\mathcal L)=\{1,2,3,4\}$,
	\item $\mathcal L^{(j)}$ is a separation line for $\mathbf z^{(j-1)}$ for every $1\le j\le k$.
	\end{itemize}
\end{lemma}

\begin{proof}

We follow the four steps in Figure~\ref{fig:sq2}, and achieve the lemma with $k=8$. 
At each step, we choose $\mathcal L^{(j)}$ with the corresponding separation in the figures. 
\begin{enumerate}
\item 
Push $R_3$ to the left and then $R_1$ to the right (that is, $v^{(1)}=v_\leftarrow$ and $v^{(2)}=v_\rightarrow$), choose $\lambda^{(1)},\lambda^{(2)}$ large enough to allow the diagonal pushes in the next step. 
\item 
Push diagonally the translated $R_3$ to the lower left and then the translated $R_1$ to the lower right, choose $\lambda^{(3)},\lambda^{(4)}$ large enough so that both lie below $R_4$. 
\item 
Push $R_2$ upwards, choose $\lambda^{(5)}$ large enough to allow the next diagonal push. 
\item 
Push diagonally the translated $R_2$ successively to the upper right, right, and lower right, choose $\lambda^{(6)},\lambda^{(7)}$ large enough to allow the next move, and $\lambda^{(8)}$ large enough so that after the push, $R_2$ lies below $R_4$. Then we have a horizontal line separating all four $R_i+z_i$ from $\Lambda_{2r}(0)$.
\end{enumerate}
At each step, choose $\mathcal L^{(j)}$ before moving the rectangle, with a positive margin in the separation conditions. Moving the rectangle perpendicularly away from this line preserves the separation. 
After fixing the parameters $\lambda^{(1)},\dots,\lambda^{(k)}$, choose $\eta'>0$ small enough that these separations and the final horizontal separation persist whenever $|\lambda_j-\lambda^{(j)}|<\eta'$ for every $j$. This is possible because there are finitely many pushes.
\end{proof}

\begin{figure}[htbp]
\begin{center}
\setlength{\fboxsep}{4pt}
\setlength{\fboxrule}{0.3pt}
\begin{subfigure}[b]{0.30\textwidth}
    \centering
    \fcolorbox{black!40}{white}{\begin{minipage}[c][3cm][c]{\dimexpr\linewidth-2\fboxsep-2\fboxrule\relax}
        \centering
        \includegraphics[width=\linewidth,height=3cm,keepaspectratio]{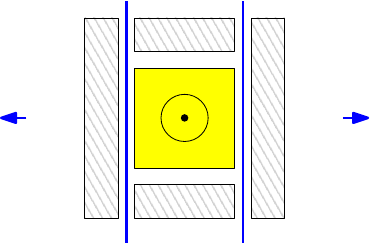}
    \end{minipage}}
    \caption{}
\end{subfigure}\hfill
\begin{subfigure}[b]{0.30\textwidth}
    \centering
    \fcolorbox{black!40}{white}{\begin{minipage}[c][3cm][c]{\dimexpr\linewidth-2\fboxsep-2\fboxrule\relax}
        \centering
        \includegraphics[width=\linewidth,height=3cm,keepaspectratio]{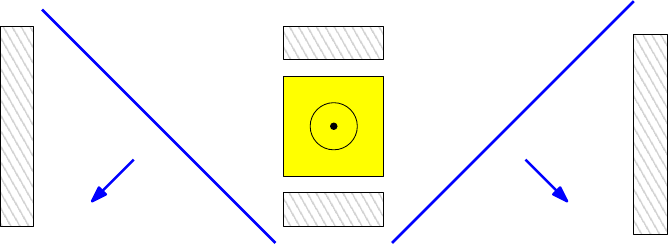}
    \end{minipage}}
    \caption{}
\end{subfigure}\hfill
\begin{subfigure}[b]{0.30\textwidth}
    \centering
    \fcolorbox{black!40}{white}{\begin{minipage}[c][3cm][c]{\dimexpr\linewidth-2\fboxsep-2\fboxrule\relax}
        \centering
        \includegraphics[width=\linewidth,height=3cm,keepaspectratio]{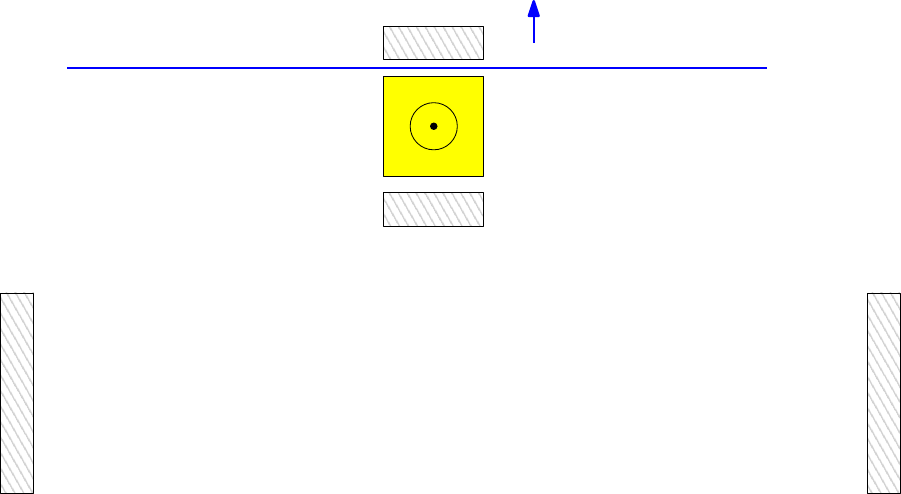}
    \end{minipage}}
    \caption{}
\end{subfigure}

\medskip
\begin{subfigure}[b]{\textwidth}
    \centering
    \fcolorbox{black!40}{white}{\includegraphics[width=\dimexpr\linewidth-2\fboxsep-2\fboxrule\relax]{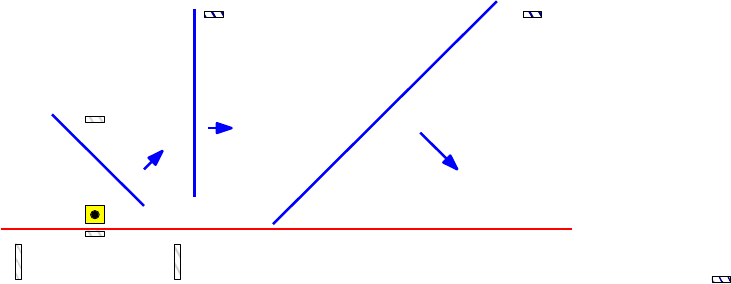}}
    \caption{}
\end{subfigure}
\caption{The four steps of the pushing procedure. Blue lines are the separating lines $\calL$ for the pushes. The red line separates the final positions of all four rectangles from $\Lambda_{2r}(0)$.}
\label{fig:sq2}
\end{center}
\end{figure}

We are now in a position to prove Theorem~\ref{thm:intermediate_statement}.

\begin{proof}[Proof of Theorem~\ref{thm:intermediate_statement}]
Fix $\eta > 0$, $0 < r <1/2$ and $M \geq 1$.
Also, set $\eps = \eps_0 > 0$. 
Consider $\mathcal L^{(1)}, \dots, \mathcal L^{(k)}$, $v^{(1)},\dots, v^{(k)}$ and $\lambda^{(1)},\dots, \lambda^{(k)}$ 
produced by Lemma~\ref{lem:finite-push-geometry}. 
Below, for $\lambda_1,\dots,\lambda_k$ with $|\lambda_j - \lambda^{(j)}| < \eta'$ for all $j$, let 
\begin{equation}
\mathbf z^{(j)}(\lambda_1,\dots,\lambda_j):=\mathbf z^{(j-1)}(\lambda_1,\dots,\lambda_{j-1})+\lambda_j\mathbf v^{(j)}.
\end{equation}

By Proposition \ref{prop:one-step-pushing}, there exists $\delta_0 > 0$ and a sequence $0 < \eps_k \leq \dots \leq \eps_1 \leq \eps_0$
such that, for any $\lambda_1,\dots, \lambda_k$ with $|\lambda_j - \lambda^{(j)}| < \eta'$ for all $1 \leq j\le k$ and every $0\le i<k$, if 
\begin{align}\label{eq:epsj+1toj}
\text{ for every }\lambda\text{ such that }|\lambda -\lambda^{(i+1)}|<\eta', \quad\mathbf z^{(i+1)}(\lambda_1,\dots,\lambda_i,\lambda) \in \calN_{\eps_{i+1}}(\delta_0)
\end{align}
then $\mathbf z^{(i)}(\lambda_1,\dots,\lambda_i) \in \calN_{\eps_i}(\delta_0)$.

Since Proposition~\ref{prop:rp-separated-smallness} implies that, by potentially decreasing $\delta_0 > 0$, 
$
\mathbf z^{(k)}(\lambda_1,\dots,\lambda_k) \in \calN_{\eps_k}(\delta_0)
$
as soon as $|\lambda_j - \lambda^{(j)}| < \eta'$ for all $1 \leq j \leq k$, we immediately conclude
that 
$
(0,\dots,0) = \mathbf z^{(0)} \in \calN_{\eps_0}(\delta_0), 
$
as claimed.
\end{proof}

\subsection{Proof of Proposition~\ref{prop:one-step-pushing}}\label{sec:push2}

\boxstatus{Tiancheng 30.09: this section is ready. Conditional on Lemma 3.5 problem.}

We wish to use the analyticity argument in a continuous limit. To formalise this, it is convenient to proceed by contradiction. 

Assume henceforth that there exist $\eta>0$, $M\ge1$, $0<r<1/2$, $0 < a< b$, $\eps>0$, a direction $v$,  $I \subset \{1,2,3,4\}$,
and a sequence of mesh sizes $(\delta_n)$ tending to zero such that the following holds. For every $n$, there exist $(A_i^{(n)})_{i\in \{1,2,3,4\}}$ and $(\zeta_i^{(n)})_{i\in\{1,2,3,4\}}$ with $\osc(\zeta_i^{(n)})\le M$ for all $i\in\{1,2,3,4\}$, 
$\mathbf z^{(n)}$, and a separation line $\mathcal L^{(n)}$ orthogonal to $v$ with $I(\calL^{(n)})=I$, 
such that 
\begin{align}\label{eq:to_contradict_push}
\lim_{n\rightarrow \infty}\sup_{\lambda \in [a,b]\cap\delta_n\mathbb Z}  \big|  \mu_{\delta_n \bbZ^2} [\Delta_r h(0)\,|\,\mathbf E^{(n)} (\mathbf z^{(n)} + \lambda \mathbf v)]  \big| &= 0, \quad \\
\inf_{n\ge1}\big| \mu_{\delta_n \bbZ^2} [\Delta_r h(0)\,|\,\mathbf E^{(n)}(\mathbf z^{(n)})]  \big| &\geq  \eps  .\label{eq:to_contradict_push_lower_bound}
\end{align}
Above, $\mathbf v$ is defined from $I$ and $v$ as usual, and $\mathbf E^{(n)}$ is the event defined in terms of $(A_i^{(n)})_{i\in \{1,2,3,4\}}$ and $(\zeta_i^{(n)})_{i\in\{1,2,3,4\}}$.


We start with a lemma that explains how to express $\mu_{\delta_n \bbZ^2} \left[\Delta_r h(0)\,\middle|\,\mathbf E^{(n)} (\mathbf z^{(n)} + \lambda \mathbf v)\right]$
for any $n$ in the sequence above.  Set
    $s_n:=\delta_n\lfloor\eta r/(8\delta_n)\rfloor$.

\begin{lemma}\label{lem:transfer_analyticity}
    For every $n$, there exists a finite complex measure $\nu_n$, supported on $[0,1]$, such that, for any $\lambda \in \delta_n\mathbb Z$ such that $\lambda>-s_n$,
    \begin{align}\label{eq:mudelta1}
           \mu_{\delta_n \bbZ^2}\!\left[ \Delta_r h(0)\,\mathbf{1}_{\mathbf E^{(n)} (\mathbf z^{(n)} + \lambda \mathbf v)}\right]&
           = \int_{[0,1]}(1-\alpha)^{(\lambda+s_n)/\delta_n}\,d\nu_n(\alpha).
    \end{align}
     Moreover,
    \begin{align}\label{eq:mudelta2}
            \|\nu_n\|_{\TV}\le C\,\mu_{\delta_n \bbZ^2} [\mathbf E^{(n)} (\mathbf z^{(n)})]
    \end{align}
    for a constant $C=C(\eta,r,M) < \infty$ that depends only on $\eta$, $r$ and $M$, but not on $n$.
\end{lemma}

\begin{proof}
We fix $n$ and omit it from notation, with the notable exception that we will still write $\nu_n$ for the measure we are aiming to construct. Recall the definition of $\widetilde\Delta_r h(0)$ and set 
\begin{align}
        \bfX&:=\widetilde\Delta_r h(0)\,\mathbf{1}_{\mathbf E_-(\mathbf z)},\\
        \bfY&:=\mathbf{1}_{\mathbf E_+(\mathbf z-s\mathbf v)}.
\end{align}
where
\begin{align}
	\mathbf E_+(\mathbf z):=\bigcap_{i\in I} E_i(z_i) \quad \text{ and } \quad 
	\mathbf E_-(\mathbf z):=\bigcap_{i\in\{1,2,3,4\} \setminus I} E_i(z_i).
\end{align}
For any $\lambda$ as above, the spatial Markov property enables us to write that
\begin{align}
        \mu_{\delta \bbZ^2}\!\left[
        \Delta_r h(0)\,
        \mathbf{1}_{\mathbf E (\mathbf z + \lambda \mathbf v)}
        \right] &=    \mu_{\delta \bbZ^2}\!\left[
        \widetilde\Delta_r h(0)\,
        \mathbf{1}_{\mathbf E (\mathbf z + \lambda \mathbf v)}
        \right] \\
        &= \mu_{\delta \bbZ^2} \!\left[ \bfX\,\tau_{(\lambda+s)v}(\bfY)  \right] \\
 &= \int_{[0,1]}(1-\alpha)^{(\lambda+s)/\delta}\,d\nu_n(\alpha), 
\end{align}
where $\nu_n$ the pushforward under $t\mapsto1-t$ of the full-plane spectral measure in \eqref{eq:pushing_spectral_representation}.
Moreover, the corresponding total-variation bound is given by~\eqref{eq:pushing_reflection_positivity}: 
\begin{equation}\label{eq:h3}
        \|\nu_n\|_{\TV}
        \le
        \mu_{\delta \bbZ^2}[\bfX \bfX^\dagger]^{1/2}\,
        \mu_{\delta \bbZ^2}[\bfY^\dagger \bfY]^{1/2}.
\end{equation}
Let us bound the right-hand side. First, Lemma \ref{lemma:quasi_independence} gives
\begin{align}  
        \mu_{\delta \bbZ^2}[\bfY \bfY^\dagger] &\leq C\, \mu_{\delta \bbZ^2}[\mathbf E_+(\mathbf z-s\mathbf v)]^2=C\, \mu_{\delta \bbZ^2}[\mathbf E_+(\mathbf z)]^2,
\end{align}
where $C$ is some constant depending on $M$ and $\eta$, but not on $\delta$. 
\commentTiancheng{Lemma 3.5 is a bit subtle here as we discussed. }

Second, Cauchy-Schwarz and the Lemmas~\ref{lemma:quasi_independence} and  \ref{lemma:L2bounds} give
\begin{align*}
    \mu_{\delta \bbZ^2}[\bfX\bfX^\dagger]&=\mu_{\delta \bbZ^2}[\mathbf E_-\cap\mathbf E_-^\dagger]\mu_{\delta \bbZ^2}[\widetilde\Delta_rh(0)(\widetilde\Delta_rh(0))^\dagger|\mathbf E_-\cap\mathbf E_-^\dagger]\\
    &\le \mu_{\delta \bbZ^2}[\mathbf E_-\cap\mathbf E_-^\dagger]\sqrt{\mu_{\delta \bbZ^2}[(\widetilde\Delta_rh(0))^2|\mathbf E_-\cap\mathbf E_-^\dagger]\mu_{\delta \bbZ^2}[[(\widetilde\Delta_rh(0))^\dagger]^2|\mathbf E_-\cap\mathbf E_-^\dagger]}\\
    &= \mu_{\delta \bbZ^2}[\mathbf E_-\cap\mathbf E_-^\dagger]\mu_{\delta \bbZ^2}[(\widetilde\Delta_rh(0))^2|\mathbf E_-\cap\mathbf E_-^\dagger]\\
    &\le C\, \mu_{\delta \bbZ^2}[\mathbf E_-]^2,
\end{align*}
where the equality is by symmetry of the model, and where once again $C$ depends on $M$ and $\eta$, but not on $\delta$. Plugging the two previous bounds in \eqref{eq:h3} and using Lemma \ref{lemma:quasi_independence} one last time gives
\begin{align}  
       \|\nu_n\|_{\rm TV}\le  \mu_{\delta \bbZ^2}[\mathbf E_+] \mu_{\delta \bbZ^2}[\mathbf E_-] & \leq C\,  \mu_{\delta \bbZ^2}[\mathbf E(\mathbf z)],
\end{align}
thus concluding the proof. 
\end{proof}

We return to the proof of Proposition~\ref{prop:one-step-pushing} and remind the reader that we are looking to contradict \eqref{eq:to_contradict_push}. 
Consider the measure $\nu_n$ obtained in the previous lemma via \eqref{eq:mudelta1}. Introduce
\begin{align}
d \widetilde\nu_n(\alpha /\delta_n) := \frac{1}{\mu_{\delta_n \bbZ^2} \left[\mathbf E^{(n)} (\mathbf z^{(n)})\right]}\, d\nu_n(\alpha)
\end{align}
and
the functions
\begin{equation*}
	\varphi_n:\lambda\mapsto \int_{\bbR_+}(1-\delta_n \alpha)^{(\lambda+s_n)/\delta_n}\,d\widetilde\nu_n(\alpha).
\end{equation*}
Now, \eqref{eq:mudelta2} implies that $\|\widetilde\nu_n\|_{\TV}\le C$ for all $n$. The functions $\varphi_n$ are therefore holomorphic and uniformly bounded on the domain $\calU=\{z\in\bbC:\Re(z)>-\eta r/16\}$. By Montel's theorem, we know that this sequence of functions is relatively compact for uniform convergence on all compacts of $\calU$. 

We reach a contradiction as follows.
Note that 
\eqref{eq:mudelta1} reads
\begin{align}\label{eq:mudelta3}
	\varphi_n(\lambda)=\frac{ \mu_{\delta_n \bbZ^2} \left[\mathbf E^{(n)} (\mathbf z^{(n)} + \lambda \mathbf v)\right]}{ \mu_{\delta_n \bbZ^2} \left[\mathbf E^{(n)} (\mathbf z^{(n)})\right]}	 \mu_{\delta_n \bbZ^2}\left[ \Delta_r h(0)\, |\, \mathbf E^{(n)} (\mathbf z^{(n)} + \lambda \mathbf v)\right].
\end{align}
Lemma~\ref{lemma:quasi_independence} and our assumption on the sequence $\zeta^{(n)}$ imply that
\begin{align}\label{eq:mixing_ratio}
	c
	\leq \frac{ \mu_{\delta_n \bbZ^2} \left[\mathbf E^{(n)} (\mathbf z^{(n)} + \lambda \mathbf v)\right]}{ \mu_{\delta_n \bbZ^2} \left[\mathbf E^{(n)} (\mathbf z^{(n)})\right]}	 
	\leq C,
\end{align}
for constants $c,C > 0$ that are uniform in $\lambda \in [0,b]\cap\delta_n\mathbb Z$ and in $n$. 
As such, \eqref{eq:to_contradict_push} yields \begin{align}\label{eq:zero_analytic}
	\lim_{n\rightarrow\infty}\sup_{\lambda\in[a,b]\cap\delta_n\mathbb Z}|\varphi_n(\lambda)|= 0.
\end{align} 
Any subsequential scaling limit of the sequence $(\varphi_n)$ therefore vanishes on $[a,b]$, so must be identically equal to zero. This means that the $(\varphi_n)$ actually converge uniformly to zero on each compact of $\calU$, and therefore $(\varphi_n(0))$ converges to zero as $n$ tends to infinity. However, this is in direct contradiction with the fact that 
\begin{align*}
	\varphi_n(0)=\mu_{\delta_n \bbZ^2}[\Delta_rh(0)|\mathbf E^{(n)}(\mathbf z^{(n)})],
\end{align*}
whose modulus remains, by \eqref{eq:to_contradict_push_lower_bound}, larger than $\eps$ for every $n$. This concludes the proof.

\section{Proof of the harmonicity theorem (Theorem~\ref{thm:harmonicity})}
\label{sec:harmonicity_proof}

\boxstatus{Ioan 05.10: NEW finished this section at 2pm}

\subsection{Proof of Theorem~\ref{thm:explored_harmo}}\label{sec:harmonicity_local}

We are now ready to explicitly define the exploration appearing in Theorem~\ref{thm:explored_harmo} and the associated notion of successful outcome. 
Recall the exploration $\calE$ of Section~\ref{sec:almost_markov_proof} consisting of 
\begin{itemize}
	\item the realisations $F_1,\dots, F_4$ of the flower domains $\calF_1,\dots, \calF_4$ in each rectangle $R_1,\dots, R_4$;
	\item the realisations $(\xi, \tau)$ of $h_4$ and $\omega^+$ on each $\calF_i^c$;
	\item the realisation $\xi$ of $h_4$ on $\calU$.
\end{itemize}
Also recall that all notions above depend on $\eta > 0$. 

Fix $\eps > 0$ and $M \geq 1$. 
We call a realisation ${\rm Exp}$ of $\calE$ $(\eps,M)$-{\em successful} if the following hold: 
\begin{itemize}
\item $\calE$ is successful in the sense of Proposition \ref{prop:strong_mixing},
\item for each $i \in \{1,\dots, 4\}$, $\osc(\xi|_{\partial F_i}) \leq M$,
\item  $\mu_{\delta\bbZ^2} [ \widetilde\Delta_r h(x)^2 \,|\,\calE = {\rm Exp}] \leq \tfrac{M}\eps$ and 
\item $\mu_{\delta\bbZ^2} [ \widetilde\Delta_r h(x)^2 \,|\, \calF_i = F_i \text{ and } (h_4,\omega^+)|_{\cup_i F_i} = (\xi, \tau^+)] \leq \tfrac{M}\eps$.
\end{itemize}

We will first prove that, for any $\eps > 0$, there exists $\eta =\eta(\eps)> 0$ and $M = M(\eta,\eps) \geq 1$ 
such that 
\begin{align}\label{eq:epsM-successful}
\mu_{\delta\bbZ^2} [ \calE \text{ is $(\eps,M)$-successful}] \geq 1 - 4\eps. 
\end{align}
We will then show that, for this notion of successful exploration, \eqref{eq:explored_harmo} holds. 
Together, these facts will prove Theorem~\ref{thm:explored_harmo}. 

\begin{proof}[Proof of Theorem~\ref{thm:explored_harmo}]
	Fix $\eps > 0$. 
	Let $\eta > 0$ be sufficiently small that
	\begin{align}
		\mu_{\delta\bbZ^2} \big[ \text{$\calE$ is successful in the sense of Proposition~\ref{prop:strong_mixing}}\big] \geq 1-\eps,
	\end{align}
	for $\delta$ small enough. 
	
	Now, Lemma~\ref{lemma:L2bounds} implies that 
	\begin{align}
		\mu_{\delta\bbZ^2} \big[ \mu_{\delta\bbZ^2} [ \widetilde\Delta_r h(x)^2 \,|\,\calE] \big] \leq C,
	\end{align}
	for some universal constant $C > 0$. Thus, for any $M \geq C/\eps^2$, we conclude by the Markov inequality that 
	\begin{align}
		\mu_{\delta\bbZ^2} \big[ \calE \text{ is such that } \mu_{\delta\bbZ^2} [ \widetilde\Delta_r h(x)^2 \,|\,\calE] \geq M/\eps \big] 
		\leq \eps.
	\end{align}
	The same holds for the measure conditioned only on the $\calF_i$. 
	
	Finally, by Lemma~\ref{lemma:arm_event_bound}, we may also choose $M $ sufficiently large that 
	\begin{align}
		\mu_{\delta\bbZ^2} [\osc(\xi|_{\partial F_i}) \leq M, \, \forall i] \geq 1 - \eps. 
	\end{align}
	The three displays above imply \eqref{eq:epsM-successful}.
	
	We now turn to proving~\eqref{eq:explored_harmo}. Fix a $(\eps,M)$-successful realisation ${\rm Exp}$ of $\calE$. 
	Let us first show that, for $\delta > 0$ small enough, 
	\begin{align}\label{eq:boundepsM}
	\Big| \mu_{\delta \bbZ^2} \big[\Delta_rh(x) \big| \calF_i = F_i \text{ and } (h_4,\omega^+)|_{\cup_i F_i} = (\xi, \tau^+)\big]  \Big| \leq {\eps}.
	\end{align}
	We may not directly apply Theorem~\ref{thm:intermediate_statement} here, as the conditioning is in terms of height gradients and percolation configurations. This is easily overcome as follows. Recall the fattened flower domain $\calF_i^+$, obtained from $\calF_i$ by exploring $h_4$ on an additional layer of faces. 
	This additional information renders the conditioning on $\omega^+$ in $\calF_{i}^c$ irrelevant for the measure in $\calF_i^+$. 
	Thus, for $\delta > 0$ sufficiently small, we have 
	\begin{align*}
	&\Big| \mu_{\delta \bbZ^2} \big[\Delta_rh(x) \big| \calF_i = F_i \text{ and } (h_4,\omega^+)|_{\cup_i F_i} = (\xi, \tau^+)\big]  \Big| \\
	&\leq \sum_{\zeta}
	\big|\mu_{\delta \bbZ^2} \big[\Delta_rh(x) \big|\, h_4|_{\cup_i F_i^+} = \zeta \big] \big| \times
	\mu_{\delta \bbZ^2} \big[h_4|_{\cup_i F_i^+} = \zeta \big| \calF_i = F_i ,\text{  } (h_4,\omega^+)|_{\cup_i F_i} = (\xi, \tau^+)\big] \\
	&\leq \sum_{\zeta}
	\eps 
	\mu_{\delta \bbZ^2} \big[h_4|_{\cup_i F_i^+} = \zeta \big| \calF_i = F_i ,\text{ } (h_4,\omega^+)|_{\cup_i F_i} = (\xi, \tau^+)\big] 
	= {\eps},
	\end{align*}
	where the sum is over all possible realisations $\zeta$ of $h_4$ on $\bigcup_i F_i^+$. 
	The bound of ${\eps}$ on the expectation of $\Delta_rh(x)$ under the conditioning on the heights in $\bigcup_iF_i^+$ 
	is due to Theorem~\ref{thm:intermediate_statement} and is valid for $\delta < \delta_0 (\eps,\eta,M)$. 
	Indeed, here $\osc(\zeta|_{\partial F_i^+}) \leq M+3$ due to the bound on the oscillations of $\xi$. 
	
	We now extract the bound on $\Delta_rh(x)$ under $\mu_{\delta \bbZ^2} \big[ \big|\, \calE={\rm Exp} \big]$.
	Let $\bbP$ be a coupling of 
	\begin{align}
	h \sim \mu_{\delta \bbZ^2} \big[ \cdot \big|\mathcal{E} = {\rm Exp}  \big] \text { and }
	h' \sim \mu_{\delta \bbZ^2} \big[\cdot \big| \calF_i = F_i \text{ and } (h_4,\omega^+)|_{\cup_i F_i} = (\xi, \tau^+)\big] 
	\end{align}
	such that $\bbP[h = h'\text{ on }\Lambda_{2r}(x)]\geq 1-\eps$. Such a coupling exists by Proposition~\ref{prop:strong_mixing}. 
	
	Then 
	\begin{align*}
	|\mu_{\delta \bbZ^2} \big[ \Delta_rh(x) \big|\, \calE={\rm Exp} \big] | 
	\leq & \big|\bbE[\Delta_rh'(x)] \big| \\
	+& \big|\bbE[\Delta_rh(x) \ind_{\{h \neq h'\text{ on }\Lambda_{2r}(x)\}}] \big|  
		+ \big|\bbE[\Delta_rh'(x) \ind_{\{h \neq h'\text{ on }\Lambda_{2r}(x)\}}] \big|.  
	\end{align*}
	The first term is bounded by \eqref{eq:boundepsM}.
	Moreover, due to our assumed bound on the second moment of $\widetilde\Delta_rh(x)$
	and the Cauchy-Swartz inequality, we find 
	\begin{align*}
	 \bbE\big[ \Delta_rh'(x) \ind_{\{h \neq h'\text{ on }\Lambda_{2r}(x)\}}\big]^2  
	 &= \bbE\big[\widetilde \Delta_rh'(x) \ind_{\{h \neq h'\text{ on }\Lambda_{2r}(x)\}}\big]^2  \\
	 & \leq \bbE\big[\widetilde \Delta_rh'(x)^2\big]\bbP\big[h \neq h' \text{ on }\Lambda_{2r}(x)\big]
	 \leq \eps. 
	\end{align*}
	The same inequality holds for the third term. 
	In conclusion 
	\begin{align*}
	|\mu_{\delta \bbZ^2} \big[ \Delta_rh(x) \big|\, \calE={\rm Exp} \big] | 
	 \leq \eps + 2\sqrt \eps. 
	\end{align*}
	This concludes the proof of Theorem~\ref{thm:explored_harmo}, with altered values of $\eps$. 
\end{proof}

\subsection{Proof of Theorem~\ref{thm:harmonicity}}

We will now see how Theorem~\ref{thm:explored_harmo} implies Theorem~\ref{thm:harmonicity}.
This is a rather simple consequence of the weak-mixing \eqref{eq:weak_mixing_1} and of the $L^2$ bound \eqref{eq:htilde}.

\begin{proof}[Proof of Theorem~\ref{thm:harmonicity}]
    Fix $K, R,\eps, r, G, \zeta$ and points $z,u_2, \dots, u_n$ as in the statement of the theorem. 
    All constants below are allowed to depend on $K$, $R$, $\eps$, $r$ and $n$, but not on $G$, $\zeta$ or the points $z,u_2,\dots, u_n$ that satisfy the requirements of the theorem. 
    
 Combining \eqref{eq:weak_mixing_1} and Theorem~\ref{thm:explored_harmo}, we find that, for any $c > 0$, if $ \delta >0$ is sufficiently small 
	\begin{align}\label{eq:explored_harmo3}
		\Big|	\mu_{G}^\zeta \big[ \Delta_rh(z) \big| \calE \text{ that is successful} \big]  \Big|< c
	\end{align}
	and 
	\begin{align}\label{eq:explored_good3}
		\mu_{G}^\zeta \big[ \text{$\calE$ is successful}\big]  > 1- c.
	\end{align}

The difference with Theorem~\ref{thm:explored_harmo} is that $\calE$ is sampled from $\mu_{G}^\zeta$ rather than the infinite volume measure. 
Equation \eqref{eq:explored_harmo3} is identical to \eqref{eq:explored_harmo}, since the conditioning on $\calE$ renders the replacement of $\mu_{\delta\bbZ^2}$ by $\mu_G^\zeta$ irrelevant. 
We will apply the above, with a value of $c$ to be decided below. 

Now, as in the proof of Theorem~\ref{thm:explored_harmo}, 
we need to bound the expectation of $ \Delta_rh(z)$ when $\calE$ fails to be successful. 
We do this using a series of $L^2$ bounds. 
We have 
\begin{align}
	\big|\Delta_{r} \mu_{G}^\zeta[h(z)h(u_2)\dots h(u_{n})] \big|  
	\leq & \Big| \mu_{G}^\zeta\big[\mu_{G}^\zeta[\Delta_rh(z) | \calE ] \cdot \widetilde h_r(u_2)\dots\widetilde h_r(u_{n}) \cdot\ind_{\calE \text{ successful}} \big] \Big| \nonumber\\
	&+ \Big| \mu_{G}^\zeta\big[\widetilde\Delta_rh (z)   \cdot \widetilde h_r(u_2)\dots \widetilde h_r(u_{n}) \cdot\ind_{\calE \text{ not successful}} \big] \Big|,\label{eq:DeltaCSfinal}
\end{align}
where the notation $\widetilde h_r(.)$ was introduced in \eqref{eq:htilde_def}. 
By the Cauchy--Schwarz inequality and \eqref{eq:explored_harmo3}, we have 
\begin{align*}
 \mu_{G}^\zeta\big[\mu_{G}^\zeta[\Delta_rh(z) | \calE ] \cdot \widetilde h_r(u_2)\cdots\widetilde h_r(u_{n}) \cdot\ind_{\calE \text{successful}} \big]^2   
 \leq c^2 \mu_{G}^\zeta\big[ \widetilde h_r(u_2)^2\cdots\widetilde h_r(u_{n})^2\big]
 \leq c^2 C,
\end{align*}
for some constant $C = C(K, R, \eps, r, n) > 0$ given by Lemma~\ref{lem:htilde}.
Also 
\begin{align*}
&\mu_{G}^\zeta\big[\widetilde\Delta_rh (z)   \cdot \widetilde h_r(u_2)\cdots \widetilde h_r(u_{n}) \cdot\ind_{\calE \text{not successful}} \big]^2\\ 
&\leq 
\mu_{G}^\zeta\big[\widetilde\Delta_rh (z)^2    \widetilde h_r(u_2)^2\cdots \widetilde h_r(u_{n})^2 \big] \mu_{G}^\zeta\big[\calE \text{ not successful}\big] 
\leq c C.
\end{align*}

Since $c >0$ may be chosen after $C$, we may choose it so that $c C \leq (\eps/2)^2$. 
Inserting the bounds above into \eqref{eq:DeltaCSfinal}, we conclude \eqref{eq:harmonicity}.
\end{proof}

\section{Boundary conditions for the height function}\label{sec:boundary_conditions_h}

\boxstatus{Tiancheng 03.10: this section comes from EA section 12.}

\boxstatus{This section is mostly done, the main thing missing are the proofs of certain lemmas which should be done in other parts. Todo: L2 bound for h}

\subsection{Statement of the result}

Let $\calD$ be a domain with $k$ marked points $a^{(1)},\ldots,a^{(k)}$ on the boundary. Let $\zeta$ be boundary conditions which are constant equal to $\zeta(i)$ on each arc $(a^{(i)}a^{(i+1)})$ (with the convention that $a^{(0)}=a^{(k)}$), and such that for each $i$, $|\zeta(i+1)-\zeta(i)|= 1$. In the case $k=0$ where there are no marked points on the boundary, denote simply by $\zeta(0)$ the constant value of $\zeta$ on the boundary.

\begin{proposition}
\label{prop:boundary_values_of_h}
    For all $r,\eta,\eps>0$, there exists $d=d(r,\eta,\eps)>0$ such that the following holds. Let $\calD\subset\delta\bbZ^2\cap \Lambda_{1/\eta}$ be a discrete domain with boundary conditions $\zeta$ as above, let $x_0,x_1,\ldots,x_\ell$ be points in $\calD$. Assume that:
    \begin{itemize}
        \item There exists some conformal map $f:\calD\to\bbD$ such that the $f(a^{(i)})$ are at least at distance $\eta$ from each other.
        \item The $x_i$ are at least at distance $r$ from each other.
        \item The point $x_0$ is at distance at most $d$ from the boundary arc $(a^{(0)}a^{(1)})$, but at distance at least $r$ from the rest of the boundary $\partial\calD\backslash(a^{(0)}a^{(1)})$ (in the case where there are no boundary points this second condition is void).
    \end{itemize}
    Then, one has
    \begin{equation*}
        \Big|\mu^\zeta_{\calD}[h(x_0)h(x_1)\cdots h(x_\ell)]-\zeta(0)\mu^\zeta_\calD[h(x_1)\cdots h(x_\ell)]\Big|\le \eps.
    \end{equation*}
\end{proposition}

\begin{remark}Let us point out that the bound on the distance between the $f(a^{(i)})$ also imposes a bound on the number $k$ of boundary points, notably $k\le \frac{2\pi}\eta$.\end{remark}

The main challenge in proving this boundary result is extending our RSW estimates to the boundary. This is the content of the following result. Note that in this result, we assume that the arc $(a^{(0)}a^{(1)})$ is formed by a path of $\omega$. 

\begin{proposition}[Boundary RSW]
\label{prop:constant_by_parts_bc_boundary_RSW}
    Take $r,\eta,d>0$ as above. There is some $c>0$ such that the following holds. Let $\calD\subset\delta\bbZ^2\cap \Lambda_{1/\eta}$ be a discrete domain with boundary conditions $\zeta$ as described above. Assume that:
    \begin{itemize}
        \item There exists some conformal map $f:\calD\to\bbD$ such that the $f(a^{(i)})$ are at least at distance $\eta$ from each other.
        \item The point $x_0$ is at distance at most $d$ from the boundary arc $(a^{(0)}a^{(1)})$, but at distance at least $r$ from the rest of the boundary $\partial\calD\backslash(a^{(0)}a^{(1)})$.
    \end{itemize}
    Then, there is some $C=C(\eta)$ depending only on $\eta$ such that
    \begin{equation*}
        \mu^\zeta_{\calD}[\text{There is no circuit of $\omega$ inside $\Lambda_r(x_0)$ of height $\zeta(0)$ surrounding $x_0$}]\le C(d/r)^c.    \end{equation*}
\end{proposition}
\eav{HUGO: shouldn't we add that the circuit is at a distance of x to have a good L2 norm}

Before proving this proposition, let us explain how it implies Proposition \ref{prop:boundary_values_of_h}.

\begin{proof}[Proof of Proposition \ref{prop:boundary_values_of_h}]
Without loss of generality, assume that $\zeta(0)=0$. By Proposition \ref{prop:constant_by_parts_bc_boundary_RSW}, there is an event $\calC$ of probability at least $1-C(\eta)(d/r)^c$ for some $C(\eta),c>0$ such that under this event, $x_0$ is surrounded by a circuit of $\omega$ inside $\Lambda_r(x_0)$ of height $0$. If this occurs, conditionally on everything outside, the expectation of $h(x_0)$ is identically zero, so
\begin{equation*}
    \mu^\zeta_{\calD}[h(x_0)h(x_1)\cdots h(x_\ell)\ind_{\calC}]=0.
\end{equation*}
Then, by Lemma \ref{lem:htilde} with $r_0$ the distance from $x_0$ to the circuit found in $\calC$ and by Cauchy-Schwarz, there is a constant $C'(r,k)$ such that
\begin{align}
    \Big|\mu^\zeta_{\calD}[h(x_0)h(x_1)\cdots h(x_\ell)\ind_{\calC^c}]\Big|&\le \mu^\zeta_{\calD}[\mu^\zeta_{\calD}[h(x_0)h(x_1)\cdots h(x_\ell)|\calF]^2]^{1/2}\mu^{\zeta}_{\calD}[\calC^c]^{1/2}\\
    &\le C'(r,k)C(\eta)^{1/2}(d/r)^{c/2},
\end{align}
where the $\sigma$-algebra $\calF$ is the set of heights necessary to determine the event $\calC$.
The right-hand side can be taken arbitrarily small by taking $d$ small enough in terms of $r$ and $\eta$, using the fact that $k$ is bounded in terms of $\eta$.
\end{proof}

The section is organised as follows. In the next section, we derive Proposition \ref{prop:constant_by_parts_bc_boundary_RSW} for constant boundary conditions (i.e.~$k=0$). Then, we explain how to interpret constant by parts boundary conditions with the help of the percolation process. Finally, we deduce Proposition \ref{prop:constant_by_parts_bc_boundary_RSW} in this setting.

\subsection{Boundary RSW for constant boundary conditions}

In the case of constant boundary conditions, the result is a direct consequence of RSW.

\begin{proof}[Proof of Proposition \ref{prop:constant_by_parts_bc_boundary_RSW} for constant boundary conditions]
    Assume that the boundary conditions $\zeta$ are constant, then they are given by a circuit of $\omega$ of height $\zeta(0)$. By FKG and RSW\eav{HUGO: should we add a precise ref?} for the measure $\mu^{\zeta(0)}_\calD$, there are some universal $c,C>0$ such that with probability at least $1-C(d/r)^c$, there is a circuit of $\omega$ inside $\Lambda_r(x_0)\backslash\Lambda_d(x_0)$. Such a circuit intersects the boundary $\partial\calD$ so it has height $\zeta(0)$.
\end{proof}

When $\zeta$ is not constant, the apparent lack of FKG makes this result harder to prove. We explain how to circumvent this in the next section. In this case, since the boundary conditions $\zeta(i)$ have alternating parity, up to adding a global constant $1$ to the heights, we will henceforth assume that each $\zeta(i)$ has the same parity as its index $i$.

\subsection{Constant by parts boundary conditions from the percolation}

The goal of this section is to explain how to interpret boundary conditions which are constant by parts as boundary conditions for the even height and percolation $(h_4^\even,\omega)$, in order to keep some estimates from RSW and FKG all the way to the boundary. There are essentially two hurdles to overcome. The first one is that although a path of $\omega$ gives weight $1$ to boundary vertices of type $\mathbf c$, this is not the case for paths of $\omega^\star$. Therefore, on the odd arcs $(a^{(i)}a^{(i+1)})$, we cannot simply express the boundary conditions using $(h_4^\even,\omega)$. Secondly, boundary arcs of different heights must not connect, and the odd boundary clusters must have certain fixed spins, which introduces additional factors.

We start with the first issue. Let $\calD$ be a domain with boundary conditions $\zeta$ which are constant by parts. For each edge $e$ between odd vertices on the boundary, let $h_e^\even\in\{0,2\}$ be a formal height modulo $4$ associated to edge $e$, denote by $h^\even_\partial$ the tuple of all $h_e^\even$. Given a value of $h^\even_\partial$, let $\zeta(h^\even_\partial)$ be the boundary conditions defined as follows (see Figure \ref{fig:cst_by_part_boundary}, Left):
\begin{enumerate}
    \item The arcs $(a^{(i)}a^{(i+1)})$ with even $i$ are formed by a path of $\omega$, with height $h_4^\even \equiv \zeta(i)\pmod 4$.
    \item The arcs $(a^{(i)}a^{(i+1)})$ with odd $i$ are formed by a path of $\omega^\star$, and outside of each edge $e$ of the boundary $(a^{(i)}a^{(i+1)})$ lies the height $h_e^\even$ modulo $4$, which interacts with the height on the inner side of $e$.
\end{enumerate}

Let us define $\mu^{\bar\zeta}_\calD$ to be the six-vertex measure conditioned on having the correct even height $h_4^\even$ on the even arcs $(a^{(i)}a^{(i+1)})$, and being of constant odd height $h^\odd$ on each of the odd arcs $(a^{(i)}a^{(i+1)})$, and with boundary weight $\mathbf c_\partial=1$. The only difference with $\mu^\zeta_\calD$ is that we do not fix the odd heights on the boundary. Let $\calE$ be the event that no two arcs $(a^{(i)}a^{(i+1)})$ of different heights are connected. Then, one can connect these different measures in the following way.

\begin{lemma}
\label{lem:constant_by_parts_bc_decomposition}
    There are values $Z_{\calD}^{\zeta(h^\even_\partial)}\ge 0$ such that the following equality of measures holds
    \begin{equation*}
        \mu_\calD^{\bar\zeta}[\bullet]=\frac1{Z_{\calD}^{\bar\zeta}}\sum_{h^\even_\partial}Z_{\calD}^{\zeta(h^\even_\partial)}\mu^{\zeta(h^\even_\partial)}_{\calD}[\bullet],
    \end{equation*}
    where $Z_\calD^{\bar\zeta}:=\sum_{h^\even_\partial}Z_{\calD}^{\zeta(h^\even_\partial)}$ is such that both sides are probability measures. Furthermore, there exists a random variable $X\in [2^{-k},1]$ such that
    \begin{equation*}
        \mu_\calD^\zeta[\bullet]=\frac{\mu^{\bar\zeta}_{\calD}[\bullet\cdot X|\calE]}{\mu^{\bar\zeta}_{\calD}[X|\calE]}.
    \end{equation*}
\end{lemma}

\begin{proof}[Proof of Lemma \ref{lem:constant_by_parts_bc_decomposition}]
Let us start by proving the first equality of measures, where the $Z_\calD^{\zeta(h^\even_\partial)}$ are the partition functions of the corresponding measures, so that $Z_\calD^{\zeta(h^\even_\partial)}\mu^{\zeta(h^\even_\partial)}_{\calD}[\bullet]$ is the weight of a configuration. The only difference between the measures on both sides is the weight on the boundary. For the height $h_4$, on the left-hand side, one has 
\begin{equation*}
    \mu_\calD^{\bar\zeta}[h_4]\propto\mathbf c^{N(h_4)},
\end{equation*}
where $N(h_4)$ denotes the number of vertices of type $5$ or $6$ in $h$ (which is measurable in terms of $h_4$).
On the right-hand side, for each edge $e$ of the boundary, writing $h_4(e)$ for the height $h_4$ on the face inside the domain and adjacent to $e$, we have\eav{HUGO: I like the notation $\ind()$. I would use that all over the paper if you agree}
\begin{align*}
    \frac1{Z_{\calD}^{\bar\zeta}}\sum_{h^\even_\partial}Z_{\calD}^{\zeta(h^\even_\partial)}\mu^{\zeta(h^\even_\partial)}_{\calD}[h_4]&=\frac1{Z_{\calD}^{\bar\zeta}}\sum_{h^\even_\partial}\mathbf c^{N(h_4)}\prod_{e\in\partial_{\text{odd}}^e\calD}(\mathbf c-1)^{\ind(h_4(e)=h^\even_\partial(e))}\\
    &=\frac{\mathbf c^{N(h_4)}}{Z_{\calD}^{\bar\zeta}}\prod_{e\in\partial_{\text{odd}}^e\calD}\sum_{\eps\in\{0,2\}}(\mathbf c-1)^{\ind(h_4(e)=\eps)}\\
    &=\mathbf c^{N(h_4)}\cdot\frac{\mathbf c^{|\partial_{\text{odd}}^e\calD|}}{Z_{\calD}^{\bar\zeta}}.
\end{align*}
where $|\partial^e_{\text{odd}}\calD|$ is the total number of edges $e$ on the odd boundary arcs of $\calD$. Since both sides are probability measures and are proportional, they must be equal.

Let us now move to the proof of the second equality of measures. The only difference between the measures $\mu^\zeta_\calD$ and $\mu^{\bar\zeta}_\calD$ is that the former requires certain heights on the boundary. Let $\Sigma_\partial$ be the event that the heights on the boundary are those given by $\zeta$, then $\mu^\zeta_\calD$ is simply the measure $\mu^{\bar\zeta}_\calD[\cdot|\Sigma_\partial]$. Under this measure, notice that the event $\calE$ is automatically satisfied (since clusters of $\omega^\star$ have constant spin). Conversely, given a configuration $(h_4^\even,\omega)$ satisfying $\calE$, Proposition \ref{prop:odd_spins_uniform} ensures that each boundary cluster of $\omega^\star$ takes the correct spin with probability $\frac12$. Therefore, if $m\le k$ is the (random) number of boundary clusters of $\omega^\star$, then taking $X=2^{-m}$ gives the desired result. See Figure \ref{fig:cst_by_part_boundary} on the right for an illustration of the event $\calE$.
\end{proof}

\begin{figure}
\centering
\includegraphics[width=0.3\textwidth]{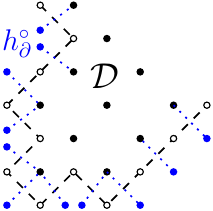}
\includegraphics[width=0.65\textwidth]{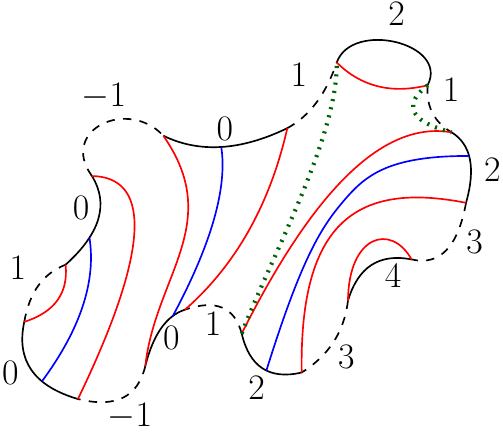}
\caption{\underline{Left:}In order to imulate boundary conditions with boundary weight $1$ on an odd path while only conditioning on $(h^\even,\omega)$, one can do the following: take a formal height $h^\even_\partial$ modulo $4$ outside $\calD$ for each medial medial vertex, and average out over all values of the spins. Then, for any configuration inside, each boundary vertex receives a weight $\frac{1+\mathbf c}2$. \newline \underline{Right:} In red, an example of interfaces of $\omega$ which satisfy $\calE$. The event $\calE$ corresponds to the existence of certain paths of $\omega$ between even parts of the boundary (in blue). Notice that conditioning on these primal paths does not necessarily determine the interfaces, since for example two of the interfaces in red can be rewired to become the dotted interfaces in green, while still being coherent with the blue paths. This makes the number of boundary clusters of $\omega^\star$ go from $m=4$ to $m=5$.}
\label{fig:cst_by_part_boundary}
\end{figure}

\subsection{Boundary RSW for constant by parts boundary conditions}

Before proving Proposition \ref{prop:constant_by_parts_bc_boundary_RSW}, let us first show the following result.

\begin{lemma}
\label{lem:diff_height_arcs_dont_connect}
    For all $\eta>0$, there is a constant $c(\eta)>0$ such that the following holds. Let $\calD$ be a domain with boundary conditions $\zeta$ as in the start of the section, and assume that there is a conformal map $f:\calD\to\bbD$ such that the $f(a^{(i)})$ are at distance at least $\eta$ from each other. Then, for all formal boundary heights $h^\even_\partial$, one has
    \begin{equation*}
        \mu^{\zeta(h^\even_\partial)}_\calD[\calE]\ge c(\eta).
    \end{equation*}
\end{lemma}

\begin{proof}
Consider the arcs $(a^{(i)}a^{(i+1)})$ with even $i$, and consider a maximal plane partition $\calP$ of these arcs, in such a way that each part of the partition contains only arcs with the same boundary value $\zeta(i)$. By planar, we mean that if even arcs $A_1,A_2,A_3,A_4$ are in this order around the boundary, then it cannot be that $A_1,A_3$ belong to some part $P\in\calP$ and arcs $A_2,A_4$ belong to part $P'\ne P$. By maximal, we mean according to the inclusion order of partitions, there may be multiple such maximal partitions.

Now, for each pair $A,A'$ of even arcs in the same part $P\in\calP$, construct a topological rectangle $\calR_{A,A'}\subset \bbD$, in such a way that
\begin{itemize}
    \item Each topological rectangle $\calR_{A,A'}$ has opposite sides $f(A)$ and $f(A')$, with conformal length bounded from above by some constant $C(\eta)$ depending only on $\eta$,
    \item Two topological rectangles corresponding to different parts do not intersect.
\end{itemize} 
The existence of these rectangles is ensured by the condition on the $f(a^{(i)})$, and by the planarity of the partition. Now, consider the event $\calE'$ that each conformal rectangle $f^{-1}(\calR_{A,A'})$ is crossed by a path of $\omega$. For each part $P\in\calP$, conditionally on the even spin and percolation outside of $\cup_{A,A'\in P}f^{-1}(\calR_{A,A'})$, Corollary \ref{prop:RSW_quads} ensures that each rectangle $f^{-1}(\calR_{A,A'})$ is crossed by a path of $\omega$ with positive probability $c(\eta)$, and since all of the $A\in P$ have the same boundary value, all these paths of $\omega$ have the same spin, so Proposition \ref{prop:FKG_omega+} ensures that all arcs of $P$ are connected with positive probability. See Figure \ref{fig:constant_by_parts_bc_conditioning} on the right for an illustration.

Now, we claim that the event $\calE'$ is contained in the event $\calE$, which would conclude the proof. To prove this, we need to show that if $B$ and $B'$ are two odd arcs with different boundary values, then there must be two even arcs $A$ and $A'$ in the same part of $\calP$ such that $A,B,A',B'$ are in this order around the boundary, see Figure \ref{fig:constant_by_parts_bc_conditioning}, Left. Assume that this is not the case. We split the boundary $\calD$ at $B$ and $B'$ into two parts $\calD_1$ and $\calD_2$. By contradiction, each part $P\in\calP$ only contains arcs in one of the $\calD_i$, and any such partitions are naturally ordered, in the sense that $P$ is \emph{inside} $P'$ if any $A\in P$ has an $A_1\in P'$ to one side of $\calD_i$ and an $A_2\in P'$ to the other side. When moving from $B$ to $B'$ along one of the $\calD_i$ and reading off the heights of the outermost parts, every height $h$ between the boundary value at $B$ and that at $B'$ must appear. Therefore, if $h$ is such a height, there is an outermost part $P_i$ in each respective $\calD_i$ which has boundary value $h$. But then by maximality, these parts $P_i$ would have to be the same (since otherwise one could replace them with the union $P_1\cup P_2$), which leads to a contradiction.
\end{proof}

\begin{figure}
\centering
\includegraphics[width=0.45\textwidth]{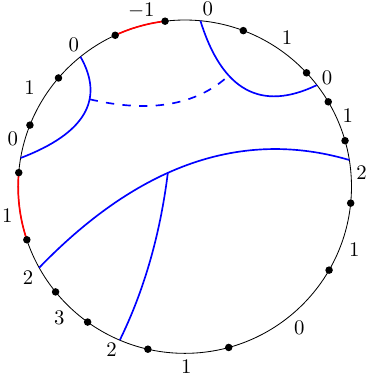}
\hspace{0.05\textwidth}
\includegraphics[width=0.45\textwidth]{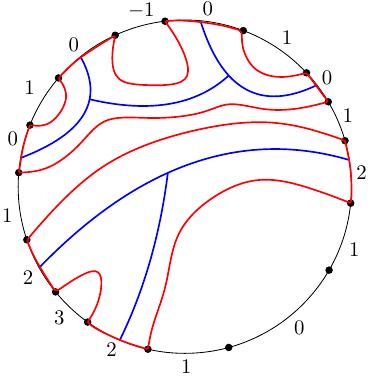}
\caption{\underline{Left:} An illustration of the argument that two odd arcs of different heights must be separated by a part $P\in\calP$. The parts are show by paths in blue, the two arcs of heights $1$ and $-1$ in red must necessarily have, on each part of the boundary between them, an outermost part of height $0$, these two parts must then be connected by maximality of $\calP$. \underline{Right:} Given the parts, one can find topological rectangles such that the sets $\cup_{A,A'\in P}\calR_{A,A'}$ (here in red) are disjoint, the existence of crossings of these rectangles implies that all arcs in the same part are connected.}
\label{fig:constant_by_parts_bc_conditioning}
\end{figure}

As a direct consequence of these lemmas, we get the following.

\begin{proof}[Proof of Proposition \ref{prop:constant_by_parts_bc_boundary_RSW} for non constant boundary conditions]
    Let $\calC$ be the event that this circuit exists. Given a configuration $h^\even_\partial$, by Proposition \ref{prop:FKG_omega+}, the measure $\mu^{\zeta(h^\even_\partial)}_{\calD}$ has RSW and FKG for the even heights and percolation $(h_4^\even,\omega)$. By RSW and FKG, as in the proof of the Proposition in the constant boundary case, the event $\calC^c$ happens with probability at most $C(d/r)^c$ under $\mu^{\zeta(h^\even_\partial)}_\calD$, for some universal constants $c,C$. By Lemma \ref{lem:diff_height_arcs_dont_connect}, the event $\calE$ happens with positive probability at least $c(\eta)>0$ under boundary conditions $\zeta(h^\even_\partial)$, so one has
    \begin{equation*}
        \mu^{\zeta(h^\even_\partial)}_{\calD}[\calC^c|\calE]\le c(\eta)^{-1}C(d/r)^c.
    \end{equation*}
    By Lemma \ref{lem:constant_by_parts_bc_decomposition}, $\mu^\zeta_\calD$ is an average of the measures $\mu^{\zeta(h^\even_\partial)}_{\calD}[\cdot|\calE]$, so the result follows.
\end{proof}

\section{Proof of GFF convergence (Theorem~\ref{thm:even})}\label{sec:GFF_convergence}

\boxstatus{Tiancheng 03.10: this section comes from EA section 13.}

\boxstatus{This section is mostly done in terms of content, depending on where certain lemmas are proved. Todo: Merging of points in correlations}

Before diving into the proof of Theorem \ref{thm:even}, we prove certain qualitative results on the correlation function.

\subsection{Hölder continuity of correlation functions}

Let us first prove that the correlation functions are Hölder continuous away from the boundary.

\begin{proposition}
\label{prop:correlations_holder}
    There exists a constant $c>0$, and for all positive $M$ a constant $C=C(M)>0$ such that the following holds. Take $\eps,r>0$, and let $\delta>0$ be small enough in terms of $\eps,r,M$. Let $\calD\subset\delta\bbZ^2$ be a discrete domain with boundary conditions $\zeta$ such that $|\zeta|\le M$. Assume that $x_1,x_1',x_2,x_3,\ldots,x_k$ are points in $\calD$. Assume $x_1$ and $x_1'$ are at distance at most $\eps$ from each other, and assume that all of the $x_i$ are at distance at least $r$ from each other, and that $x_1$ is at distance at least $r$ from $\partial\calD$. Then,
    \begin{equation*}
        \Big|\mu^\zeta_\calD[h(x_1)h(x_2)\cdots h(x_k)]-\mu^\zeta_\calD[h(x_1')h(x_2)\cdots h(x_k)]\Big|\le C(\eps/r)^c.
    \end{equation*}
\end{proposition}

\begin{proof}
    We rewrite the left-hand side as the expectation $\mu^\zeta_{\calD}[(h(x_1)-h(x_1'))h(x_2)\cdots h(x_k)]$. Assume that $r\ge 4\eps$, and consider the event $\calC$ that there is a circuit of $\omega$ in the annulus $\Lambda_{r/2}(x_1)\backslash\Lambda_{2\eps}(x_1)$ going around $\Lambda_{2\eps}(x_1)$. On the event $\calC$, conditionally on the heights $h(x_2),\ldots, h(x_k)$ which depend only on faces outside $\Lambda_{r/2}(x_1)$, the expected values of $h(x_1)$ and $h(x_1')$ are both equal to the height of the circuit of $\omega$. Therefore,
    \begin{equation*}
        \mu^\zeta_{\calD}[(h(x_1)-h(x_1'))h(x_2)\cdots h(x_k)\ind_\calC]=0.
    \end{equation*}
    Since the box $\Lambda_{r/2}(x_1)$ does not intersect the boundary $\partial\calD$, by RSW \eav{ref HUGO: indeed}, the event $\calC$ occurs in $\mu^\zeta_\calD$ with probability at least $1-C(\eps/r)^c$ for some universal constant $c$ and some $C=C(M)>0$. The contribution of the event $\calC^c$ is then controlled by Cauchy-Schwarz and Lemma \ref{lem:htilde}. \eav{L2 lemma for h. HUGO Being careful like in derivation of Prop 9.1 from 9.2}.
\end{proof}

\subsection{Correlation functions when merging points}

In this section, we show the following result which explains how the correlation functions behave in the limit when two points get close. Note that the points $1$ and $2$ chosen in the statement are completely arbitrary and may be replaced by any other non-zero points of 
$\bbR^2$ up to changing the constant $C$.

\begin{proposition}[Merging of two-point correlations]
\label{prop:two_point_merging}
	There exists a universal constant $C>0$ with the following properties.
	Let $r\geq 1$ and let $\calD\subset\bbZ^2$ be a simply connected domain with
	$\Lambda_r\subset\calD\subset\Lambda_{2r}$.
	Let $\zeta$ be any admissible boundary condition on $\calD$ with
	$|\zeta|\leq M\in\Z_{\geq 1}$.
	Then for any $x\in [-r/2,r/2]^2$, we have
	\begin{equation}
		\Big|\mu_{\calD}^\zeta[h(0)h(x)]-\mu_{\Z^2}[(h(0)-h(re_1))(h(x)-h(2re_1))]\Big|\le C\cdot M^2.
	\end{equation}
\end{proposition}

\begin{proof}
	Recall Equation~(404) in~\cite{DumKozLamMan26}.
	Both expectations can be calculated as follows:
	\begin{itemize}
		\item Explore the largest $\omega^+$-circuit $\gamma$ surrounding $0$ and $x$,
		\begin{itemize}
			\item Within $\calD$ (in the fixed boundary case), and
			\item Not surrounding $re_1$ and $2re_1$ (in the full-plane case).
		\end{itemize}
		Then both expectations are equal up to an additive term of order $O(M^2)$,
		plus $X_{0,x}(\calA)$, where $\calA$ is the domain bounded by $\gamma$.
	\end{itemize}
	To prove the proposition, it suffices to prove that
	\begin{equation}
		\mu_{\calD}^\zeta[X_{0,x}(\calA)]-\mu_{\Z^2}[X_{0,x}(\calA)]=O(M).
	\end{equation}
	This is straightforward from the RSW estimates and the Lipschitz property of $X_{0,x}(\calA)$,
	established in~\cite[Equation~(406)]{DumKozLamMan26}.
\end{proof}




\subsection{Tightness of the correlation functions}

\eav{There is the set $\calD_k$ which is an unfortunate notation since $\calD$ is our domain. I will replace it by $D_k$ for readability for now.}

In order to work with continuous correlation functions instead of discrete ones, let us define $\bar h$ to be the averaged out height function, defined as follows. If $\delta$ is the mesh size of the lattice under consideration, for $x\in\bbR^2$, set $\bar h(x) := \oint_{\partial B_\delta(x)}h(y)dy$ the average of $h$ on a circle of radius $\delta$ around $x$. The main advantage of $\bar h$ is that it is automatically Lipschitz, thus continuous. 

Throughout the rest of the proof, let $\calD_\delta$ be a sequence of domains in $\delta\bbZ^2$ converging to some domain $\calD\subset\bbR^2$ as in Theorem \ref{thm:even}. We will use all the notations of the theorem, including the marked points $a^{(i)}_\delta$ and boundary conditions $\zeta_\delta(i)$. Throughout this part, we will set
\begin{equation*}
    D_k(\calD):=\{(u_1,\ldots,u_k)\in \calD^k|\forall i\ne j, u_i\ne u_j\}
\end{equation*}
to be the set $D_k$ of the introduction, defined for the continuum domain $\calD$. Let us point out that an exhaustion of compacts of $D_k(\calD)$ are the compacts
\begin{equation*}
    K_\eps:=\{(u_1,\ldots,u_k)\in\calD^k|\forall i\ne j, \text{dist}(u_i,u_j)>\eps\text{ and }\forall i, \text{dist}(u_i,\partial\calD)>\eps\}.
\end{equation*}

The following result is an immediate consequence of the Hölder regularity of Proposition~\ref{prop:correlations_holder}, and thus we omit its proof.

\begin{proposition}[Tightness of the correlation functions]
\label{prop:correlation_function_compactness}
    Let $\calD_\delta$ be a sequence of domains in $\delta\bbZ^2$ converging to some domain $\calD\subset\bbR^2$ in the sense of Theorem \ref{thm:even}. Let $k$ be a positive integer. For each $\delta>0$, define the $k$-point correlation function
    \begin{equation*}
        \Phi_k^{\delta}:(x_1,\ldots,x_k)\in(\bbR^2)^k\mapsto \mu^{\zeta_\delta}_{\calD_\delta}[\bar h(x_1)\bar h(x_2)\cdots \bar h(x_k)]
    \end{equation*}
   Then, for every compact $K$ of $D_k(\calD)$, there is some $\delta_0(K)>0$ such that the functions $((\Phi_k^\delta)_{|K})_{\delta<\delta_0}$ are well defined and form a precompact set in $(\calC^0(K),||\cdot||_\infty)$.
\end{proposition}

\begin{remark}
\label{rem:smoothing_correlations_changes_by_little}
Let us mention that by Proposition \ref{prop:correlations_holder}, for any compact $K$ of $D_k(\calD)$, there is a constant $C(K)$ such that for all $\delta$ small enough and $(x_1,\ldots,x_k)\in K$,
\begin{equation*}
    \Big|\Phi_k^{\delta}(x_1,\ldots,x_k)-\mu^{\zeta_\delta}_{\calD_\delta}[h(x_1)h(x_2)\cdots h(x_k)]\Big|\le C(K)\delta^c.
\end{equation*}
Therefore, any result on the uniform convergence of $\Phi_k^\delta$ on all compacts of $D_k(\calD)$ implies the same convergence result for the correlation functions $\mu^{\zeta_\delta}_{\calD_\delta}[h(x_1)h(x_2)\cdots h(x_k)]$. We shall work with the former due to the tightness result.
\end{remark}

\subsection{Identification of the limit}

We are now ready to show the convergence of the correlation functions. If $\calD\subset\bbR^2$ is a domain, let $\Gamma_\calD$ be a Dirichlet boundary conditions GFF with variance $1$. If $\calD$ has marked points $a_1,\ldots,a_k$ on its boundary with certain boundary conditions $\zeta(i)$, let $H_\zeta$ be the harmonic extension of these boundary conditions to the domain $\calD$.

%

\begin{proof}[Proof of Theorem \ref{thm:even}]
Identifying the $k$-point correlations and Remark \ref{rem:smoothing_correlations_changes_by_little} give the result of Theorem \ref{thm:even} in the sense of convergence of correlation functions. Obtaining the convergence in distribution of $h$ is then a direct consequence of this result and of the Hölder continuity of Proposition \ref{prop:correlations_holder}. The proof of this result is completely analogous to that in \cite[Section~8]{DumKozLamMan26}.

We therefore focus on identifying the $k$-point correlations. By Proposition \ref{prop:correlation_function_compactness}, the functions $\Phi_k^\delta$ are precompact when restricted to any compact subset $K$ of $D_k(\calD)$. Therefore, we may extract subsequential limits $\Psi_k$ in the topology of uniform convergence in any compact subset of $D_k(\calD)$. It remains to show that all of the $\Psi_k$ are unique and given by the correlation functions of the desired GFF. Let us notice that for all $k\ge 1$, the function $\Psi_k$ satisfies the following properties:
\begin{enumerate}
    \item By Theorem \ref{thm:harmonicity}, $\Psi_k$ is harmonic in each variable inside $D_k(\calD)$.
    \item By Proposition \ref{prop:boundary_values_of_h}, as $x_1$ converges to a point on the boundary arc $(a^{(i)}a^{(i+1)})$ (or to any point of the boundary if there are no marked points), one has
    \begin{equation*}
        \Psi_k(x_1,x_2,\ldots,x_k)\xrightarrow{x_1\to (a^{(i)}a^{(i+1)})} \zeta(i)\Psi_{k-1}(x_2,\ldots,x_k).
    \end{equation*}
    \item By Proposition \ref{prop:two_point_merging}, assuming that $k\ge 2$, given fixed points $(x_2,\ldots,x_k)\in D_{k-1}(\calD)$, there is a constant $C(x_2,\ldots,x_k)>0$ such that for all $x_1$ close enough to $x_2$, one has
    \begin{equation*}
        \Big|\Psi_k(x_1,\ldots,x_k)-\sigma^2G_{\bbR^2}(x_1,x_2)\Psi_{k-2}(x_3,\ldots,x_k)\Big|\le C(x_2,\ldots,x_k),
    \end{equation*}
    where $G_{\bbR^2}(\cdot,\cdot)$ is the full-plane Green's function, and where we set $\Psi_0:=1$.
	Indeed, this follows by exploring a flower domain around $x_1$ and $x_2$, using that proposition,
	and noticing that the conditional expectation of $h(x_1)h(x_2)$ is given by the full-plane plus an error term of the order of the maximum square height on the flower domain boundary.
\end{enumerate}
We now show by induction on $k$ that this characterises the $\Psi_k$ as the correlation functions $\Psi_k^{\text{GFF}}$ of $\sigma\cdot \Gamma_\calD+H_{\zeta}$. When $k=0$, we immediately have $\Psi_0=1=\Psi_0^{\text{GFF}}$. Let us move to the case $k=1$. In this case, $\Psi_1$ is a harmonic function on $\calD$, whose boundary conditions on each arc $(a^{(i)}a^{(i+1)})$ are constant by parts and given by $\zeta$. Furthermore, the function $\Psi_1$ is bounded between the smallest and largest values of $\zeta$, so it stays bounded around the marked points $a^{(i)}$. Therefore, $\Psi_1$ is a harmonic function with boundary values constant by parts given by $\zeta$, so $\Psi_1=H_\zeta=\Psi_1^{\text{GFF}}$.

Assume now that $k\ge 2$ and that for all $\ell<k$, we have $\Psi_\ell=\Psi_\ell^{\text{GFF}}$. Fix $(x_2,\ldots,x_k)\in D_{k-1}(\calD)$, and consider the function
\begin{equation*}
    F(x_1):= \Psi_k(x_1,\ldots,x_k)-\Psi_k^{\text{GFF}}(x_1,\ldots,x_k).
\end{equation*}
Then, the function $F$ is harmonic in $\calD\backslash\{x_2,\ldots,x_k\}$. Consider what happens when $x_1$ converges to one of the other points, without loss of generality assume that it is $x_2$. Then, as $x_1$ converges to $x_2$, we may use the merging result, as well as the analogous result for the continuum GFF to write
\begin{equation*}
    F(x_1)=\sigma^2G_{\bbR^2}(x_1,x_2)\Psi_{k-2}(x_3,\ldots,x_k)-\sigma^2G_{\bbR^2}(x_1,x_2)\Psi_{k-2}^{\text{GFF}}(x_3,\ldots,x_k)+O_{x_2,\ldots,x_k}(1).
\end{equation*}
By induction, the first two terms cancel, so $F$ is actually bounded around $x_2$, and $x_2$ is therefore a removable singularity of $F$ (this is for example a simple consequence of the fact that Brownian motion started from a point $y\ne x_2$ almost surely does not pass through $x_2$). Therefore, $F$ can be extended to a harmonic function on all of $\calD$. Furthermore, as $x_1$ converges to any boundary arc $(a^{(i)}a^{(i+1)})$, the function $F$ converges to zero, and as $x_1$ converges to any marked point $a^{(i)}$, the function $F$ stays bounded (one can for example use Lemma \ref{lem:htilde}) \eav{Add reference here}. This forces the function $F$ to be identically zero on all of $\calD$, which concludes the induction step.
\end{proof}


\bigbreak\noindent{\bf AI use Statement.}
No mathematical content was generated or found with the assistance of artificial intelligence.

\bigbreak\noindent{\bf Acknowledgments.} The second author was funded by the ERC SPRaPaF. All but the fourth authors were partially funded by the Swiss FNS as well as the NCCR SwissMap.
The fourth author acknowledges support by the French National Research
Agency (ANR), project number ANR-23-CPJ1-0150-01.

\end{document}